\documentclass{article}

\usepackage{authblk}
\usepackage{fullpage}

\usepackage[T1]{fontenc}

\usepackage{nicefrac}
\usepackage{comment}

\usepackage{amsmath,amsfonts,amssymb}
\usepackage{graphicx}
\usepackage{wrapfig}
\usepackage{sidecap}
\usepackage{bm}
\usepackage{tikz}
\usepackage{xspace}
\usepackage{xcolor}
\usepackage{paralist}
\usepackage{enumitem}
\usepackage{algorithm}
\usepackage{subcaption}
\setlist[itemize]{leftmargin=*,itemsep=1pt,parsep=0.5pt}
\setlist[enumerate]{leftmargin=*,itemsep=1pt,parsep=0.5pt}
\setlist[description]{itemsep=1pt,parsep=0.5pt}

\usepackage[bookmarks=true,bookmarksnumbered,bookmarksdepth=2,pdfstartview=FitH,colorlinks,linkcolor=blue,filecolor=blue,citecolor=blue,urlcolor=blue]{hyperref}

\usepackage{algorithmicx}
\usepackage{algpseudocode}

\usepackage{amsthm}

\newtheorem{thm}{Theorem}
\newtheorem*{thm*}{Theorem}
\newtheorem{lem}{Lemma}
\newtheorem*{lem*}{Lemma}
\newtheorem{claim}{Claim}
\newtheorem*{claim*}{Claim}
\newtheorem{corol}{Corollary}
\newtheorem*{corol*}{Corollary}
\newtheorem{remark}{Remark}

\theoremstyle{definition}
\newtheorem{defn}{Definition}

\newtheorem*{mech*}{Construction}

\newtheorem*{notation*}{Notation}
\newtheorem*{remark*}{Remark}
 
\renewcommand{\paragraph}[1]{\smallskip\noindent{\bf #1}~}
\renewcommand{\subparagraph}[1]{\smallskip\noindent{\textit{#1}}~}

\renewenvironment{proof}{\noindent{{\em Proof:~}}}{\hfill $\Box$ \smallbreak}

\usepackage{environ}
\NewEnviron{Boxed}[1][\linewidth]{\medskip\noindent\makebox[\textwidth][c]{\framebox{\begin{minipage}{#1}\BODY\end{minipage}}}\medskip}{}

\newcommand{\namedref}[2]{\hyperref[#2]{#1~\ref*{#2}}}

\newcommand{\Algorithmref}[1]{\namedref{Algorithm}{algo:#1}}
\newcommand{\Sectionref}[1]{\namedref{Section}{sec:#1}}
\newcommand{\Appendixref}[1]{\namedref{Appendix}{app:#1}}
\newcommand{\Theoremref}[1]{\namedref{Theorem}{thm:#1}}
\newcommand{\Remarkref}[1]{\namedref{Remark}{remark:#1}}
\newcommand{\Corollaryref}[1]{\namedref{Corollary}{corol:#1}}

\newcommand{\Definitionref}[1]{\namedref{Definition}{def:#1}}
\newcommand{\Lemmaref}[1]{\namedref{Lemma}{lem:#1}}
\newcommand{\Claimref}[1]{\namedref{Claim}{clm:#1}}

\newcommand{\Figureref}[1]{\namedref{Figure}{fig:#1}}
\newcommand{\Equationref}[1]{(\ref{eq:#1})}
\newcommand{\Footnoteref}[1]{\namedref{Footnote}{foot:#1}}

\renewcommand{\implies}{\ensuremath{\Rightarrow}\xspace}

\renewcommand{\iff}{\ensuremath{\Leftrightarrow}\xspace}

\DeclareMathOperator*{\E}{{\mathbb E}}

\newcommand{\N}{\ensuremath{{\mathbb N}}\xspace}
\newcommand{\R}{\ensuremath{{\mathbb R}}\xspace}
\newcommand{\Rplus}{\ensuremath{\R_{\ge0}}\xspace}

\newcommand{\eps}{\ensuremath{\epsilon}\xspace}

\newcommand{\cA}{\ensuremath{\mathcal{A}}\xspace}

\newcommand{\cH}{\ensuremath{\mathcal{H}}\xspace}

\newcommand{\cS}{\ensuremath{\mathcal{S}}\xspace}

\newcommand{\A}{\ensuremath{\mathcal{A}}\xspace}
\newcommand{\B}{\ensuremath{\mathcal{B}}\xspace}
\newcommand{\C}{\ensuremath{\mathcal{C}}\xspace}
\newcommand{\G}{\ensuremath{\mathcal{G}}\xspace}
\newcommand{\X}{\ensuremath{\mathcal{X}}\xspace}
\newcommand{\Y}{\ensuremath{\mathcal{Y}}\xspace}

\newcommand{\M}{\ensuremath{\mathcal{M}}\xspace}

\newcommand{\Lap}{\ensuremath{\mathrm{Lap}}\xspace}

\newcommand{\x}{\ensuremath{\mathbf{x}}\xspace}
\newcommand{\y}{\ensuremath{\mathbf{y}}\xspace}
\newcommand{\z}{\ensuremath{\mathbf{z}}\xspace}
\newcommand{\s}{\ensuremath{\mathbf{s}}\xspace}
\newcommand{\dn}{\ensuremath{\mathsf{\partial}}\xspace}
\newcommand{\dnx}{\ensuremath{{\widehat\dn}}\xspace}
\newcommand{\met}{\ensuremath{\mathfrak{d}}\xspace}
\newcommand{\drop}{\ensuremath{\dn_{\mathrm{drop}}}\xspace}

\newcommand{\move}{\ensuremath{\dn_{\mathrm{move}}}\xspace}
\newcommand{\dropmove}[1]{\ensuremath{\dn_{\mathrm{drmv}}^{#1}}\xspace}
\newcommand{\drme}{\ensuremath{\dn_{\mathrm{drmv}}^{\eta}}\xspace}

\newcommand{\hist}{\ensuremath{\mathrm{hist}}\xspace}
\newcommand{\thresh}{\ensuremath{\tau}\xspace}
\newcommand{\Hspace}[1]{\ensuremath{\cH_{#1}}\xspace}

\newcommand{\prob}[1]{\ensuremath{\mathfrak{p}_{_{#1}}}\xspace}
\newcommand{\Prob}[2]{\ensuremath{{#1}(#2)}\xspace}

\newcommand{\W}[1]{\ensuremath{W_{{#1}}}\xspace}
\newcommand{\Winf}[1]{\ensuremath{W_{#1}^\infty}\xspace}
\newcommand{\Winfz}{\ensuremath{W^\infty}\xspace}

\newcommand{\dd}{\ensuremath{\,\mathrm{d}}\xspace}
\newcommand{\id}{\ensuremath{\mathrm{id}}\xspace}

\newcommand{\STLap}{\ensuremath{\text{Shifted-Truncated Laplace}}\xspace}
\newcommand{\epdel}{\ensuremath{(\epsilon, \delta)}}
\newcommand{\eeps}{\ensuremath{e^{\epsilon}}}

\newcommand{\eepsratio}[1]{ \left(\frac{e^{#1}-1}{\eeps-1}\right)}

\newcommand{\Lnoise}{\pi_q}
\newcommand{\lnoise}[1]{\Lnoise(#1)}
\newcommand{\Dnoise}{\pi_{q'}}
\newcommand{\dnoise}[1]{\Dnoise(#1)}

\newcommand{\mtrlap}[1]{\ensuremath{\M_{\mathrm{STLap}}^{#1}}\xspace}
\newcommand{\mbuc}[1]{\ensuremath{\M_{\mathrm{buc}}^{#1}}\xspace}
\newcommand{\mmax}[1]{\ensuremath{\M_{\mathrm{max}}^{#1}}\xspace}
\newcommand{\msupp}[1]{\ensuremath{\M_{\mathrm{supp}}^{#1}}\xspace}

\newcommand{\mBhist}[1]{\ensuremath{\M_{\mathrm{BucHist}}^{#1}}\xspace}

\newcommand{\fsupp}{\ensuremath{f_{\mathrm{supp}}}\xspace}
\newcommand{\fmin}{\ensuremath{f_{\mathrm{min}}}\xspace}
\newcommand{\fmax}{\ensuremath{f_{\mathrm{max}}}\xspace}
\newcommand{\dhist}[2]{\ensuremath{\met_{\mathrm{hist}}(#1, #2)}\xspace}
\newcommand{\dhistx}{\ensuremath{\met_{\mathrm{hist}}}\xspace}

\newcommand{\dsupp}[2]{\ensuremath{\met_{\mathrm{supp}}(#1, #2)}\xspace}
\newcommand{\dsuppx}{\ensuremath{\met_{\mathrm{supp}}}\xspace}
\newcommand{\dA}[2]{\ensuremath{\met_{\mathrm{\A}}(#1, #2)}\xspace}
\newcommand{\dAx}{\ensuremath{\met_{\mathrm{\A}}}\xspace}
\newcommand{\dB}[2]{\ensuremath{\met_{\mathrm{\B}}(#1, #2)}\xspace}
\newcommand{\dBx}{\ensuremath{\met_{\mathrm{\B}}}\xspace}
\newcommand{\dC}[2]{\ensuremath{\met_{\mathrm{\C}}(#1, #2)}\xspace}

\newcommand{\dG}[2]{\ensuremath{\met_{\mathrm{\G}}(#1, #2)}\xspace}
\newcommand{\nhist}{\ensuremath{\sim_{\mathrm{hist}}\xspace}}

\newcommand{\distsens}[2]{\ensuremath{\sigma_{#1}^{#2}}\xspace}

\newcommand{\support}{\ensuremath{\mathrm{support}}\xspace}

\newcommand{\fhbs}{\ensuremath{{f_{\mathrm{HBS}}}}\xspace}

\newcommand{\err}{\ensuremath{\mathrm{err}}\xspace}
\newcommand{\mode}{\ensuremath{\mathrm{mode}}\xspace}

\begin{document}

\title{Flexible Accuracy for Differential Privacy}

\author[1]{Aman Bansal}
\author[1]{Rahul Chunduru}
\author[2]{Deepesh Data}
\author[1]{Manoj Prabhakaran}
\affil[1]{Indian Institute of Technology Bombay, India \authorcr \texttt{aman0456b@gmail.com},\ \texttt{\{chrahul,mp\}@cse.iitb.ac.in}\vspace{0.25cm}}
\affil[2]{University of California, Los Angeles, USA \authorcr \texttt{deepesh.data@gmail.com}}

\date{}
\maketitle
{\allowdisplaybreaks

\begin{abstract}
Differential Privacy (DP) has become a gold standard in privacy-preserving data analysis.
While it provides one of the most rigorous notions of privacy, there are many settings where its applicability is limited.

Our main contribution is in augmenting differential privacy with {\em Flexible Accuracy}, which 
allows small distortions in the input (e.g., dropping
outliers) before measuring accuracy of the output, allowing one to extend DP
mechanisms to high-sensitivity functions. 
We present mechanisms that can help in achieving this notion for functions that had no 
meaningful differentially private mechanisms previously. In
particular, we illustrate an application to differentially private histograms,
which in turn yields mechanisms for revealing the support of a dataset or the
extremal values in the data. Analyses of our constructions exploit new
versatile composition theorems that facilitate modular design. 

All the above extensions use our new definitional framework, which is in terms of 
``lossy Wasserstein distance'' -- a 2-parameter error measure for distributions. This may be of independent interest.
\end{abstract}

\section{Introduction}
\label{sec:intro}

In the era of big data, privacy has been a major concern, to the point that
recent legislative moves, like General Data Protection Regulation (GDPR) in
the European Union, have mandated various measures for ensuring privacy. Further, in
the face of a global pandemic that has prompted governments to collect and
share individual-level information for epidemiological purposes, debates on
privacy-utility trade-offs have been brought to sharper relief. 
Against this backdrop, mathematical theories of privacy are of great
importance. 
Differential Privacy \cite{DworkMNS06} is by far the most impactful
mathematical framework today for privacy in statistical databases.  It has
seen large scale adoption in theory and practice, including machine learning
applications and large scale commercial implementations (e.g.,
\cite{AbadiCGMMTZ16,BorgsCS15,BorgsCSZ18,Rappor14,AppleDP17}).  

In this work, we make foundational contributions to the area of Differential
Privacy (DP), extending its applicability. Our main contribution is the notion of \emph{Flexible Accuracy} --
a new framework for measuring the \emph{accuracy} of a mechanism (while
retaining the DP framework unaltered for quantifying privacy). This lets us
develop new DP mechanisms with non-trivial provable (and empirically
demonstrable) accuracy guarantees in settings involving high-sensitivity functions.

\paragraph{Motivating Flexible Accuracy (FA).}
Consider querying a database consisting of integer valued observations --
say, ages of patients who recovered from a certain disease -- for the
maximum value. For the sake of privacy, one may wish to apply a DP
mechanism, rather than output the maximum in the data itself.  Two possible
datasets which differ in only one patient are considered neighbors and a DP
mechanism needs to make the outputs on these two samples indistinguishable
from each other. However, the function in question is \emph{highly
sensitive} -- two neighboring datasets can have their maxima differ by as
much as the entire range of possible ages%
\footnote{In fact, \emph{all datasets} with low maximum values have high
sensitivity \emph{locally}, by considering a neighboring dataset with a
single additional data item with a large value.}
-- and, as we shall see in our empirical evaluations in \Sectionref{eval}, the various kinds of mechanisms in the literature~\cite{ExponentialMech,BNS,VadhanSurvey,DworkL09,NissimRS07,BeimelNiSt16} do not provide a satisfactory solution.

The difficulty in solving this problem is related to another issue.
Consider the problem of reporting a \emph{histogram} (again,
say, of patients' ages). Here a standard DP mechanism, of adding a zero-mean
Laplace noise to each bar of the histogram is indeed reasonable, as the
histogram function has low sensitivity in each bar. Now, note that
\emph{maximum can be computed as a function of the histogram}. However, even
though the histogram mechanism was sufficiently accurate in the standard
sense, the maximum computed from its output is no longer accurate! This is
because when a non-zero count is added to a large-valued item which
originally has a count of 0, the maximum can increase arbitrarily.

Flexible Accuracy (FA) is a relaxed notion of accuracy that lets us address both
of the above issues. In particular, it not only enables new DP mechanisms
for maximum, but also allows one to derive the mechanism from a new DP
mechanism for histograms. We provide a general \emph{composition theorem} that
enables such transfer of accuracy guarantees that is not applicable to
conventional accuracy measures.

The high-level idea of Flexible Accuracy is to allow for some
\emph{distortion of the input} when measuring accuracy.  We shall require
distortion to be defined using a \emph{quasi-metric} over the input space (a
quasi-metric is akin to a metric, but is not required to be symmetric). A
good example of distortion is \emph{dropping a few items} from the dataset;
note that in this case, \emph{adding} a data item is \emph{not} considered low distortion.
Referring back to the example of reporting maximum, given a dataset with a
single elderly patient and many young patients, flexible accuracy with
respect to this distortion allows a mechanism for maximum to report the
maximum age of the younger group.%
\footnote{Of course, it is not obvious what should determine which items
should be dropped and with what probability.  This will be the subject of
our new mechanisms.}

Flexible accuracy needs to account for errors that can be attributed to
distortion of the input (input error), as well as to inaccuracies in the
output (output error).  To be able to exploit input distortion while retaining privacy, we allow
input distortion to be randomized.  A side-effect of this is that our
measure of output accuracy needs to allow the ``correct output'' to be
randomized (i.e., defined by a distribution), even if we are interested in
only deterministic functions. To generalize the conventional
\emph{probabilistically approximately correct} (PAC) guarantees to this
setting, we introduce a natural, but new quantity called \emph{lossy
$\infty$-Wasserstein distance}.
Our final definition of flexible accuracy is a 3-parameter quantity, with
one parameter accounting for input distortion, and 2 parameters used for 
output error measured using lossy $\infty$-Wasserstein distance.

\subsection{Our Contributions}
Our contributions are in three parts:

\begin{itemize}
\item \textit{Definitions:} We present a conceptual
	enhancement to the framework of DP -- \emph{flexible accuracy} -- which considers error after allowing for a small \emph{distortion of the input}; see \Definitionref{alpha-beta-gamma-accu}. 
	To account for randomized distortion (and more generally, to be able to
consider distributions over inputs and/or randomized functions) we
need an error measure that compares a mechanism's output distribution to not
a fixed ``correct value,'' but a ``correct distribution.'' For this, we
introduce and use a new measure called \emph{lossy $\infty$-Wasserstein distance} (see \Definitionref{infty-delta-wass-dist}), 
extending the classical notion of Wasserstein distance (or Earth Mover
Distance). This also generalizes several existing notions, such as the PAC guarantee, the notion of total variation distance, etc.

\item \textit{Composition Theorems:} 
We present a composition theorem for flexible accuracy (see \Theoremref{compose-accuracy}),
which gives an FA guarantee for a composed mechanism from those of the constituent ones. This involves identifying new quantities including \emph{distortion sensitivity} (see \Definitionref{dist-sens}) 
and \emph{error sensitivity} (see \Definitionref{err-sens}). 
To be able to use such composed mechanisms for DP, we rely on the well-known post-processing theorem of DP,
as well as a new pre-processing theorem (see \Theoremref{compose-DP}).

\item \textit{Mechanisms:} We give a DP mechanism with FA guarantee for releasing a sanitized histogram (called the \STLap mechanism; see \Algorithmref{hist-mech} and \Algorithmref{histogram-mech}), 
which, via our composition theorems, yield DP mechanisms with FA
	guarantees for \emph{histogram-based statistics} (see
	\Theoremref{bucketing-general}).
	These functions include several high-sensitivity functions, such as maximum
	and minimum, support of a set, range, median, maximum margin separator, etc.\ (we give concrete bounds for max/min and support).
	We present an
	empirical comparison against state-of-the-art DP mechanisms, which
	reveals that apart from the theoretical guarantees we obtain (where
	none were available till now), our mechanisms compare favorably with
	the others in terms of accuracy (flexible and otherwise) empirically
	as well.
\end{itemize}

\subsection{The Surprising Power of Flexible Accuracy} 
Consider a sequence of $n+1$ neighboring histograms, such that the first in the sequence has all
its $n$ elements in the first bar, and the last one has all elements in the last
bar, and the first and the last bars are far away from each other. 
In any reasonably accurate (flexible or not) mechanism for a histogram-based statistic like max, the answers for these two extremes must be very different with
probability almost 1. So, intuitively, there should be some pair of neighbors in this
sequence for which the
answers should be significantly different with probability at least $1/n$.
This seems to preclude obtaining $(\eps,\delta)$-DP for a small constant
$\eps$ with $\delta \ll 1/n$. Remarkably, this intuition turns out to be
wrong! By carefully calibrating the probability of the responses (while
also making sure that the responses can be attributed to only dropping a few
items -- as permitted by flexible accuracy), our mechanism 
can obtain the following guarantee for the max function (see \Corollaryref{bucketing-max}): \\

\noindent \fbox{\parbox{\linewidth}{{\bf Informal result for max:}\label{informal-result-max} Our flexibly-accurate mechanism for max over a bounded range achieves $\left(\eps,\eps e^{-\Omega(\eps\alpha n)}\right)$-DP while incurring an arbitrarily small output error after dropping only $\alpha n$ elements.}} \\

The above result gives a trade-off between the privacy guarantee and number of elements dropped. For example: 
{\sf (i)} By choosing $\eps=\frac{1}{n^{1/4}}$ and $\alpha=\frac{1}{\sqrt{n}}$, our mechanism is $(\frac{1}{n^{1/4}},e^{-\Omega(n^{1/4})})$-DP while dropping only $O(\sqrt{n})$ elements. 
{\sf (ii)} By choosing $\eps$ to be a small constant (say, $0.1$) and say, $\alpha=\frac{\log^2 n}{n}$, our mechanism is $(0.1,n^{-\Omega(\log n)})$-DP while dropping only $O(\log^2 n)$ elements.
See \Sectionref{choosing-params} for several other parameter choices that are of interest.

\paragraph{Significance of the New Mechanisms.} Traditional DP literature has largely
not addressed functions like the maximum function, \fmax. This is in part due to the very high
sensitivity of such functions: When the database has entries from $[0,B]$, 
the sensitivity of $\fmax$ is $B$.%
\footnote{The sensitivity of a real-valued function $f:\X\to\R$ is defined by $\Delta_f:=\max_{\x,\x'\in\X:\x\sim\x'}|f(\x)-f(\x')|$. In the case of \fmax, there are neighboring databases $\x,\x'$, where the first database $\x$ has all the inputs as $0$ and the second database has $n-1$ inputs as $0$ but one input is $B$, so, $\Delta_{\fmax}=B$.}
The same holds for other functions like a ``thresholded maximum'' $\max_k$
which outputs the maximum value that appears at least $k$ times in the
database. Despite being natural functions about the shape of the data, no DP
mechanisms have been offered in the literature for these functions.
With FA, \emph{for the first time, we provide DP mechanisms for such
functions, with meaningful worst-case accuracy guarantees}. We
emphasize that we retain the \emph{standard} definition of
$(\eps,\delta)$-DP, and achieve strong parameters for it (see above).
Further, the additional dimension of inaccuracy that we allow -- namely,
input distortion -- is in line with what applications like (robust) Machine
Learning often anticipate and tolerate. 

We also remark that, \emph{on specific data distributions}, some of the
existing DP mechanisms may already enable empirical FA guarantees (see
\Sectionref{eval} where such guarantees are compared). But crucially, 
such guarantees are not always available in the worst-case, and even on data
distributions where they do exist, they were not identified previously.

\subsection{Related Work and Paper Organization}
\paragraph{Related work.}
DP, defined by Dwork et al.~\cite{DworkMNS06} has developed into a highly
influential framework for providing formal privacy guarantees (see
\cite{DworkRo14} for more details).  The notion of flexible accuracy we
define is motivated by the difficulty in handling outliers in the data. Some
of the work leading to DP explicitly attempts to address the privacy of
outliers \cite{ChawlaDMSW05,ChawlaDMT05}, as did some of the later works
within the DP framework \cite{DworkL09,BNS,Stable1}. These results rely on having a distribution over the data, or
respond only when the answer is a ``stable value''. 
Blum et al.~\cite{BLR} introduced the notion of \emph{usefulness}, that is motivated by similar limitations of
DP as those which motivated flexible accuracy, but as explained later, is
less generally applicable.
Incidentally, Wasserstein distance has been used
in privacy mechanisms in the Pufferfish framework \cite{KiferM14,SongWC17},
but assuming a data distribution.

Several DP mechanisms for histograms are available with a variety of
accuracy guarantees, as discussed in \Sectionref{eval}. While these
mechanisms do not claim any accuracy guarantees for functions computed from
histograms, on specific data distributions and for some of these mechanisms,
we see that FA can be used to empirically capture meaningful accuracy guarantees.

\paragraph{Paper organization.} 
We define the lossy Wasserstein distance and its properties in \Sectionref{lossy-wass}. We define flexible accuracy in \Sectionref{flex-accu}, where we also give several examples of distortion measure.
In \Sectionref{composition-theorems}, we present our composition theorems for flexible accuracy and differential privacy. We also motivate and define distortion and error sensitivities (with examples) in \Sectionref{dist-sens} and \Sectionref{err-sens}, respectively. 
In \Sectionref{histogram_mechs}, we present our (bucketed)-histogram mechanism and state its flexible accuracy and privacy guarantees, and we post-process that mechanism by any histogram-bases-statistic in \Sectionref{HBS}. Results with distortions other than dropping elements are presented in \Sectionref{beyond-drop}. All the proofs are presented in \Sectionref{proofs}. We empirically evaluate our mechanisms with several other mechanisms from literature in \Sectionref{eval}. Omitted details are provided in appendices.

\section{Lossy Wasserstein Distance}\label{sec:lossy-wass}
Central to the formalization of all the results in this work is a new notion of distance between
distributions over a metric space, that we call \emph{lossy Wasserstein
distance}. Lossy Wasserstein distance generalizes the notion of Wasserstein
distance~\cite{Villani_OptimalTransport08}, or Earth Mover Distance,
which is the minimum cost of transporting probability mass (``earth'') of one
distribution to make it match the other. Loss refers to the fact that some
of the mass is allowed to be lost during this transportation. We shall
use the ``infinity
norm'' version, where the cost paid is the maximum distance any mass is
transported.

Formally, consider a metric space with ground set
$\Omega$, and metric \met, where Wasserstein distance can be defined. For
example, one may consider $\Omega=\R^n$ and the metric \met being an
$\ell_p$-metric.  
For $\gamma\in[0,1]$, and distributions $P,Q$ over the metric space
$(\Omega,\met)$,\footnote{We will use upper case letters ($P,Q,X,Y$, etc.) to denote random variables (r.v.), as well as the probability distributions associated with them. Sometimes, we will also denote the probability distribution associated with a r.v.\ $X$ by $\prob{X}$.} we define $\Phi^{\gamma}(P,Q)$, the set of \emph{$\gamma$-lossy couplings of $P$
and $Q$}, as consisting of joint distributions $\phi$
over $\Omega^2$ with marginals $\phi_1$ and $\phi_2$ such that
$\Delta(\phi_1,P) + \Delta(\phi_2,Q)\leq\gamma$, where 
$\Delta(P,Q) := \frac{1}{2}\int_{\Omega}|\Prob{P}\omega-\Prob{Q}\omega|\dd \omega$ denotes the total variation distance between $P$ and $Q$.
Note that $\Phi^0(P,Q)$ consists of joint distributions with marginals exactly equal to $P$ and $Q$.

\begin{defn}[$\gamma$-Lossy $\infty$-Wasserstein Distance]\label{def:infty-delta-wass-dist}
Let $P$ and $Q$ be two distributions  over a metric space
$(\Omega,\met)$.
For $\gamma\in[0,1]$, the $\gamma$-lossy $\infty$-Wasserstein distance between $P$ and $Q$ is defined as:
\begin{equation}\label{eq:infty-delta-wass-dist}
\Winf{\gamma}(P,Q) = \inf_{\phi\in\Phi^{\gamma}(P,Q)} \sup_{(x,y)\leftarrow\phi}\met(x,y).
\end{equation}
\end{defn}
For simplicity, we write $\Winfz(p,q)$ to denote $\Winf0(p,q)$. 
We remark that while our definition of \Winf\gamma uses a worst case notion of distance
(as signified by $\infty$), there is an analogous average case definition,
that may be of independent interest. We define this in
\Appendixref{average_lossy-wass}.

\subsection{Lossy $\infty$-Wasserstein Distance Generalizes Some Existing Notions}
Now we show that the Lossy $\infty$-Wasserstein distance generalizes the guarantee of
being ``Probably Approximately Correct'' (PAC) and also the definition of total variation distance, as shown below.
\begin{itemize}
\item {\bf Generalizing the PAC guarantee:} The PAC guarantee states that
a randomized quantity $G$ is, except with some small probability $\gamma$,
within an approximation radius $\beta$ of a desired \emph{deterministic}
quantity $f$: i.e., $\Pr_{g\leftarrow G}[\met(f,g) > \beta] \le \gamma$. For example, when $G$ takes values over $\R$, $\met$ can be the standard different metric over $\R$, i.e., $\met(f,g)=|f-g|$.
Representing $f$ by a point distribution $F_f$, this can
be equivalently written as $\Winf{\gamma}(F_f,G) \le \beta$, where the underlying metric is $\met$;
see \Lemmaref{generalize-pac} in \Appendixref{wass-generalizes} for a proof of this.

\item {\bf Generalizing the total variation distance:} It also
generalizes the total variation distance $\Delta(P,Q)$ between two distributions, since
$\Winf\gamma(P,Q)=0$ iff $\Delta(P,Q) \le \gamma$; see \Lemmaref{generalizing-tv} in \Appendixref{wass-generalizes} for a proof of this.
\end{itemize}

\subsection{Triangle Inequality for Lossy Wasserstein Distance}
The Lossy $\infty$-Wasserstein distance satisfies the following triangle inequality.
\begin{lem}\label{lem:wass-triangle}
For distributions $P$, $Q$, and $R$ over a metric space $(\Omega,\met)$ and for all
$\gamma_1,\gamma_2 \in [0,1]$, we have
\begin{align}
\Winf{\gamma_1 + \gamma_2}(P, R) &\leq \Winf{\gamma_1}(P, Q) + \Winf{\gamma_2}(Q, R). \label{eq:wass-triangle-gamma-infty}
\end{align}
\end{lem}
We can easily prove \Lemmaref{wass-triangle} for the special case when $\gamma_1=\gamma_2=0$ using standard tools from \cite{Villani_OptimalTransport08}; see \Lemmaref{Winf_triangle} in \Appendixref{wasserstein} for a proof. However, proving \Lemmaref{wass-triangle} in its full generality requires a significantly more involved proof, which we present in \Sectionref{triangle-ineq_Wass}.

\section{Flexible Accuracy}\label{sec:flex-accu}
The high-level idea of flexible accuracy is to
allow for some \emph{distortion of the input} before measuring accuracy.
We would like to define ``natural'' distortions of a database, that are meaningful for the
function in question. For many functions, removing a few data points (say, outliers) would be a natural distortion, while for others, perturbing the data points (or a combination of both) is more natural.  Note that \emph{adding} new entries -- even just one -- is often not a reasonable
distortion. Therefore, distortion is generally defined not using a metric over
databases, but a {\em quasi-metric} (which is  not required 
to be symmetric).\footnote{A function $\dn:\X\times\X\to[0,\infty)$ is called a quasi-metric, if for every $\x,\y,\z$, we have {\sf (i)} $\dn(\x,\y)=0\iff\x=\y$ and {\sf (ii)} $\dn(\x,\y)\leq\dn(\x,\z)+\dn(\z,\y)$.} 

\subsection{Measure of Distortion}\label{sec:distortion-measures}
We shall use quasi-metrics with range $\Rplus\cup\{\infty\}$ to define a measure of distortion,
where $\infty$ indicates that one database cannot be distorted into another one.
As we shall need distortion measure between two distributions in our accuracy guarantees and also in the definitions of distortion and error sensitivities, it will be useful to extend the distortion measure to distributions. This can be done in same way as $\Winfz$, but with respect to a quasi-metric rather than a metric.
\begin{defn}[Measure of Distortion]\label{def:distortion}
A \emph{measure of distortion} on a set $\X$ is a function $\dn:\X\times \X
\rightarrow \Rplus \cup \{\infty\}$ which forms a quasi-metric over $\X$. 
We also define \dnx as the extension to $\dn$ to distribution, which maps a pair
of distributions $P,Q$ over \X to a real number as
    \[
       \dnx(P, Q) := \inf_{\phi\in\Phi^0(P, Q)} \sup_{(x,y)\leftarrow\phi}\dn(x,y).
    \]
If $P$ is a point distribution with all its mass on a point $x$, we denote
$\dnx(P,Q)$ as $\dnx(x,Q)$, which can be simplified as $\dnx(x,Q)=\sup_{x'\in\support(Q)}\dn(x,x')$.
Furthermore, if both $P$ and $Q$ are point distributions on $x$ and $y$, respectively, then $\dnx(P,Q)=\dn(x,y)$, and we will write $\dnx(P,Q)$ simply by $\dn(x,y)$.
\end{defn}
It is easy to verify that if $\dn$ is a quasi-metric, so is $\dnx$. We prove this in \Lemmaref{dnx-quasi-metric} in \Appendixref{beyond-drop}.

\paragraph{Examples of measures of distortion.}
We formally define three measures of distortion: $\drop$ for dropping elements, $\move$ for perturbing/moving elements, and $\dropmove\eta$ for a combination of dropping and moving elements.
These are defined when each element in $\X$ is a finite multiset over a ground set \G. Formally, 
$\x\in\X$ is a function $\x:\G\rightarrow\N$ (where $\mathbb{N}$ denotes the set of all non-negative integers, including zero) that outputs the multiplicity
of each element of \G in $\x$. We denote the {\emph size} of $\x$ by $|\x| := \sum_{i\in\G} \x(i)$.

\begin{enumerate}
\item {\bf Dropping elements:} 
For finite $\x,\x' \in \X$, we define $\drop$, a measure of distortion for dropping elements, as follows:
\begin{equation}\label{eq:drop_defn}
\drop(\x,\x') := 
\begin{cases}
\frac{\sum_{g\in G} \x(g)-\x'(g)}{\sum_{g\in G} \x(g)} & \text{ if } \forall g\in\G, \x(g) \ge \x'(g), \\
\infty & \text{ otherwise.}
\end{cases}
\end{equation}
That is, $\drop(\x,\x')$ measures the fraction of elements in \x that are to
be dropped for it to become $\x'$ (unless $\x'$ cannot be derived thus).
It is easy to see that $\drop$ is a quasi-metric.

\item {\bf Perturbing/Moving elements:} 
For finite $\x,\y \in \X$, we define $\move$, a measure of distortion for moving elements, as follows:
\begin{equation}\label{eq:move_defn}
\move(\x,\y) = 
\begin{cases}
\Winf{}(\frac{\x}{|\x|},\frac{\y}{|\y|}) & \text{ if } |\x|=|\y|, \\
\infty & \text{ otherwise},
\end{cases}
\end{equation}
where $\frac{\x}{|\x|}$ (similarly, $\frac{\y}{|\y|}$) is treated as a probability vector of size $|\G|$, indexed by the elements of $\G$; the $i$'th element of $\frac{\x}{|\x|}$ is equal to $\frac{\x(i)}{|\x|}$.
We show in that \Claimref{move-metric} in \Appendixref{beyond-drop} that $\move$ is a metric.

\item {\bf Both dropping and moving elements:} For finite $\x,\y \in \X$, we define $\dropmove\eta$, a measure of distortion for both moving and dropping elements, as follows:
\begin{equation}\label{eq:drop_move_defn}
\dropmove\eta(\x,\y) = \inf_{\z}  \left(\drop(\x,\z) + \eta \cdot \move(\z,\y)\right).
\end{equation}
We show in  \Claimref{drop-move-quasi-metric} in \Appendixref{beyond-drop} that $\dropmove\eta$ is a quasi-metric.\footnote{While showing that $\drop$ is a quasi-metric is trivial, it is not always so with other measures of distortion; in particular, showing that $\dropmove\eta$ is a quasi-metric is non-trivial.} 
\end{enumerate}
Most of the results in this paper are derived w.r.t.\ the distortion $\drop$, but they can also be extended to the distortion $\dropmove\eta$; see \Sectionref{beyond-drop} for the extension.

\subsection{Defining Flexible Accuracy}\label{sec:defn-fa}
Informally, flexible accuracy with a distortion bound $\alpha$ guarantees
that on an input \x, a mechanism shall produce an output that corresponds to
$f(\x')$ for some $\x'$ such that $\dn(\x,\x')\leq \alpha$. In addition to
such input distortion, we may allow the output to be also probably
approximately correct, with an approximation error parameter $\beta$ and an
error probability parameter $\gamma$. Formally, the probabilistic
approximation guarantee of the output is given as a bound of $\beta$ on a
$\gamma$-lossy $\infty$-Wasserstein distance.

\begin{defn}[$(\alpha,\beta,\gamma)$-accuracy]\label{def:alpha-beta-gamma-accu}
Let \dn be a measure of distortion on a set $\X$ and  $f:\X\to\Y$ be a randomized function 
such that \Y admits a metric. A mechanism $\M$ is said to
be \emph{$(\alpha,\beta,\gamma)$-accurate for $f$ with respect to $\dn$}, if
\begin{equation}\label{eq:fa-defn}
\left(\sup_{\x\in \X}\inf_{X':\dnx(\x,\prob{X'})\leq\alpha}\Winf{\gamma}(\M(\x),f(X'))\right) \leq \beta.
\end{equation}
In other words, for each $x\in\X$, there is a random variable $X'$ satisfying $\dnx(\x,\prob{X'})\leq\alpha$ (i.e., $\dn(\x,\x')\leq\alpha$ for all $\x'\in\support(X')$) such that $\Winf{\gamma}(\M(\x),f(X'))\leq\beta$.
\end{defn}

See \Figureref{fa} on page~\pageref{fig:fa} for an illustration of flexible accuracy using a pebbling game.

\paragraph{Flexible accuracy generalizes existing accuracy definitions.}
It should be noted that flexible accuracy is not a completely disparate notion but a more generalized form of the standard accuracy guarantees. In particular: 

\begin{itemize}
\item As mentioned in \Sectionref{lossy-wass}, $(0,\beta,\gamma)$-accuracy already extends the PAC guarantees. For example, the Laplace mechanism (see \cite[Chapter 3]{DworkRo14}) for a function $f:\X\to\R$ that achieves $\epsilon$-DP is $(0,\frac{\nabla_f}{\eps}\ln(1/\gamma),\gamma)$-accurate for any $\gamma>0$, where $\nabla_f$ is the sensitivity of $f$.
\item Blum et al.~\cite{BLR} introduced \emph{usefulness} to
measure accuracy with respect to a ``perturbed'' function.  While adequate
for the function classes they considered (half-space queries, range queries
etc.), it is not applicable to queries like maximum. Flexible accuracy 
generalizes usefulness (see \Appendixref{Comparison_BLR}).

\end{itemize}

As we show later, flexible accuracy lets us develop DP mechanisms for highly
sensitive functions ({\em e.g.}, $\max$), for which existing DP mechanisms offered only limited,
if not vacuous, guarantees.

\section{Composition Theorems}\label{sec:composition-theorems} 

It is often convenient to design a mechanism as the function composition
of two mechanisms, $\M = \M_2 \circ \M_1$. We present ``composition theorems''
which yield flexible accuracy and differential privacy guarantees for \M in terms of those for $\M_1$ and $\M_2$.

\subsection{Flexible Accuracy Under Composition}\label{sec:comp-accuracy}

In order to give our composition theorem for flexible accuracy, we need to define two new sensitivity notions: {\em distortion sensitivity} for a function and {\em error sensitivity} for a mechanism. 
We give motivation behind each of these sensitivity notions when defining them in their respective subsections below.

\subsubsection{Distortion Sensitivity}\label{sec:dist-sens}
When we compose two flexibly accurate mechanisms $M_1$ and $M_2$ for $f_1:A\to B$ and $f_2:B\to C$, respectively, to obtain the flexible accuracy guarantee of $M_2\circ M_1$ for $f_2\circ f_1:A\to C$, we would like to attribute all the distortion made in $A$ and $B$ (for measuring the output error of $M_1$ and $M_2$, respectively) to the distortion in $A$. This requires transferring the input distortion from $B$ back into $A$, and the notion of distortion sensitivity allows us to quantify this. 
Informally, distortion sensitivity of a function $f$ (denoted by $\distsens{f}{}$) captures the amount of distortion required in the domain of $f$ to capture a certain amount of distortion in the codomain of $f$. We formalize this intuition below.
\begin{defn}[Distortion sensitivity]\label{def:dist-sens}
Let $f: A \to B$ be a randomized function where $B$ admits Wasserstein distances.
Let $\dn_1,\dn_2$ be measures of distortion on $A, B$, respectively. 
Then, the \emph{distortion-sensitivity} of $f$ w.r.t.\ $(\dn_1,\dn_2)$ is defined as the function
$\distsens{f}{}:\Rplus \cup \{\infty\}\to\Rplus \cup \{\infty\}$ given by 
\begin{equation}\label{eq:dist-sens}
    \distsens{f}{}(\alpha) =
    \sup_{\substack{x,Y:\\\dnx_2(f(x),\prob{Y}) \le \alpha}} \inf_{\substack{X:\\ f(X) = Y}} \dnx_1(x, \prob{X})
\end{equation}
where $x\in A$, and the random variables $X$ and $Y$ are distributed over $A$ and $B$,
respectively. 
Above, infimum over an empty set is defined to be $\infty$.
\end{defn}

See \Figureref{dist-sens} on page~\pageref{fig:dist-sens} for an illustration of distortion sensitivity using a pebbling game.

\paragraph{Distortion sensitivity at $\alpha=0$.} It is easy to verify that for any randomized function $f$, we have $\distsens{f}{}(0)=0$. We will use this for deriving flexible accuracy guarantees for any histogram-based-statistic in \Sectionref{HBS}.

\paragraph {Distortion sensitivity of deterministic bijective functions:} When $f:A\to B$ is a deterministic and bijective map, then for every $x,Y$ such that $\dnx_2(f(x),\prob{Y})\leq\alpha$, there is only one choice of $X$ for which $f(X)= Y$ holds, that is $X=f^{-1}(Y)$. 
Since for any point $x\in A$ and distribution $P$ over $A$, we have $\dnx_1(x,P)=\sup_{x'\in\support(P)}\dn_1(x,x')$, it follows that
\begin{align}\label{bijective-distsens}
\distsens{f}{}(\alpha)\quad = \sup_{\substack{x,Y:\\\dnx_2(f(x),\prob{Y}) \le \alpha}}\dnx_1(x,\prob{f^{-1}(Y)})
    \quad = \sup_{\substack{x\in A,y\in B:\\\dn_2(f(x),y) \le \alpha}} \dn_1(x, f^{-1}(y)).
\end{align}
In particular, if $f:A\to A$ is an identity function and $\dn_1=\dn_2$, then we we have $\distsens{f}{}(\alpha)\leq\alpha$. Many of our mechanisms in this paper for which we derive flexible accuracy guarantees are given for the identity function over the space of histograms; see, for example, the result for our basic histogram mechanism (\Theoremref{hist-priv-accu}), the bucking mechanism (see \Claimref{bucket-accuracy}), and their composition (\Theoremref{bucketing-hist}), etc.

\paragraph{Distortion sensitivity of the histogram function w.r.t.\ $(\drop,\drop)$:} 
Let $\G=\{g_1,g_2,\hdots,g_k\}$ denote a finite set.
The histogram function $f_{\hist}$ takes an unordered dataset $\x=\{x_1,\hdots,x_n\}$ (where each $x_i\in\G$ and the ordering of $x_i$'s does not matter) as input and outputs the histogram of the dataset, i.e., a $k$-tuple $(\x(g_1),\hdots,\x(g_k))$ where $\x(g_i):=|\{j:x_j=g_i\}|$ denotes the multiplicity of $g_i$ in $\x$. Note that $f_{\hist}$ is a deterministic bijective function. It is easy to verify (from \eqref{bijective-distsens}) that $\distsens{f_{\hist}}{}(\alpha)\leq\alpha$ w.r.t.\ $(\drop,\drop)$. Later on, it will be convenient for us to represent a histogram over a finite set $\G$ as a map $\x:\G\to\N$, that outputs the multiplicity of any element from $\G$ in the dataset. 

\begin{remark}
Though the distortion sensitivity is bounded in many circumstances (including all the applications we consider in this paper); however, due to the strict requirement of having an $X$ such that $f(X)=Y$ (under infimum) in its definition, it may be infinite in other situations where this condition cannot be satisfied. To accommodate more functions, we can relax the definition with more parameters $\theta\in[0,1]$, $\omega\geq0$ as follows: 
\[    \distsens{f}{\gamma,\omega}(\alpha) =
    \sup_{\substack{x,Y:\\\dnx_2(f(x),\prob{Y}) \le \alpha}} \inf_{\substack{X:\\ \Winf{\gamma}(f(X), \prob{Y})\leq\omega}} \dnx_1(x, \prob{X}).\]
All the results in this paper can be extended to work with this more general definition of distortion sensitivity.
\end{remark}

\begin{figure}[t]
\centering
\begin{subfigure}{0.29\textwidth}
\centering
\includegraphics[width=\linewidth]{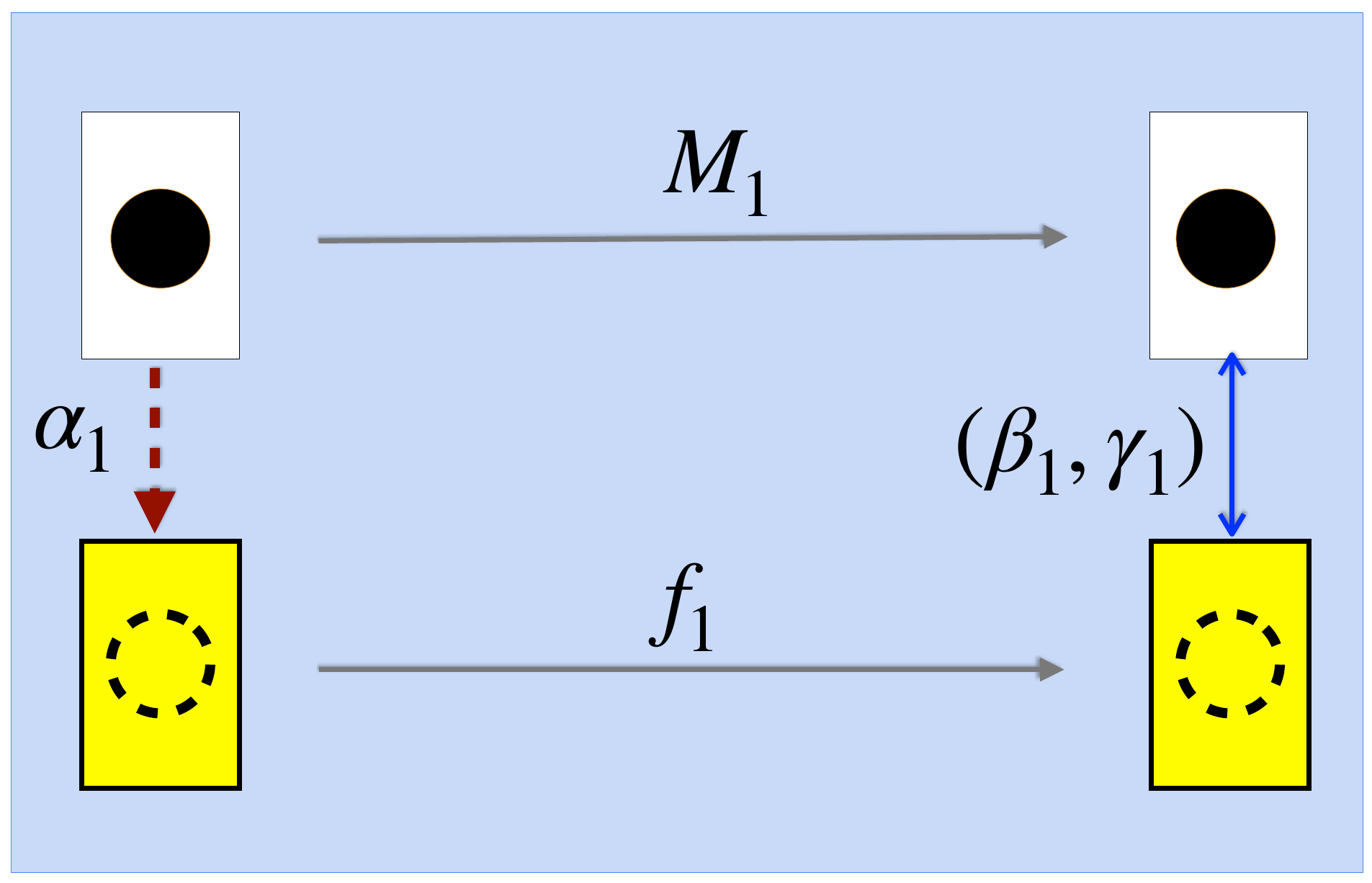}
\captionsetup{justification=centering}
\caption{Flexible Accuracy}
\label{fig:fa}
\end{subfigure}\hfill
\begin{subfigure}{0.29\textwidth}
\centering
\includegraphics[width=\linewidth]{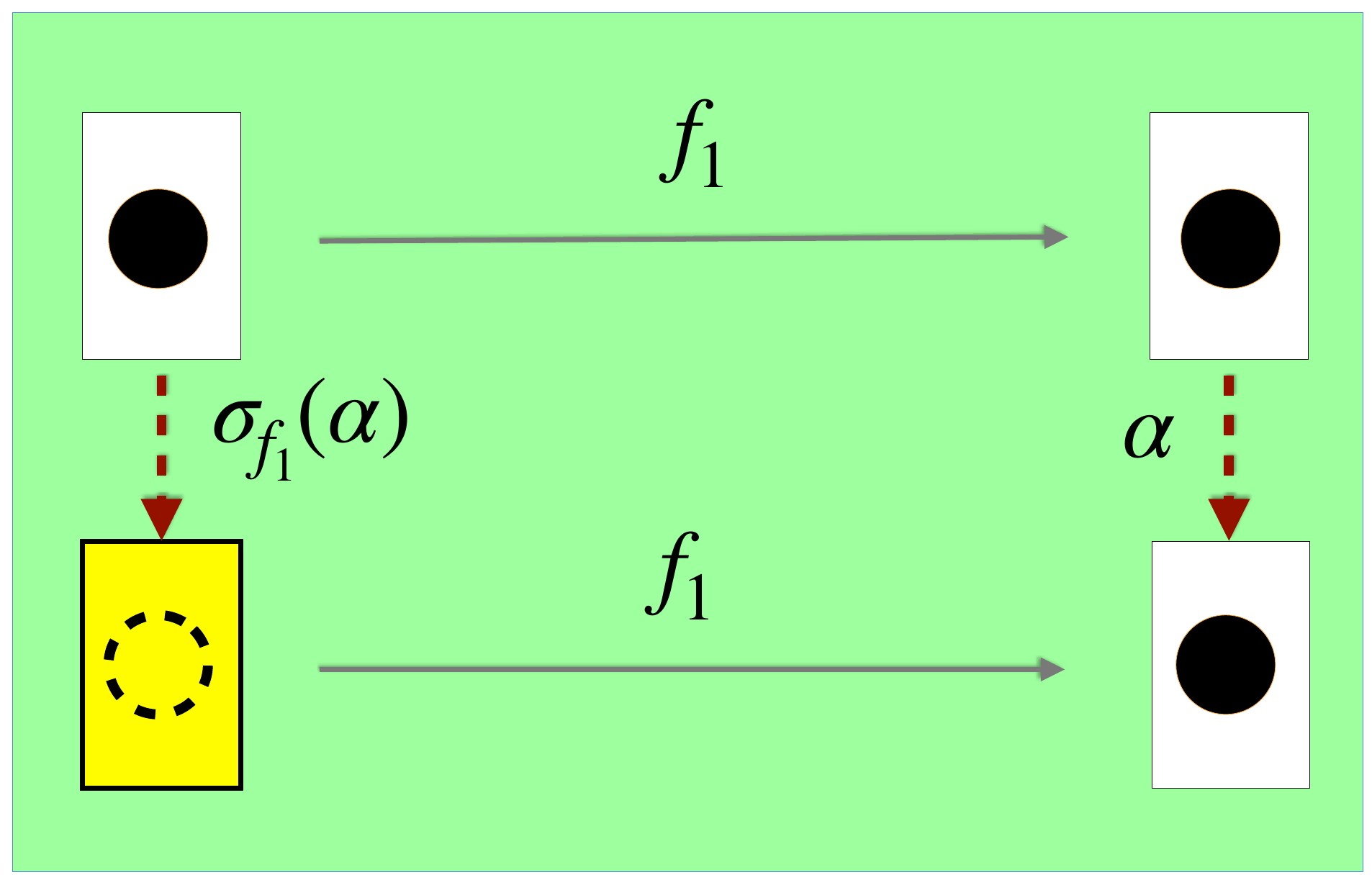}
\captionsetup{justification=centering}
\caption{Distortion Sensitivity}
\label{fig:dist-sens}
\end{subfigure}\hfill
\begin{subfigure}{0.40\textwidth}
\centering
\includegraphics[width=\linewidth]{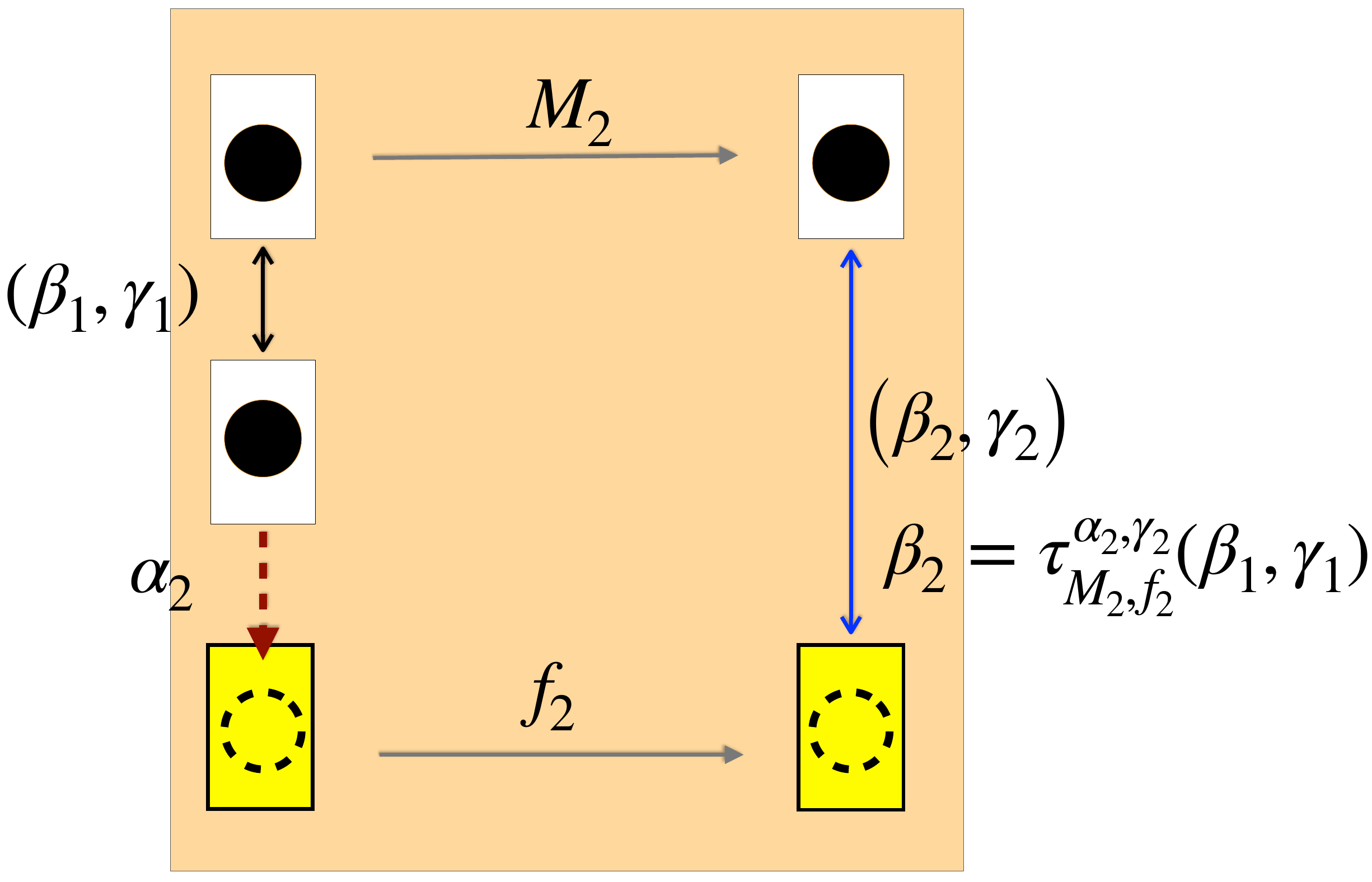}
\captionsetup{justification=centering}
\caption{Error Sensitivity}
\label{fig:err-sens}
\end{subfigure}
\caption{An illustration of the flexible accuracy, distortion sensitivity, and error sensitivity.
Dotted arrows
	indicate closeness in terms of distortion between histograms (or
	distributions thereof), and the solid two-sided arrows indicate closeness 
	in terms of the lossy Wasserstein distance. Each figure
	shows the corresponding guarantee (accuracy, error sensitivity or
	distortion sensitivity) as a pebbling game: The white boxes with
	black pebbles correspond to given histograms, and the yellow
	boxes indicate histograms that are guaranteed to exist, such that the
	given closeness relations hold. This allows those boxes to be pebbled.
	Accuracy guarantee of $M_2 \circ M_1$
	is derived by first applying the pebbling rule of accuracy of $M_1$
	(to obtain the purple pebbles), then that of the error sensitivity of
	$M_2$ (to get the pink pebbles) and finally using the pebbling rule of
	the distortion sensitivity of $f_1$ to pebble the remaining yellow box.}
\label{fig:fa-dist-err-sens}
\end{figure}

\begin{SCfigure}
  \centering
    \includegraphics[width=0.55\textwidth]{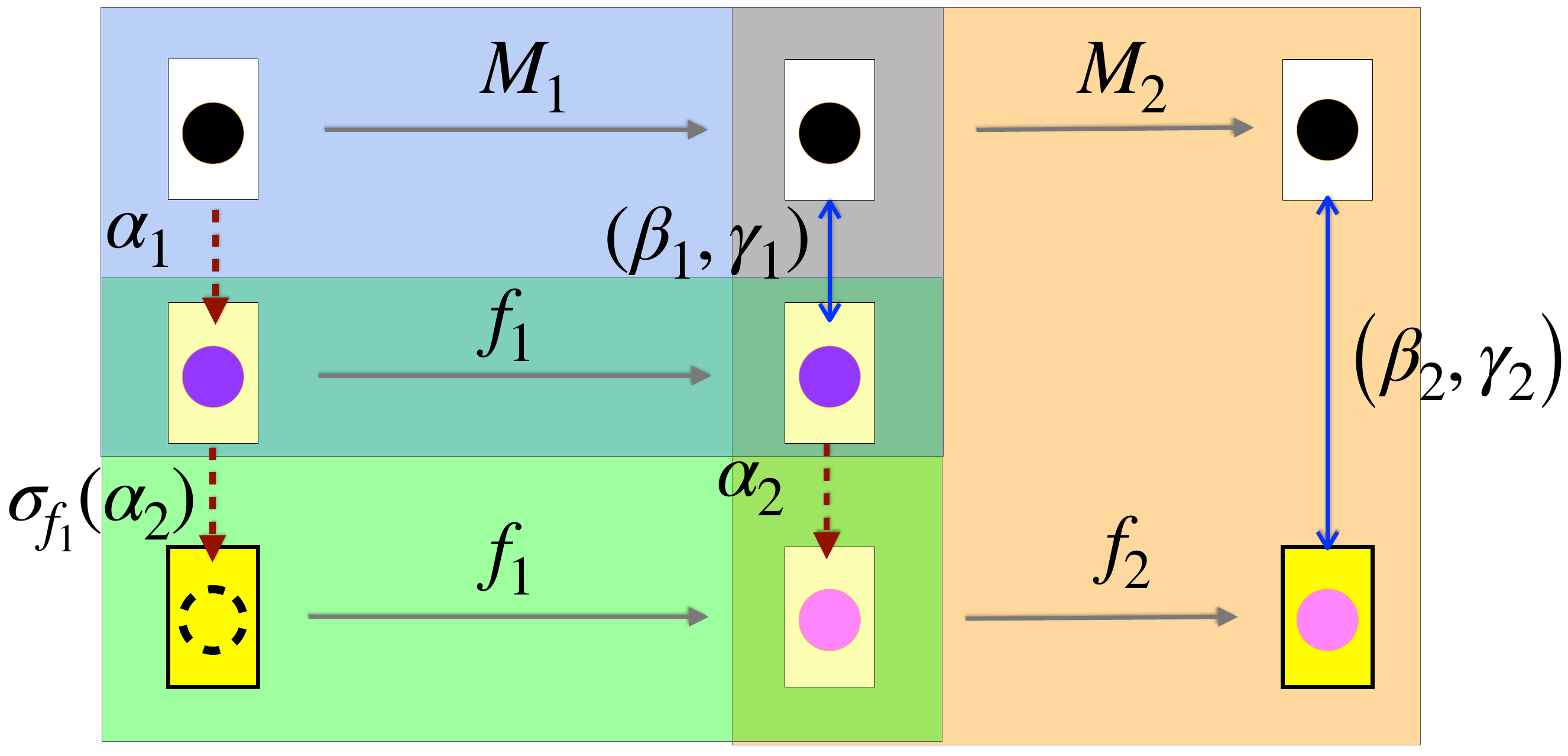}
  \caption{An illustration of the composition theorem, \Theoremref{compose-accuracy}. Accuracy guarantee of $M_2 \circ M_1$
	is derived by first applying the pebbling rule of accuracy of $M_1$
	(to obtain the purple pebbles), then that of the error sensitivity of
	$M_2$ (to get the pink pebbles), and finally using the pebbling rule of
	the distortion sensitivity of $f_1$ to pebble the remaining yellow box.
	The final parameters are $\alpha=\alpha_1+\distsens{f_1}(\alpha_2)$, $\beta=\tau_{M_2,f_2}^{\alpha_2,\gamma_2}(\beta_1,\gamma_1)$, and $\gamma=\gamma_2$.}
		\label{fig:composition}
\end{SCfigure}

\subsubsection{Error Sensitivity}\label{sec:err-sens}

Suppose we want to compose an $(\alpha_1,\beta_1,\gamma_1)$-accurate mechanism $M_1$ for $f_1:A\to B$ with another flexibly accurate mechanism $M_2$ for $f_2:B\to C$ to obtain flexible accuracy guarantee of the composed mechanism $M_2\circ M_1$ for $f_2\circ f_1:A\to C$. 
For this, on any input $x\in A$, first we measure the output error of $M_1$ on input $x$ in terms of $\Winf{\gamma_1}(M_1(x),f_1(X'))$, where $X'$ is an $\alpha_1$-distortion of the {\em same} $x$ on which we run the mechanism $M_1$; see \eqref{eq:fa-defn}. 
Now, for composition, we need to run $M_2$ on $M_1(x)$ and distort $f_1(X')$ to obtain another r.v.\ $Y$, and the output error of the composed mechanism is given by $\Winf{\gamma}(M_2(M_1(x)),f_2(Y))$. The problem here is that since the input (distribution) $f(X')$ that we distort is {\em not the same} as the input (distribution) $M_1(x)$ that we run $M_2$ on, we cannot directly obtain the output error guarantee of the composed mechanism from that of $M_2$. Therefore, we need a way to generalize the measure of accuracy (output error) of a flexible accurate mechanism when the input (distribution) to the mechanism is not the same as the input (distribution) that we distort, but they are at a bounded distance from each other (as measured in terms on the lossy $\infty$-Wasserstein distance). The notion of error sensitivity formalizes this intuition. Informally, it captures the sensitivity of the output error of a flexibly accurate mechanism in such situations.

\begin{defn}[Error sensitivity]\label{def:err-sens}
Let $\M: B \to C$ be any mechanism for a function $f:B\to C$,
where both $B$ and $C$ have associated Wasserstein distances.
Let $\dnx$ be a measure of distortion on $B$.  Then, for $\alpha_2,
\gamma_2\geq0$, the error-sensitivity $\tau_{\M,f}^{\alpha_2, \gamma_2}:\Rplus
\times [0,1] \to \Rplus$ of \M w.r.t.\ $f$ is defined as:
\begin{align}\label{eq:err-sens}
\tau_{\M,f}^{\alpha_2, \gamma_2}(\beta_1, \gamma_1) = \sup_{\substack{X,X': \\ \Winf{\gamma_1}(\prob{X},\prob{X'})\leq \beta_1}} \ \inf_{\substack{Y: \\ \dnx(X', Y) \le \alpha_2}} \Winf{\gamma_2}(\M(X),f(Y)).
\end{align}
\end{defn}
In other words, if $\tau_{\M,f}^{\alpha_2, \gamma_2}(\beta_1, \gamma_1) = \beta_2$, then for distributions $X,X'$ over $A$ such that $\Winf{\gamma_1}(\prob{X},\prob{X'})\leq \beta_1$, one can $\alpha_2$-distort $X'$ to $Y$ in such a way that $\Winf{\gamma_2}(\M(X),f(Y))\leq\beta_2$. 
See \Figureref{err-sens} on page~\pageref{fig:err-sens} for an illustration of err sensitivity using a pebbling game.
\begin{remark}
As mentioned earlier, the notion of error sensitivity generalizes the definition of flexible accuracy. In other words, if a mechanism $\M$ for computing a function $f$ is $(\alpha,\beta,\gamma)$-accurate, then $\beta=\tau_{\M,f}^{\alpha, \gamma}(0, 0)$. 
\end{remark}
We can simplify the expression of error sensitivity in some special cases that arise later on in \Sectionref{mechanisms}; we discuss these after stating our composition theorem for flexible accuracy in the next subsection.

\subsubsection{Composition Theorem for Flexible Accuracy}\label{sec:compostion_FA}
Having defined the distortion and error sensitivities, we shall now see how they play in a composition $M_2 \circ M_1$ for $f_2 \circ f_1$, where $M_1, M_2$ are mechanisms with flexible accuracy guarantees.
\begin{thm}[Flexible Accuracy Composition]\label{thm:compose-accuracy}
Let  $\M_1 : A \to B$  and $\M_2 : B \to C$ be mechanisms, respectively, with $(\alpha_1,\beta_1,\gamma_1)$-accuracy for $f_1: A \to B$ and $\tau_{\M_2,f_2}$ error sensitivity for $f_2: B \to C$, w.r.t.\ measures of distortion $\dn_1$, $\dn_2$ defined on $A, B$ and metrics $\met_1, \met_2$ defined on $B, C$, respectively.
Suppose $f_1,\alpha_2$ are such that $\distsens{f_1}{}(\alpha_2)$ is finite.
Then, for any $\alpha_2 \ge 0$ and $\gamma_2 \in [0,1]$, the mechanism $\M_2 \circ \M_1 : A \to C$ is $(\alpha, \beta ,\gamma)$-accurate for the function $f_2\circ f_1$ w.r.t.~$\dn_1$ and $\met_2$, where 
$\alpha = \alpha_1 + \distsens{f_1}{}(\alpha_2)$, $\beta = \tau_{\M_2,f_2}^{\alpha_2, \gamma_2}(\beta_1, \gamma_1)$, and $\gamma = \gamma_2$. 
\end{thm}
We prove \Theoremref{compose-accuracy} in \Sectionref{compose-accuracy_proof}.
An illustration of how the composition theorem works is given as a pebbling game in \Figureref{composition}.

\Theoremref{compose-accuracy} requires computing/bounding the error sensitivity of $\M_2$ in order to compute the flexible accuracy parameter $\beta$ of the composed mechanism $\M_2\circ\M_1$. 
Now we show that the expression of error sensitivity can be simplified in some important special cases.

\paragraph{$\bullet$ When $\M_1,f_1$ are deterministic maps and $\M_1$ is $(0,\beta_1,0)$-accurate.}
This setting arises when we compute the flexible accuracy parameters of our bucketed histogram mechanism $\mBhist{} = \mtrlap{} \circ \mbuc{}$ (\Algorithmref{histogram-mech}) while proving \Theoremref{bucketing-hist}. 

In this case, for any $x\in\A$, both $\M_1(x),f_1(x)$ are point distributions. This means that in order to compute the error sensitivity of $\M_2$, we only need to take the supremum in \eqref{eq:err-sens} over point distributions $\prob{x},\prob{x'}$ over $\B$ (where $\prob{x},\prob{x'}$ can be thought of being supported on $x:=\M_1(x)$ and $x':=f_1(x)$, respectively) such that $\Winf{}(\prob{x},\prob{x'})\leq\beta_1$.
Since $\Winf{}(\prob{x},\prob{x'})=\dBx(x,x')$, we only need to take the supremum in \eqref{eq:err-sens} over $x,x'\in\B$ such that $\dBx(x,x')\leq\beta_1$.

\paragraph{$\bullet$ When $\M_2,f_2$ are deterministic maps and $\M_1$ is $(\alpha_1,\beta_1,0)$-accurate and $\alpha_2=\gamma_2=0$.} 
This setting arises in the case of histogram-based-statistics (denoted by a deterministic function $\fhbs$) in \Sectionref{HBS}, in which we use the composed mechanism $\fhbs \circ \mBhist{}$ for computing $\fhbs$, where \mBhist{} is our final histogram mechanism that is $(\alpha,\beta,0)$-accurate (see \Theoremref{bucketing-hist}) and $\fhbs$ (as a mechanism) is $(0,0,0)$-accurate for computing $\fhbs$. 

Upon substituting these parameters in \eqref{eq:err-sens}, the expression for the error sensitivity reduces to computing $\tau_{\M_2,f_2}^{0,0}(\beta_1,0) = \sup_{X,X': \Winf{}(\prob{X},\prob{X'})\leq \beta_1} \Winf{}(\M_2(X),f_2(X'))$, which can be simplified further as shown in the lemma below, which we prove in \Appendixref{err-sens-deterministic_proof}.
\begin{lem}\label{lem:err-sens-deterministic}
Let $\M:\B\to\C$ be a deterministic mechanism for a deterministic function $f:\B\to\C$. Then, for any $\beta_1\geq0$, we have
\begin{align}\label{eq:err-sens-deterministic}
\tau_{\M,f}^{0,0}(\beta_1,0) \quad = \sup_{\substack{X,X': \\ \Winf{}(\prob{X},\prob{X'})\leq \beta_1}} \Winf{}(\M(X),f(X')) \quad = \sup_{\substack{x,x'\in\A : \\ \dB{x}{x'} \leq \beta_1}} \dC{\M(x)}{f(x')}.
\end{align}
\end{lem}

\subsection{Differential Privacy Under Composition}
First we formally define the notion of differential privacy.

\paragraph{Differential Privacy.}
Let \X denote a universe of possible ``databases'' with a
symmetric neighborhood relation $\sim$. In typical applications, two
databases \x and $\x'$ are considered neighbors if one is obtained from the
other by removing the data corresponding to a single ``individual.''
A \emph{mechanism} \M over \X is an algorithm which takes $\x\in\X$ as
input and samples an output from an output space \Y, according to some
distribution. We shall denote this distribution by $\M(\x)$.
\begin{defn}[Differential Privacy \cite{DworkMNS06,DworkKMMN06}]\label{def:epsdel-DP}
A randomized algorithm $\M:\X\to\Y$ is $(\eps,\delta)$-differentially private (DP), if for all neighboring databases $\x,\x'\in\X$
and all measurable subsets $S\subseteq\Y$, we have
$\Pr[\M(\x)\in S] \leq e^{\eps}\Pr[\M(\x')\in S] + \delta$.
\end{defn}
A simple but very useful result in differential privacy is the ``post-processing'' theorem for DP (see \cite[Proposition 2.1]{DworkRo14}), which states that if $\M_1$ is $(\epsilon,\delta)$-DP,  then for any mechanism $\M_2$, the composed mechanism $\M_2 \circ \M_1$ would remain $(\epsilon,\delta)$-DP. 
We prove a ``pre-processing'' theorem for differential privacy, which can be viewed as complementing the ``post-processing'' theorem for DP.
Our pre-processing theorem for DP states that if $\M_2$ is private, 
then so would $\M_2 \circ \M_1$ be (i.e.,
pre-processing does not hurt privacy), provided that $\M_1$ is well-behaved.
The following notion of being well-behaved suffices for our purposes.

\begin{defn}[Neighborhood preserving Mechanism]
A mechanism $\M:A\rightarrow B$ is \emph{neighborhood preserving} w.r.t.\
neighborhood relations $\sim_A$ over $A$ and  $\sim_B$ over $B$, if for all
$x,y \in A$ s.t. $x \sim_A y$, there exists a pair of jointly distributed
random variables $(X,Y)$ s.t. $\prob{X}=\M(x)$, $\prob{Y}=\M(y)$, and
$\Pr[X\sim_B Y] = 1$.  
\end{defn}
The following theorem states our pre-processing theorem for DP, which we prove in \Appendixref{comp-privacy}.
\begin{thm}[Differential Privacy Composition]\label{thm:compose-DP}
Let $\M_1:A\to B$ and $\M_2:B\to C$ be any two mechanisms.
If $\M_1$ is neighborhood-preserving w.r.t.\
neighborhood relations $\sim_A$ and $\sim_B$ over $A$ and $B$, respectively,
and $\M_2$ is $(\epsilon, \delta)$-DP w.r.t.\ $\sim_B$, 
then $\M_2\circ \M_1:A\to C$ is $(\epsilon, \delta)$-DP w.r.t.\ $\sim_A$.
\end{thm}
It is important to note here is that we are not releasing the output of the neighboring-preserving mechanism $\M_1$; we only release the output of $\M_2\circ \M_1$. 

Looking ahead, we will require \Theoremref{compose-DP} to establish the DP guarantee of our bucketed-histogram mechanism (\Algorithmref{histogram-mech}) which is obtained by pre-processing our $(\eps,\delta)$-DP histogram mechanism (\Algorithmref{hist-mech}) with the neighborhood-preserving bucketing mechanism (\Algorithmref{bucketing}).

\section{Mechanisms That Exploit Flexible Accuracy}\label{sec:mechanisms}
In this section, we propose and analyze concrete mechanisms for several
important functions. First, we present a new DP mechanism for the histogram
function with flexible accuracy in \Sectionref{histogram_mechs} and then extend it to any ``histogram based
statistic'' ({\em e.g.}, max and support) in \Sectionref{HBS}.
In \Sectionref{beyond-drop}, we show our results for other measures of distortion, beyond just dropping elements. 
Also, in \Appendixref{Comparison_BLR}, we note that the mechanisms (e.g., for half-space
queries) which required \cite{BLR} to introduce the accuracy notion of
\emph{usefulness} can be cast in the framework of flexible accuracy.

\subsection{A Private Mechanism for Releasing Histograms with Flexible Accuracy}\label{sec:histogram_mechs}
Before describing our new mechanism for releasing histograms with flexible accuracy, let us consider a simpler Boolean task of privately reporting whether a given set is empty or not. Deriving a solution to this simpler problem will pave a way towards our new histogram mechanism. 

\paragraph{Private mechanism for determining whether a given set is empty or not.}
For this, the only input distortion we are allowed is to drop some elements -- i.e., we cannot
report an empty set as non-empty. Since we seek to limit the extent of
distortion, let us add a constraint that if a set has $q$ or more elements,
then with probability 1 (or very close to 1) we should report the set as
being non-empty.  Let $p_k$ denote the probability that a set of size $k\in[0,q]$ is
reported as being non-empty, so that $p_0=0$ and $p_q=1$.

For our scheme to be \epdel{}-differential private, we require
\begin{align*}
p_k  \leq p_{k+1} \eeps + \delta, &\qquad
p_{k+1} \leq p_k \eeps + \delta, \\
(1 - p_k) \leq (1 - p_{k+1}) \eeps + \delta, &\qquad
(1 - p_{k+1})  \leq (1 - p_k)  \eeps + \delta,
\end{align*}
for $0\le k < q$, with boundary conditions
$p_0 = 0$ and $p_q = 1$.
We are interested in simultaneously reducing $\epsilon$ and $\delta$
subject to the above constraints. 
The pareto-optimal \epdel{} turn out to be given by
$\delta \eepsratio{(q/2) \eps} = \frac{1}{2}$, with corresponding
values of $p_k$ being given by
\begin{align}
\label{eq:optprob}
    p_{k} = \delta \eepsratio{{k}\eps}\text{ for } k \le \nicefrac{q}{2} \quad
\text{ and }\quad
    p_{k} = 1 - p_{q-k}\text{ for } k \ge \nicefrac{q}{2}.
\end{align}

The condition $\delta \eepsratio{(q/2) \eps} = \frac{1}{2}$ implies that we can achieve $(\eps,\eps e^{-\Omega(\eps q)})$-differential privacy.
In particular, we may choose $\eps = O\left(\frac1{\sqrt{q}}\right)$,  and
$\displaystyle \delta = O\left(\frac{e^{-\sqrt{q}/2}}{\sqrt{q}}\right)$, providing a
useful privacy guarantee when $q$ is sufficiently large.

In \Figureref{truncLap}, on the left, we plot the probabilities $p_k$ against $\nicefrac{k}q$ for this choice of \epdel.

\begin{figure}[!htbp]
\begin{center}
\includegraphics[width=0.4\textwidth]{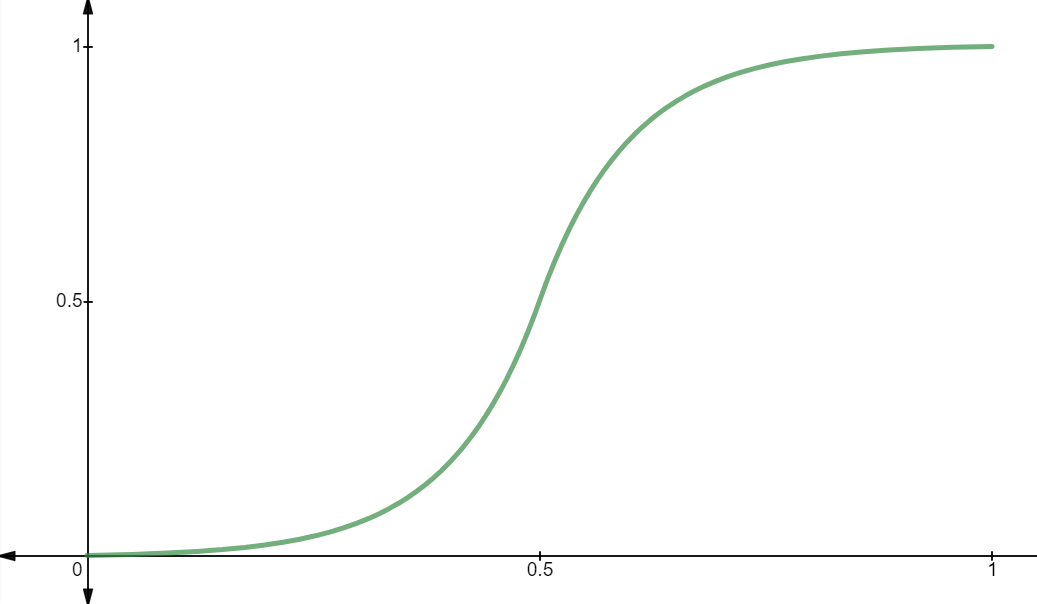}
\includegraphics[width=0.4\textwidth]{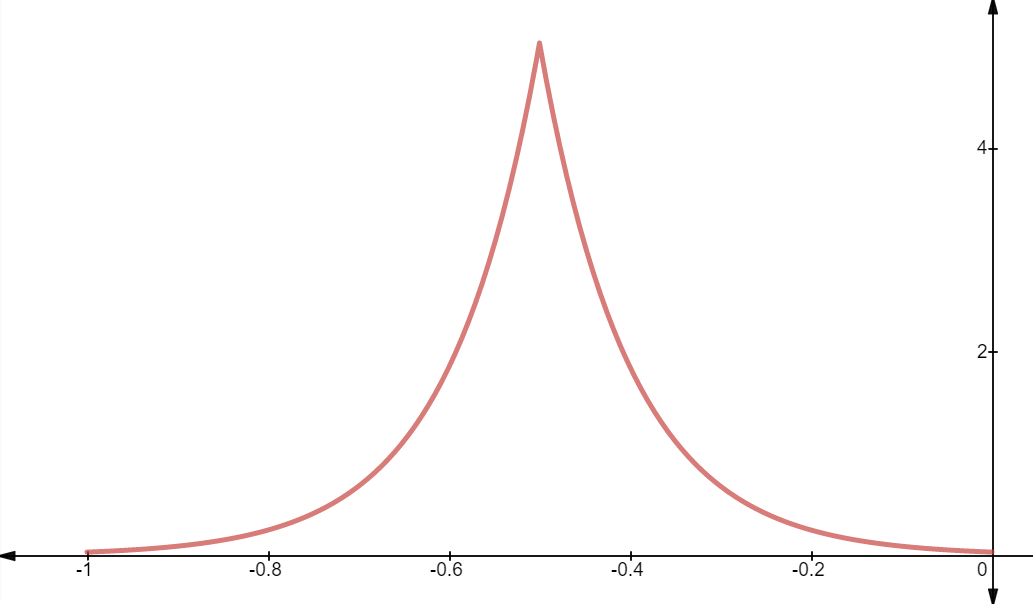}
\end{center}
\caption{The probability function in the optimal mechanism for reporting whether a set is empty or not
(left), which can be interpreted as adding a noise according to a truncated
Laplace distribution with a negative mean (right). \label{fig:truncLap}}
\end{figure}

\paragraph{Towards a private mechanism for histograms.} To generalize this Boolean mechanism to a full-fledged histogram mechanism, we reinterpret it. In a histogram mechanism,
where again, the distortion allowed in the input is to only drop elements,
we can add a \emph{negative noise} to the count in each ``bar'' of the
histogram. (If the reduced count is negative, we report it as 0.) We seek a noise function such that the probability of the
reported count being 0 (when the actual count is $k\in[0,q]$) is the same as
that of the above mechanism reporting that a set of size $k$ is empty.
That is, the probability of adding a noise $\nu \le -k$ should be $1-p_k$.
That is, if the noise distribution is given by the density function $\sigma$, 
we require that
\begin{align*}
\int_{-q}^{-k} \sigma(t) \cdot dt = 1-p_k \qquad \text{ and } \qquad
\sigma(t) = 0  \text{ for } t\not\in[-q,0] \notag.
\end{align*}
Substituting the expression for $p_k$ from \eqref{eq:optprob}, and then differentiating this identity with respect to $k$, 
we obtain the following expression for $\sigma(t)$: 
\begin{equation}\label{eq:trun_Lap_noise}
     \sigma(t) = 
\begin{cases}
\frac{1}{1 - e^{-\eps q/2}}\Lap(t \mid -\frac{q}{2}, \frac1\eps), & \text{ if } t\in[-q,0], \\
0, & \text{ otherwise},
\end{cases}
\end{equation}
where $\Lap$ is the Laplace noise distribution with mean $-\frac{q}2$ and
scale parameter $\nicefrac1\eps$.\footnote{The Laplace distribution over \R, with \emph{scaling parameter} $b > 0$ and mean $\mu$, is defined by the density function $\Lap(x|\mu,b) := \frac{1}{2b}e^{\frac{-|x-\mu|}{b}}$ for all $x\in\R$.
We denote a random variable that is distributed according to the Laplace distribution with the scaling parameter $b$ and mean 0 by $\Lap(b)$. 
} 
We call $\sigma(t)$ the shifted-truncated Laplace distribution, which is equal to the (normalized) Laplace distribution with mean $-\frac{q}2$ and scale parameter $\frac1\eps$ when $t\in[-q,0]$, and equal to zero when $t\notin[-q,0]$.

\begin{algorithm}
\caption{Shifted and Truncated Laplace Mechanism, \mtrlap{\thresh,\eps,\G}}\label{algo:hist-mech}
{\bf Parameter:} Threshold $\thresh\in[0,1)$; ground set \G; and $\eps>0$. \\
{\bf Input:} A histogram, $\x:\G\rightarrow\N$. \\
{\bf Output:} A histogram, $\y:\G\rightarrow\N$. \\
\vspace{-0.3cm}
\begin{algorithmic}[1] 
\ForAll{$g\in\G$} 
    \State $z_g \leftarrow \Lnoise$, where $q:=\thresh|\x|$ and
    $\displaystyle
	\Lnoise(z) = 
 \begin{cases} 
      \frac{1}{1-e^{-\eps{q}/{2}}} \Lap(z \mid -\frac{q}{2},\frac{1}{\eps}) & \text{ if } z\in[-q,0], \\
      0 & \text{ otherwise. }\\
     \end{cases}
     $
     \State $\y(g) := \max(0,\lfloor \x(g) + z_g \rceil)$
	 \Comment{$z_g$ need not be computed for $g$ s.t.\ $\x(g)=0$}
\EndFor
\State Return $\y$.
\end{algorithmic}
\end{algorithm}

\paragraph{The shifted-truncated Laplace mechanism for releasing histograms with flexible accuracy.}
Our final histogram mechanism is derived by adding the noise distributed according to $\sigma(t)$ from \eqref{eq:trun_Lap_noise} with appropriate parameter $q$ to each bar of the histogram, followed by rounding to the nearest integer (or to $0$, if it is negative). 
Before describing the mechanism, we need some notation.

Datasets can be abstractly represented by multi-sets, and each element in the multi-set belongs to a ground set $\G$.
Formally, a multi-set $\x$ over the ground set $\G$
is a function $\x:\G\rightarrow\N$ that outputs the multiplicity of elements
in \G. The \emph{size} and \emph{support} of \x are defined as $|\x| :=
\sum_{i\in\G} \x(i)$ and $\support(\x):=\{i\in\G:\x(i)\neq0\}$,
respectively. We shall be interested in finite-sized multi-sets, which we
refer to as histograms. We denote the domain of all histograms over \G by
\Hspace\G. For DP, the standard notion of neighborhood among histograms is
defined as $\x\nhist\x'$ iff $\sum_{i\in\G}|\x(i)-\x'(i)|\le 1$. 
Later, we shall also require \G to be a metric space, endowed with a metric
\met.

We describe our shifted-truncated Laplace mechanism for the identity function (which maps histograms to histograms and is denoted by $\mtrlap{\thresh,\eps,\G}:\Hspace\G\to\Hspace\G$) in \Algorithmref{hist-mech}.
It simply \emph{decreases} the multiplicity of each element by adding a bounded quantity
sampled from the shifted-truncated Laplace distribution. The following theorem, proven in \Sectionref{hist-priv-accu-proof}, summarizes the privacy and flexible accuracy guarantees achieved by \mtrlap{\thresh,\eps,\G} for a particular choice of $\eps$.

\begin{thm}\label{thm:hist-priv-accu}
On inputs $\x$ of size $n$, \mtrlap{\thresh,\eps,\G} from  \Algorithmref{hist-mech} satisfies the following guarantees: 
\begin{itemize}
\item \underline{Privacy:} For any $\eps,\tau$ such that $\eps\tau n \geq 2$, \mtrlap{\thresh,\eps,\G} is $\left(\eps,\eps e^{-\Omega(\eps\tau n)}\right)$-DP w.r.t.\ \nhist.
\item \underline{Flexible accuracy:} If $|\support(\x)|\le t$, then for any $\eps>0$, \mtrlap{\thresh,\eps,\G} is $(\thresh t,0,0)$-accurate for the identity function, w.r.t.\ the distortion measure \drop.
\end{itemize}
\end{thm}
\begin{remark}\label{remark:hist-params}
There are many choices of $\eps,\tau$ for which we get favorable privacy parameters in \Theoremref{hist-priv-accu}. For instance, choosing $\eps=\frac{1}{\sqrt{\tau n}}$ gives that \mtrlap{\thresh,\eps,\G} is $\left(\frac{1}{\sqrt{\tau n}},\frac{e^{-\Omega(\sqrt{\tau n})}}{\sqrt{\tau n}}\right)$-DP, provided $\tau$ is such that $\sqrt{\tau n} \geq 2$. Note that $\tau$ is the maximum overall fraction of elements we drop from each bar of the histogram. For example, by choosing $\tau=\frac{1}{n^{1/2}}$, we get that \mtrlap{\thresh,\eps,\G} is $\left(\frac{1}{n^{1/4}},\frac{e^{-\Omega(n^{1/4})}}{n^{1/4}}\right)$-DP and $(\frac{t}{n^{1/4}},0,0)$-accurate. See also \Sectionref{choosing-params} for more discussion.
\end{remark}

\Remarkref{hist-params} shows that the privacy parameters of \mtrlap{\thresh,\eps,\G} improve as the database size $|\x|$ grows, by dropping only a small number of elements, provided that the support size $t$ is small. To handle larger supports, this mechanism can be composed with a simple fixed width $w$ bucketing mechanism to give small support size, as described next. 

\paragraph{Bucketed shifted-truncated Laplace mechanism.}
In order to explain the idea behind our bucketing mechanism, for simplicity, we consider the ground set $\G=[0,B)$.\footnote{\label{foot:dim-d}We also present the general results for $\G=[0,B)^d$ (which is a $d$-dimensional cube with side-length equal to $B$) in \Appendixref{d-dim-results}. Also see \Remarkref{d-dim} in \Sectionref{beyond-drop}.} In our bucketing mechanism, we divide the interval $[0,B)$ into $t=\lceil \frac{B}{w}\rceil$ sub-intervals (buckets) of length $w$, and map each input point to the center of the nearest sub-interval (bucket). This mapping of input points to the nearest bucket introduces error in the output space, and the value of $w$ depends on the amount of error we want to tolerate in the output space. In our bucketed shifted-truncated Laplace mechanism, we run our shifted-truncated Laplace mechanism (\Algorithmref{hist-mech}) on the bucketed histogram.

\begin{algorithm}
\caption{Bucketing Mechanism, \mbuc{w,[0,B)}}\label{algo:bucketing}
{\bf Parameter:} Bucket width $w$; ground set $[0,B)$. \\
{\bf Input:} A histogram $\x$ over $[0,B)$. \\
{\bf Output:} A histogram $\y$ over $S =\{ w(i-\frac12) : i \in [t], t = \lceil \frac{B}{w} \rceil \}$, and $|\y| = |\x|$. \\
\vspace{-0.3cm}
\begin{algorithmic}[1] 
\ForAll{$s \in S $}
\State $\y(s) := \sum_{g:g-s \in [\frac{-w}2,\frac{w}2)} \; \x(g)$
\EndFor
\State Return \y
\end{algorithmic}
\end{algorithm}

\begin{algorithm}
\caption{BucketHist Mechanism, \mBhist{\alpha,\beta,[0,B)}}\label{algo:histogram-mech}
{\bf Parameter:} Accuracy parameters $\alpha,\beta$; ground set $[0,B)$. \\
{\bf Input:} A histogram \x over $[0,B)$. \\
{\bf Output:} A histogram \y over $[0,B)$. \\
\vspace{-0.3cm}
\begin{algorithmic}[1] 
\State $w := 2\beta$, $t := \lceil \frac{B}{w}\rceil$, $\thresh := \alpha/t$
\State Return $\mtrlap{\thresh,\eps,[0,B)} \circ \mbuc{w,[0,B)} (\x)$
\Comment{where \mbuc{w,[0,B)} is in \Algorithmref{bucketing}}
\end{algorithmic}
\end{algorithm}

Our bucketing mechanism $\mbuc{w,[0,B)}$ and the final bucketed-histogram mechanism $\mBhist{\alpha,\beta,[0,B)}$ are presented in \Algorithmref{bucketing} and \Algorithmref{histogram-mech}, respectively.

Since \mbuc{w,[0,B)} introduces error in the output space, we need a metric over $\Hspace{[0,B)}$ to analyze its flexible accuracy. We use the following natural metic \dhistx over $\Hspace{[0,B)}$, which is defined as $\dhist{\y}{\y'}:=\Winf{}(\frac{\y}{|\y|},\frac{\y'}{|\y'|})$. Here, $\frac{\y}{|\y|}$ is treated as a probability distribution and the underlying metric for $\Winf{}$ is the standard distance metric over $\R$.

The following theorem presents the accuracy and privacy guarantees of $\mBhist{\alpha,\beta,[0,B)}$, which we prove in \Sectionref{bucketHist-priv-accu-proof}. 

\begin{thm}\label{thm:bucketing-hist}
On inputs of size $n$, $\mBhist{\alpha,\beta,[0,B)}$ is $(\alpha, \beta, 0)$-accurate for the identity function, w.r.t.~the distortion measure $\drop$ and metric $\dhistx$. Furthermore, for any $\eps>0$, and $\tau=\alpha(\frac{2\beta}B)$, if $\eps\tau n \geq 2$, then $\mBhist{\alpha,\beta,[0,B)}$ is $\left(\eps,\eps e^{-\Omega(\eps\tau n)}\right)$-DP w.r.t.\ \nhist.
\end{thm}

We can instantiate \Theoremref{bucketing-hist} with different parameter settings to achieve favorable privacy-accuracy tradeoffs. See \Sectionref{choosing-params} for more details.

\subsection{Histogram-Based-Statistics}\label{sec:HBS}
\Theoremref{bucketing-hist} provides a powerful tool to obtain a DP mechanism
for \emph{any deterministic} histogram-based-statistic $\fhbs: \Hspace{[0,B)} \rightarrow \A$, simply
by defining 
\begin{align} 
\M_{\fhbs}^{\alpha, \beta, [0,B)} = \fhbs \circ \mBhist{\alpha,\beta,[0,B)}. \label{fhbs-mech-1dim}
\end{align}
To analyze the flexible accuracy of $\M_{\fhbs}$, we 
define the \emph{metric sensitivity} function of $\fhbs$.
\begin{defn}\label{def:metric-sensitivity}
The \emph{metric sensitivity} of a histogram-based-statistic $\fhbs: \Hspace{[0,B)} \rightarrow \A$, is given by $\Delta_{\fhbs}: \Rplus \rightarrow \Rplus$,
in terms of a metric \dAx over $\A$,
\begin{align}\label{eq:fhbs-sens}
 \Delta_\fhbs(\beta) = \sup_{\substack{\x, \x'\in\Hspace{[0,B)} : \\ \dhist{\x}{\x'} \leq \beta}} \dA{\fhbs(\x)}{\fhbs(\x')}.
\end{align}
\end{defn}
The privacy and accuracy guarantees of our HBS mechanism are stated in the following theorem, which we prove in \Sectionref{proof_bucketing-general}.
\begin{thm}\label{thm:bucketing-general}
On inputs of size $n$, $\M_{\fhbs}^{\alpha, \beta, [0,B)}$ is $(\alpha, \Delta_{\fhbs}(\beta), 0)$-accurate for \fhbs w.r.t.\ distortion \drop and metric \dAx. Furthermore, for any $\eps>0$, and $\tau=\alpha(\frac{2\beta}B)$, if $\eps\tau n \geq 2$, then $\M_{\fhbs}^{\alpha, \beta, [0,B)}$ is $\left(\eps,\eps e^{-\Omega(\eps\tau n)}\right)$-DP.
\end{thm}
We can instantiate \Theoremref{bucketing-general} with different parameter settings to achieve favorable privacy-accuracy tradeoffs. See \Sectionref{choosing-params} for more details.

\Theoremref{bucketing-general} has direct applications to functions
which have high sensitivity (defined w.r.t.\ the neighborhood relation
$\sim$), but low metric sensitivity. We point out two such examples,
for which no solutions with non-trivial guarantees were previously offered.

\subsubsection{Computing the Maximum or Minimum Element of a Multi-set}\label{sec:max}
We define \fmax (or simply $\max$) for histograms over real numbers as $\fmax(\x):=\max\{g:\x(g)>0\}$.
Similarly, we can define \fmin (or simply $\min$) as $\fmin(\x):=\min\{g:\x(g)>0\}$. 
We give our result for $\fmax$ only; the same result holds for $\fmin$ as well.
\begin{corol}\label{corol:bucketing-max}
On inputs of size $n$, $\mmax{\alpha, \beta, [0,B)}$ is $(\alpha, \beta, 0)$-accurate for $\fmax$ w.r.t.\ the distortion \drop and the standard distance metric over \R. Furthermore, for any $\eps>0$, and $\tau=\alpha(\frac{2\beta}B)$, if $\eps\tau n \geq 2$, then $\mmax{\alpha, \beta, [0,B)}$ is $\left(\eps,\eps e^{-\Omega(\eps\tau n)}\right)$-DP.
\end{corol}

The proof of \Corollaryref{bucketing-max} is straight-forward, and we prove it in \Sectionref{bucketing-max_proof}.

\subsubsection{Computing the Support of a Multi-set}\label{sec:support}
\fsupp (or simply \support) is defined as $\fsupp(\x):=\{g:\x(g)>0\}$,
which maps a multiset to the set that forms its
support. To measure accuracy, we use a metric \dsuppx over
the set of finite subsets of \R: for any two finite subsets $\cS_1,\cS_2\subseteq \R$, define 
\[\dsupp{\cS_1}{\cS_2} := \max\left\{
\max_{s_1\in \cS_1} \min_{s_2\in \cS_2} |s_1-s_2|,\, \max_{s_2 \in \cS_2}\min_{s_1 \in \cS_1} |s_2-s_1|\right\}.\]
\dsuppx measures the farthest that a point in one of the sets is from any
point on the other set. For example, if $s_i^{\min}:=\min_{s\in\cS_i}\{s\}$ and $s_i^{\max}:=\max_{s\in\cS_i}\{s\}$ denote the minimum and the maximum elements of the set $\cS_i$ (for $i=1,2$), respectively, then it can be verified that $\dsupp{\cS_1}{\cS_2}=\max\{|s_1^{\min}-s_2^{\min}|,|s_1^{\max}-s_2^{\max}|\}$.

\begin{corol}\label{corol:bucketing-supp}
On inputs of size $n$, $\msupp{\alpha, \beta, [0,B)}$ is $(\alpha, \beta, 0)$-accurate for $\fsupp$ w.r.t.\ the distortion \drop and metric $\dsuppx$. Furthermore, for any $\eps>0$, and $\tau=\alpha(\frac{2\beta}B)$, if $\eps\tau n \geq 2$, then $\msupp{\alpha, \beta, [0,B)}$ is $\left(\eps,\eps e^{-\Omega(\eps\tau n)}\right)$-DP.
\end{corol}

The proof of \Corollaryref{bucketing-supp} is straight-forward, and we prove it in \Sectionref{bucketing-supp_proof}.

\subsubsection{Choosing the Parameters}\label{sec:choosing-params}
As mentioned in \Remarkref{hist-params} for \Theoremref{hist-priv-accu}, there are many choices of $\eps,\tau$ for which we can get favorable privacy, accuracy parameters in Theorems~\ref{thm:bucketing-hist},~\ref{thm:bucketing-general}, and Corollaries~\ref{corol:bucketing-max},~\ref{corol:bucketing-supp}. 
For concreteness, in the following, we illustrate the privacy accuracy trade-off by choosing parameters for the $\mmax{\alpha, \beta, [0,B)}$ mechanism in \Corollaryref{bucketing-max}; the same result applies to Theorems~\ref{thm:bucketing-hist},~\ref{thm:bucketing-general}, and \Corollaryref{bucketing-supp} as well. 

If we choose $\eps=\frac{1}{\sqrt{\tau n}}$ and $\tau$ is such that $\frac{1}{\eps}=\sqrt{\tau n} \geq 2$, then by dropping only $\alpha n = \frac{1}{\eps^2}\frac{2\beta}{B}$ elements from the entire dataset, the mechanism $\mmax{\alpha,\beta,[0,B)}$ achieves $\left(\frac{1}{\sqrt{\tau n}},\frac{e^{-\Omega(\sqrt{\tau n})}}{\sqrt{\tau n}}\right)$-differential privacy. 
If $\beta/B$ is a small constant (say, $1/100$), which corresponds to perturbing the output by a small constant fraction of the whole range $B$, then by dropping only $\alpha n = O(\frac{1}{\eps^2})$ elements, $\mmax{\alpha,\beta,[0,B)}$ achieves $(\eps,\eps e^{-\Omega(\frac{1}{\eps})})$-differential privacy. 
We can set any $\tau$ that satisfies $\frac{1}{\eps}=\sqrt{\tau n}\geq2$ in this result. For example,\\

\parbox{15.5cm}{By setting $\eps=\frac{1}{(\log n)^2}$, we get that by dropping only $O((\log n)^4)$ elements from the entire dataset, $\mmax{\alpha,\beta,[0,B)}$ achieves $(\frac{1}{(\log n)^2},\frac{n^{-\Omega(\log n)}}{(\log n)^2})$-differential privacy while incurring only a small constant error (of the entire range) in the output.\\} 

\noindent Note that in the above setting of parameters, we take $\eps=\frac{1}{\sqrt{\tau n}}$, which implies that the bound on $\delta$ can at best be a small constant for any constant $\eps$. This is because $\eps\tau n=\sqrt{\tau n} = \frac{1}{\eps}$ is a constant, which implies that $\delta=\eps e^{-\Omega(\frac{1}{\eps})}$ will be a constant too. Therefore, for getting privacy guarantees with small constant $\eps$ such that $\delta$ (exponentially) decays with $n$, we will work with the general privacy result of $(\eps,\eps e^{-\Omega(\eps\tau n)})$-DP as in \Corollaryref{bucketing-max}. 
For example, \\

\parbox{15.5cm}{By setting $\eps=0.1$ and $\tau=\frac{1}{n^c}$ (for any $c\in(0,1)$), we get that by dropping only $\alpha n = \tau n \frac{B}{2\beta} = O(n^{1-c})$ elements from the entire dataset, $\mmax{\alpha,\beta,[0,B)}$ achieves $(0.1,e^{-\Omega(n^{1-c})})$-differential privacy while incurring only a small constant error (of the entire range) in the output.\\} 

\noindent For other parameter settings, see the result on page~\pageref{informal-result-max} after we stated our informal result for max.

\subsection{Further Applications: Beyond \drop}\label{sec:beyond-drop}
Useful variants of \Theoremref{bucketing-general} can be obtained with
measures of distortion other than \drop. In particular, in \eqref{eq:move_defn} and \eqref{eq:drop_move_defn}, we defined the distortions $\move$ and $\dropmove\eta$, respectively, where $\move$ allows moving/perturbing of data points and $\dropmove\eta$ allows both dropping and moving.

The following theorem provides the privacy and accuracy guarantees of $\M_{\fhbs}^{\alpha, \beta, [0,B)}$ (defined in \eqref{fhbs-mech-1dim}) w.r.t.\ the distortion measure $\drme$, and we prove it in \Sectionref{beyond-drop_proofs}.

\begin{thm}\label{thm:bucketing-general-drmv}
On inputs of size $n$, $\M_{\fhbs}^{\alpha, \beta, [0,B)}$ is $(\alpha+\eta\beta, 0, 0)$-accurate for \fhbs w.r.t.\ the distortion measure $\dropmove\eta$. Furthermore, for any $\eps>0$, and $\tau=\alpha(\frac{2\beta}B)$, if $\eps\tau n \geq 2$, then $\M_{\fhbs}^{\alpha, \beta, [0,B)}$ is $\left(\eps,\eps e^{-\Omega(\eps\tau n)}\right)$-DP.
\end{thm}
\begin{remark}
This is analogous to \Theoremref{bucketing-general}, 
but with the important difference that it does
not refer to the metric sensitivity of the function \fhbs, and does not even
require a metric over its codomain \cA. This makes this result applicable to
complex function families like maximum-margin separators or neural net
classifiers. However, the accuracy notion uses a measure of distortion that
allows dropping a (small) fraction of the data \emph{and} (slightly) moving all
data points, which may or may not be acceptable to all applications.
\end{remark}
\begin{remark}[Extending the results from $[0,B)$ to $[0,B)^d$]\label{remark:d-dim}
Note that the bucketing mechanism $\mbuc{w,[0,B)}$ (\Algorithmref{bucketing}) and the bucketed-histogram mechanism $\mBhist{\alpha,\beta,[0,B)}$ (\Algorithmref{histogram-mech}) are given for the ground set $\G=[0,B)$.
However, as mentioned in \Footnoteref{dim-d}, they can easily be extended to the $d$-dimensional ground set $\G=[0,B)^d$, and we present the $d$-dimensional analogues of the above two mechanisms in \Appendixref{d-dim-results}. All our results in \Theoremref{bucketing-hist}, \Theoremref{bucketing-general}, and \Theoremref{bucketing-general-drmv} will hold verbatim with these generalized mechanisms, except for the value of $\tau$, which will be replaced by $\tau=\alpha(\frac{2\beta}{B\sqrt{d}})^d$; see \Appendixref{d-dim-results} for a proof of this.
\end{remark}

\section{Proofs}\label{sec:proofs}
In our proofs, when dealing with infimum/supremum (for example, in the definitions of the lossy Wasserstein distance, measure of distortion, distortion and error sensitivities, etc.), for simplicity, we assume that the infimum/supremum is always achieved; all our proofs can be easily extended to work without this assumption by taking appropriate limits when working with infinitesimal quantities. 

\subsection{Proof of \Lemmaref{wass-triangle} -- Triangle Inequality for $\Winf{\gamma}$}\label{sec:triangle-ineq_Wass}
In this section we prove \Lemmaref{wass-triangle}, and along the way derive useful properties about lossy Wasserstein distance, that may be of independent interest. 

The following lemma is crucial to proving \Lemmaref{wass-triangle}.
\begin{lem}\label{lem:wass-marginal-loss}
Let $P$ and $Q$ be any two distributions over a metric space
$(\Omega,\met)$. If $\Winf{\gamma}(P, Q) = \beta$, then for all $\gamma_1 \in [0, \gamma]$, there exist distributions $P'$ and $Q'$ s.t. $\Delta(P, P') \le \gamma_1$, $\Delta(Q, Q') \le \gamma - \gamma_1$, and $\Winf{}(P', Q') = \beta$.
\end{lem}

\begin{proof}
Let $P$ and $Q$ be any two distributions over a metric space
$(\Omega,\met)$. Let us assume that the optimal $\Winf{\gamma}(P, Q)$ (= $\beta$) is obtained at the joint distribution $\phi_{opt}$. Let the first and the second marginal distributions of $\phi_{opt}$ be $P_{opt}$ and $Q_{opt}$, respectively. Let $\Delta(P, P_{opt}) = \gamma_{opt}$, which implies that $\Delta(Q, Q_{opt}) \le \gamma - \gamma_{opt}$. 
Define a function $R_{opt}:\Omega\to\R$ as $R_{opt}(\omega) := \Prob{P_{opt}}{\omega} - \Prob{P}{\omega}$ for all $\omega\in\Omega$. Clearly, $\int_{\Omega} R_{opt}(\omega)\dd\omega= 0$ and $\int_{\Omega} |R_{opt}(\omega)|\dd\omega= 2\gamma_{opt}$.

In the discussion below, we shall take a general $\gamma_1 \in [0,\gamma_{opt})$ and construct distributions $P'$ and $Q'$ s.t.\ $\Delta(P, P') \le \gamma_1$, $\Delta(Q, Q') \le \gamma - \gamma_1$, and $\Winf{}(P', Q') = \beta$, as required in the conclusion of \Lemmaref{wass-marginal-loss}. 
We can show a similar result for the other case also when $\gamma_1 \in (\gamma_{opt}, \gamma]$ (by swapping the roles of $P$ and $Q$ in the above as well as in the argument below). This will complete the proof of \Lemmaref{wass-marginal-loss}.

Define a function $R':\Omega\to\R$ as $R'(\omega) := \frac{\gamma_1}{\gamma_{opt}}R_{opt}(\omega)$. 
For any $\omega\in\Omega$, let $P'(\omega)=\Prob{P}{\omega} + R'(\omega)$. 
After substituting the value of $R_{opt}(\omega) = \Prob{P_{opt}}{\omega} - \Prob{P}{\omega}$, we get $P'(\omega)=\frac{\gamma_1}{\gamma_{opt}}\Prob{P_{opt}}{\omega} + \left(1-\frac{\gamma_1}{\gamma_{opt}}\right)\Prob{P}{\omega}$. Since $P'$ is a convex combination of two distributions, it is also a valid distribution.
It is easy to see that $\Delta(P, P') = \gamma_1$. 
Define a joint distribution $\phi'$ as follow: for every $(x,y)\in\Omega\times\Omega$, define

\begin{align*}
    \phi'(x, y) := \begin{cases}
    \phi_{opt}(x, y)\frac{\Prob{P'}{x}}{\Prob{P_{opt}}{x}} & \text{if } \Prob{P_{opt}}{x} > 0\\
    \Prob{P'}{x}\delta(x -  y) & \text{otherwise}
    \end{cases}
\end{align*}
where $\delta(\cdot)$ is the Dirac delta function. 
It follows from the definition that $\int_{\Omega} \phi'(x, y)\dd y = \Prob{P'}{x}$, i.e., the first marginal of $\phi'$ is $\Prob{P'}{\cdot}$. This also implies that $\phi'$ is a valid joint distribution because 
{\sf (i)} $\phi'(x,y)\geq0$ for all $(x,y)\in\Omega\times\Omega$, and
{\sf (ii)} $\int_{\Omega\times\Omega}\phi'(x,y)\dd x \dd y = \int_{\Omega}\Prob{P'}{x}\dd x = 1$. 

Let the second marginal of $\phi'$ be $Q'$. We show in \Claimref{TV_dist_Q-Qprime} in \Appendixref{wasserstein} that $\Delta(Q, Q')\leq\gamma-\gamma_1$.

The only thing left to prove is to show that $\Winf{}(P', Q') = \beta$ for the above constructed $P'$ and $Q'$. 
First we show $\Winf{}(P', Q') \ge \beta$ and then show $\Winf{}(P', Q') \le \beta$.
\begin{itemize}
\item {\bf Showing $\Winf{}(P', Q') \ge \beta$:} This follows from the following claim, which we prove in \Appendixref{wasserstein}.
\begin{claim}\label{clm:wass_alternate}
For distributions $P$ and $Q$ over a metric space $(\Omega,\met)$ and $\gamma\in [0,1]$, we have
\begin{align}\label{eq:wass_alternate}
\Winf{\gamma}(P,Q) \quad= \displaystyle \inf_{\substack{\hat{P},\hat{Q}:\\ \Delta(P,\hat{P}) + \Delta(Q,\hat{Q})\leq\gamma}} \Winf{} (\hat{P},\hat{Q}).
\end{align}
\end{claim}

Now, since $P',Q'$ satisfy $\Delta(P,P') + \Delta(Q,Q')\leq\gamma$, we have $\Winf{\gamma}(P,Q)\leq\Winf{}(P',Q')$. Since $\Winf{\gamma}(P,Q)=\beta$, we have shown that $\Winf{}(P',Q')\geq\beta$.

\item {\bf Showing $\Winf{}(P', Q') \le \beta$:}
For the sake of contradiction, let us assume that $\Winf{}(P', Q') > \beta$. 
Then there is a pair $(x, y) \in \Omega^2$ such that $\phi'(x, y) > 0$ and $\met(x, y) > \beta$. 
This implies that $\phi_{opt}(x, y) = 0$, because, otherwise, we would have $\Winf{\gamma}(P, Q) > \beta$, 
which contradicts our hypothesis that $\Winf{\gamma}(P, Q) = \beta$. 
So, we know that $\phi'(x, y) > 0$ and $\phi_{opt}(x, y) = 0$. From the definition of $\phi'$, this is only possible if $\Prob{P_{opt}}{x} = 0$ and $\Prob{P'}{x}\delta(x -  y) > 0$. This can happen only if $x =  y$, but this implies $\met(x, y) = 0 \le \beta$, which is a contradiction. Hence $\Winf{}(P', Q') \le \beta$.
\end{itemize}
This completes the proof of \Lemmaref{wass-marginal-loss}.
\end{proof}

Now we are ready to prove \Lemmaref{wass-triangle}.
Let $\Winf{\gamma_1}(P, Q) = \beta_1$ and $\Winf{\gamma_2}(Q, R) = \beta_2$.
It follows from \Lemmaref{wass-marginal-loss} that there exists a distribution $P'$ such that $\Delta(P, P') \le \gamma_1$ and $\Winf{}(P', Q) = \beta_1$. Similarly, there exists a distribution $R'$ such that $\Delta(R, R') \le \gamma_2$ and $\Winf{}(Q, R') = \beta_2$.
Using these, we have from \Lemmaref{Winf_triangle} that $\Winf{}(P', R') \leq \beta_1 + \beta_2$.

Now, the result follows from the following set of inequalities.
\begin{align*}
\Winf{\gamma_1+\gamma_2}(P,R) 
\ \stackrel{\text{(d)}}{=}\hspace{-0.5cm}\displaystyle \inf_{\substack{\hat{P},\hat{R}:\\ \Delta(P,\hat{P}) + \Delta(R,\hat{R})\leq\gamma_1+\gamma_2}}\hspace{-0.5cm} \Winf{} (\hat{P},\hat{R}) 
\ \stackrel{\text{(e)}}{\leq} \ \Winf{}(P',R') 
\ \leq\ \beta_1 + \beta_2 
\ = \ \Winf{\gamma_1}(P, Q) + \Winf{\gamma_2}(Q, R),
\end{align*}
where (d) follows from \Claimref{wass_alternate}
and (e) follows because $P',R'$ satisfy $\Delta(P,P') + \Delta(R,R')\leq\gamma_1+\gamma_2$.

This concludes the proof of \Lemmaref{wass-triangle}.

\subsection{Proof of \Theoremref{compose-accuracy} -- Composition Theorem for Flexible Accuracy}\label{sec:compose-accuracy_proof}

The following lemma will be useful in proving \Theoremref{compose-accuracy}. 
It translates the definition of distortion sensitivity
(\Definitionref{dist-sens}) to apply to distortion of input distributions.
We prove it in \Appendixref{comp-accuracy}.

\begin{lem}\label{lem:distsens-distrib}
Suppose $f: A \to B$ has distortion sensitivity \distsens{f}{} w.r.t.\ $(\dn_1,\dn_2)$. 
For all r.v.s $X_0$ over $A$ and $Y$ over $B$ such that $\dnx_2(f(X_0),\prob{Y}) \le \alpha$ for some $\alpha\geq0$, there must exist a r.v.\ $X$ over $A$ such that $Y=f(X)$ and $\dnx_1(\prob{X_0}, \prob{X}) \le \distsens{f}{}(\alpha)$, provided $\distsens{f}{}(\alpha)$ is finite.
\end{lem}

Now we prove \Theoremref{compose-accuracy}, which is essentially formalizing the pictorial proof given in \Figureref{composition}.

For a given element $x \in A$, since $\M_1$ is $(\alpha_1 , \beta_1, \gamma_1)$-accurate mechanism for $f_1$, we have from \Definitionref{alpha-beta-gamma-accu} that there exists a r.v.\ $X'$ such that
\begin{align}
\dnx_1(x, \prob{X'}) &\le \alpha_1,            \label{eq:compose-accuracy-dnx1}\\
\Winf{\gamma_1}(f_1(X'),\M_1(x))&\leq\beta_1.  \label{eq:compose-accuracy-W1}
\end{align}
Now, applying the mechanism $\M_2$ on $\M_1(x)$, we incur an overall error of at most $\tau_{\M_2,f_2}^{\alpha_2, \gamma_2}(\beta_1,\gamma_1)$ to the output of function $f_2$ over a distorted input (see \Definitionref{err-sens}). Therefore, there exists a r.v.\ $Y^*$ such that,
\begin{align}
    \dnx_2( f_1(X'), \prob{Y^*}) &\leq \alpha_2,          \label{eq:compose-accuracy-dnx2}\\
    \Winf{\gamma_2}(f_2 (Y^*),\M_2(\M_1(x))) &\leq \tau_{\M_2,f_2}^{\alpha_2, \gamma_2}(\beta_1,\gamma_1).   \label{eq:compose-accuracy-W2}
\end{align}
Since $\distsens{f_1}{}(\alpha_2)$ is finite (by assumption), it follows from \eqref{eq:compose-accuracy-dnx2} and \Lemmaref{distsens-distrib} that there exists a r.v.\ $X$ over $A$ such that
\begin{align}
\dnx_1(\prob{X'}, \prob{X}) &\leq \distsens{f_1}{}(\alpha_2), \label{eq:compose-accuracy_interim 2} \\
Y^* &= f_1(X). \label{eq:compose-accuracy_interim 3}
\end{align}
Since $\dn_1$ is a quasi-metric, it follows that $\dnx_1$ is also a quasi-metric; see \Lemmaref{dnx-quasi-metric} in \Appendixref{comp-accuracy} for a proof. This, together with
\eqref{eq:compose-accuracy-dnx1} and \eqref{eq:compose-accuracy_interim 2}, implies that
\begin{align}
\dnx_1(x, \prob{X}) \leq \alpha_1 + \distsens{f_1}{}(\alpha_2). \label{eq:compose-accuracy_interim1}
\end{align}
Substituting $Y^* = f_1(X)$ from \eqref{eq:compose-accuracy_interim 3} into \eqref{eq:compose-accuracy-W2} gives
\begin{align}\label{eq:compose-accuracy_interim5}
\Winf{\gamma_2}(f_2 (f_1(X)),\M_2(\M_1(x))) \leq \tau_{\M_2,f_2}^{\alpha_2, \gamma_2}(\beta_1,\gamma_1).
\end{align}
\eqref{eq:compose-accuracy_interim1} and \eqref{eq:compose-accuracy_interim5} imply that $\M_2\circ\M_1$ is $(\alpha,\beta,\gamma)$-accurate for $f_2\circ f_1$ w.r.t.\ the distortion measure $\dn_1$ on $A$ and metric $\met_2$ on $C$, where $\alpha=\alpha_1 + \distsens{f_1}{}(\alpha_2)$, $\beta=\tau_{\M_2,f_2}^{\alpha_2, \gamma_2}(\beta_1,\gamma_1)$, and $\gamma=\gamma_2$.

This concludes the proof of \Theoremref{compose-accuracy}.

\subsection{Proof of \Theoremref{hist-priv-accu} -- Truncated Laplace Mechanism for Histograms}\label{sec:hist-priv-accu-proof}
First we prove the flexible accuracy part, which is easy, and then we will move on to proving the privacy part, which is more involved than the existing privacy analysis of differentially-private histogram mechanisms.
We also note that the requirement of $|\support(\x)|\le t$ is only needed the accuracy result.

\paragraph{Flexible accuracy.}
Note that the noise added by \mtrlap{\thresh,\eps,\G} in each bar of the histogram is bounded by $-q=-\thresh |\x|$, which can lead to a drop of at most $\thresh$ fraction of total number of elements from each bar. Combined with the fact that $|\support(\x)|\le t$, the fraction of the maximum fraction of elements that can be dropped is $\thresh t$. Hence, \mtrlap{\thresh,\eps,\G} is $(\thresh t,0,0)$-accurate. \\

\paragraph{Differential privacy.} Our proof of the privacy part of \Theoremref{hist-priv-accu} depends on the following lemma.
\begin{lem}\label{lem:hist-epdel}
For any $\nu\geq0,\eps>0$ and on inputs $\x$ s.t.\ $|\x| \ge \frac{2}{\eps\thresh} \ln\left(1 + \frac{1 -  e^{-\frac{\eps\thresh}{2}}}{e^{\eps(\nu + \frac{\thresh}{2})} - 1} \right)$, \mtrlap{\thresh,\eps,\G} is 
$\left((1 + \nu)\eps,\frac{e^{\eps} - 1}{2(e^{\nicefrac{\eps q}{2}} - 1)}\right)$-DP w.r.t.\ \nhist, where $q=\tau|\x|$.
\end{lem}
\begin{proof}
We shall in fact prove that a mechanism that outputs $\hat\y$ with $\hat\y(i) :=
\x(i)+z_i$ (without rounding, and without replacing negative values with 0)
is already differentially private as desired. Then, since the actual mechanism is a
post-processing of this mechanism, it will also be differentially
private with the same parameters.

Let $\x$ and $\x'$ be two neighbouring histograms.
For simplicity, for every $i\in\G$, define $x_i:=\x(i)$ and $x_i':=\x'(i)$.
Since $\x\sim\x'$, there exists an
$i^*\in\G$ such that $|x_{i^*}-x_{i^*}'|= 1$ and that $x_i=x_i'$ for every $i\in \G\setminus\{i^*\}$. 
Without loss of generality, assume that $x_{i^*}=x_{i^*}'+1$, which implies $|\x| = |\x'|+1 = n + 1$.
Let $q = \thresh (n + 1)$ and $q' = \thresh n$. 
For simplicity of notation, we will denote $\support(\y)$ by $\G_{\y}$ for any $\y\in\{\x,\x'\}$.

In order to prove the lemma, for every subset $S\subseteq\Hspace\G$, we need to show that 
\begin{align}
\Pr[\mtrlap{\thresh,\eps,\G}(\x')\in S] \leq e^{(1 + \nu)\eps}\Pr[\mtrlap{\thresh,\eps,\G}(\x)\in S]+\delta, \label{eq:hist-dp_interim2} \\
\Pr[\mtrlap{\thresh,\eps,\G}(\x)\in S] \leq e^{(1 + \nu)\eps}\Pr[\mtrlap{\thresh,\eps,\G}(\x')\in S]+\delta, \label{eq:hist-dp_interim1}
\end{align}
where $\delta=\frac{e^{\eps} - 1}{2(e^{\eps \nicefrac{q}{2}} - 1)}$.
We only prove \eqref{eq:hist-dp_interim2}; \eqref{eq:hist-dp_interim1} can be shown similarly.

Fix an arbitrary subset $S\subseteq\Hspace\G$.
Since $\mtrlap{\thresh,\eps,\G}$ adds independent noise to each bar of the histogram according to $\lnoise{z}$, we have that for every $\s\in\Hspace\G$, we have $\prob{\mtrlap{\thresh,\eps,\G}(\x)}(\s)=\prod_{i\in\G_{\x}} \lnoise{s_i - x_i}$ where $s_i = \s(i)$. Thus, we have
\begin{align}
\Pr[\mtrlap{\thresh,\eps,\G}(\x)\in S] &= \int_{S} \big[\prod_{i\in\G_{\x}} \lnoise{s_i - x_i}\big]\dd\s, \label{eq:app-mech2-eps-delta-dp2} \\
\Pr[\mtrlap{\thresh,\eps,\G}(\x')\in S] &= \int_{S} \big[\prod_{i\in\G_{\x'}} \dnoise{s_i - x_i'}\big]\dd\s. \label{eq:app-mech2-eps-delta-dp1}
\end{align}
Now, using the fact that $\forall k \neq i^*, x_k = x_k'$ and $x_{i^*} = x_{i^*}' + 1$, we partition $S$ into three disjoint sets:
\begin{enumerate}
\item $S_0:=\{\s\in\Hspace\G: s_{i^*} - x_{i^*}' < -q'\} \cup \{\s\in\Hspace\G: 0 < s_{i^*} - x_{i^*}'\}$.
\item $S_1:=\{\s\in\Hspace\G: -q' \le s_{i^*} - x_{i^*}' < -q' + (1 - \thresh)\}$.
\item $S_2:=\{\s\in\Hspace\G: -q' + (1 - \thresh) \le s_{i^*} - x_{i^*}' \le 0)\}$.
\end{enumerate}
The proof of \eqref{eq:hist-dp_interim2} is a simple corollary of the following two claims, which we prove in \Appendixref{histogram}.
\begin{claim}\label{clm:hist-epdel-c2}
$\Pr[\mtrlap{\thresh,\eps,\G}(\x')\in S_0\cup S_2]\leq e^{(1 + \nu)\eps}\Pr[\mtrlap{\thresh,\eps,\G}(\x)\in S_0\cup S_2]$, provided $n \ge \frac{2}{\eps\thresh} \ln\left(1 + \frac{1 -  e^{-\frac{\eps\thresh}{2}}}{e^{\eps(\nu + \frac{\thresh}{2})} - 1} \right)$.
\end{claim}
\begin{claim}\label{clm:hist-epdel-c3}
$\Pr[\mtrlap{\thresh,\eps,\G}(\x')\in S_{1}]\leq \delta$, where $\delta=\frac{e^{\eps} - 1}{2(e^{\nicefrac{\eps q}{2}} - 1)}$.
\end{claim}

The above two claims together imply \eqref{eq:hist-dp_interim2} as follows:
\begin{align}
\Pr[\mtrlap{\thresh,\eps,\G}(\x')\in S] &= \Pr[\mtrlap{\thresh,\eps,\G}(\x')\in S_0\cup S_2] + \Pr[\mtrlap{\thresh,\eps,\G}(\x')\in S_1]\nonumber\\
&\leq e^{(1+\nu)\eps}\Pr[\mtrlap{\thresh,\eps,\G}(\x)\in S_0\cup S_2] + \delta \nonumber\\
&\leq e^{(1+\nu)\eps}\Pr[\mtrlap{\thresh,\eps,\G}(\x)\in S] + \delta. \tag{Since $S_0\cup S_2\subseteq S$}
\end{align}
This completes the proof of \Lemmaref{hist-epdel}.
\end{proof}

In \Lemmaref{hist-epdel}, $\nu$ is a free variable. By taking $\nu=0$, we get the following result in \Corollaryref{simpler-dp-hist}. We can also get different guarantees by restricting to $\nu>0$; see \Remarkref{priv-nu-bigger-0} below for this.

\begin{corol}\label{corol:simpler-dp-hist}
For any $\eps,\tau,\x$ such that $\tau|\x|\eps \ge 2$, \mtrlap{\thresh,\eps,\G} is 
$\left(\eps,\frac{e^{\eps} - 1}{2(e^{\nicefrac{\eps q}{2}} - 1)}\right)$-DP w.r.t.\ \nhist, where $q=\tau|\x|$.
\end{corol}
\begin{proof}
Substituting $\nu=0$ in \Lemmaref{hist-epdel} gives that when $\x$ satisfies $|\x| \ge \frac{2}{\eps\thresh} \ln\left(1 + \frac{1 -  e^{-\frac{\eps\thresh}{2}}}{e^{\frac{\eps\thresh}{2}} - 1} \right)$, we have that \mtrlap{\thresh,\eps,\G} is $\left(\eps,\frac{e^{\eps} - 1}{2(e^{\nicefrac{\eps n\thresh}{2}} - 1)}\right)$-DP w.r.t.\ \nhist. Now, the corollary follows because 
$\frac{2}{\eps\tau}\geq\frac{2}{\eps\thresh} e^{-\frac{\eps\thresh}{2}} \geq \frac{2}{\eps\thresh} \ln\left(1 + e^{-\frac{\eps\thresh}{2}} \right) = \frac{2}{\eps\thresh} \ln\left(1 + \frac{1 -  e^{-\frac{\eps\thresh}{2}}}{e^{\frac{\eps\thresh}{2}} - 1} \right)$,
where the first inequality uses $x \ge \ln(1+x)$ for $x>0$.
\end{proof}

\begin{remark}\label{remark:priv-nu-bigger-0}
We show in \Lemmaref{hist-priv-nu-bigger-0} in \Appendixref{histogram} that by restricting \Lemmaref{hist-epdel} to $\nu>0$, we can get a weaker condition than what we have in \Corollaryref{simpler-dp-hist} with a slight increase in the privacy parameter $\eps$. In particular, we show that for all $\eps,\x$ such that $\eps\nu\geq\ln\left(1+\frac{1}{|\x|}\right)$, \mtrlap{\thresh,\eps,\G} is $\left((1+\nu)\eps,\frac{e^{\eps} - 1}{2(e^{\nicefrac{\eps q}{2}} - 1)}\right)$-DP w.r.t.\ \nhist. We can take $\nu=1$ here.
\end{remark}

Now the privacy part of \Theoremref{hist-priv-accu} follows because $q=\tau|\x|$ and $\tau|\x|\eps\geq2$ (note that $\tau|\x|\eps$ is typically a much bigger number than $2$ as it scales with the size of the dataset), which implies that $\frac{e^{\eps} - 1}{2(e^{\nicefrac{\eps q}{2}} - 1)}= \eps e^{-\Omega(\eps\tau|\x|)}$. Hence, \mtrlap{\thresh,\eps,\G} is $(\eps,\eps e^{-\Omega(\eps\tau|\x|)})$-DP.

This completes the proof of \Theoremref{hist-priv-accu}.

\subsection{Proof of \Theoremref{bucketing-hist} -- Bucketed Truncated Laplace Mechanism}\label{sec:bucketHist-priv-accu-proof}
Note that $\mBhist{\alpha,\beta,[0,B)} = \mtrlap{\thresh,\eps,[0,B)} \circ \mbuc{w,[0,B)}$, with $w = 2\beta$ and $\thresh = \frac{\alpha}{t}$, where $t = \lceil \frac{B}{2\beta} \rceil$. 
We will use \Theoremref{compose-DP} to show the DP guarantee and \Theoremref{compose-accuracy} to show the flexible accuracy guarantee of $\mBhist{\alpha,\beta,[0,B)}$.

\paragraph{Differential privacy.}
First note that $\mbuc{w,[0,B)^d}$ is a neighborhood-preserving mechanism w.r.t.\ the neighborhood relation \nhist. This follows because adding/removing any one element changes the output of bucketing by at most one element; hence, neighbors remain neighbors after bucketing. Now, since $\mbuc{w,[0,B)}$ outputs a histogram whose support size is at most $t=\lceil\frac{B}{w}\rceil$, and $\mtrlap{\thresh,\eps,[0,B)}$ on input histograms with support size at most $t$ is $\left(\eps,\eps e^{-\Omega(\eps\tau n)}\right)$-differentially private w.r.t.\ \nhist, it follows from \Theoremref{compose-DP} that $\mBhist{\alpha,\beta,[0,B)}$ is also differentially private w.r.t.\ \nhist with the same parameters.

\paragraph{Flexible accuracy.}
First we show in the following claim that the flexible accuracy guarantee of the bucketing mechanism $\mbuc{w,[0,B)}(\x)$, and we prove it in \Appendixref{bucketHist_proofs}.

\begin{claim}\label{clm:bucket-accuracy}
$\mbuc{w,[0,B)}$ is $\left(0, \frac{w}{2}, 0\right)$-accurate for the identity function $f_{\emph{id}}$ over $\Hspace{[0,B)}$ w.r.t~ the metric $\dhistx$. 
\end{claim}

Note that when we apply \Theoremref{hist-priv-accu} to compute the flexible accuracy parameters of the composed mechanism $\mtrlap{\thresh,\eps,[0,B)} \circ \mbuc{w,[0,B)}$, the parameters of the composed mechanism depend on the distortion sensitivity $\distsens{f_1}{}(\alpha_2)$ and the error sensitivity of $\mtrlap{\thresh,\eps,[0,B)}$. We compute them below. \\

\noindent $\bullet$ {\it Distortion sensitivity of $f_1$:}
Since $f_1$ is the identity function $f_{\text{id}}$ over $\Hspace{[0,B)}$, we have (as noted in the first example in \Sectionref{dist-sens}) that $\distsens{f_1}{}(\alpha_2)\leq\alpha_2$. \\

\noindent $\bullet$ {\it Error sensitivity of $\mtrlap{\thresh,\eps,[0,B)}$:}
Note that the bucketing mechanism $\mbuc{w,[0,B)}:\Hspace{[0,B)}\to\Hspace{[0,B)}$ is a deterministic map, and is $(0,\beta,0)$-accurate (see \Claimref{bucket-accuracy}) for computing the identity function $f_{\text{id}}$, where $\beta=\frac{w}{2}$. As mentioned in the first bullet after the statement of \Theoremref{compose-accuracy}, this implies that when computing the error sensitivity of $\mtrlap{\thresh,\eps,[0,B)}$ (which is required for calculating the output error $\beta$ of the composed mechanism $\mtrlap{\thresh,\eps,[0,B)} \circ \mbuc{w,[0,B)}$), we only need to take supremum in \eqref{eq:err-sens} over point distributions $\x,\x'$ such that $\dhist{\x}{\x'}\leq\beta$, where $\dhist{}{}$ is the metric that we use over $\Hspace{[0,B)}$. In other words, in order to compute the error sensitivity of $\mtrlap{\thresh,[0,B)}$, we only need to bound $\sup_{\x, \x': \\ \dhist{\x}{\x'}\leq \beta} \ \inf_{Y:\\ \dnx(\x',Y) \le \alpha}\Winf{}(\mtrlap{\thresh,\eps,\G}(\x),\prob{Y})$. We bound this in \Lemmaref{hist-error-sens} below.
\begin{lem}\label{lem:hist-error-sens}
For any $\alpha,\beta \geq0$, we have 
\[\tau^{\alpha, 0}_{\mtrlap{\thresh,\eps,[0,B)}, f_\id} (\beta, 0)\quad = \sup_{\substack{\x, \x': \\ \dhist{\x}{\x'}\leq \beta}} \
 \inf_{\substack{Y:\\ \dnx(\x',Y) \le \alpha}}
\Winf{}(\mtrlap{\thresh,\eps,\G}(\x),\prob{Y}) \leq \beta\]
w.r.t.\ the distortion \drop and the metric $\dhistx$. Here, input histograms to the mechanism $\mtrlap{\thresh,\eps,\G}$ are restricted to $t$ bars and $\tau = \alpha/t$. 
\end{lem}

\begin{proof}
For simplicity, we denote $[0,B)$ by $\G$.
For any two histograms $\x,\x'\in\Hspace\G$ such that $\dhist{\x}{\x'} \le \beta$, we will construct a r.v.\ $Y$ over $\Hspace\G$ such that $\dnx_{\text{drop}}(\x',\prob{Y}) \le \alpha$ and $\Winf{}(\mtrlap{\thresh,\eps,\G}(\x),\prob{Y}) \leq \beta$. The claim then immediately follows from this. Details follow.

Consider any two histograms $\x,\x'\in\Hspace\G$ such that $\dhist{\x}{\x'} \le \beta$.
Let $\dG{\cdot}{\cdot}$ denote the underlying metric over $\G$ (consists of $t$ elements) and $|\x|$ denote number of elements in the histogram $\x$. 
By definition of $\dhist\cdot\cdot$, we have $\dhist{\x}{\x'}=\Winf{}(\frac{\x}{|\x|},\frac{\x'}{|\x'|})$.
Let $\phi$ be an optimal coupling of $\frac{\x}{|\x|}$ and $\frac{\x'}{|\x'|}$ such that
\begin{align}
\dhist{\x}{\x'} = \Winf{}(\frac{\x}{|\x|},\frac{\x'}{|\x'|}) = \sup_{\substack{(a, b) \leftarrow \phi}}\dG{a}{b} \leq \beta.
\end{align}

Using $\phi$ we define a transformation $f_{\phi}$, which, when given a histogram $\z$ that is $\alpha$-distorted from $\x$, returns $f_{\phi}(\z)$ that is an $\alpha$-distorted histogram from $\x'$. Recall that for a histogram $\x$ and $a\in\G$, we denote by $\x(a)$ the multiplicity of $a$ in $\x$. Now, for any $b\in[0,B)$, we define $f_{\phi}(\z)(b)$ as follows:
\[f_{\phi}(\z)(b) := |\x'|\sum_{a\in\G} \frac{\z(a) \phi(a , b)}{\x(a)}.\]
The following claim is proved in \Appendixref{bucketHist_proofs}.
\begin{claim}\label{clm:fphi_distort}
For any $\x\in\Hspace\G$, if $\z$ is $\alpha$-distorted from $\x$, then $f_{\phi}(\z)$ is $\alpha$-distorted from $\x'$.
\end{claim}
Recall that $\mtrlap{\thresh,\eps,\G}(\x)$ outputs $\alpha$-distorted histograms from $\x$. This suggests 
defining a r.v.\ $Y:=f_{\phi}\left(\mtrlap{\thresh,\eps,\G}(\x)\right)$ over $\Hspace\G$, whose distribution is given as follows: 
\[\text{For }\y\in\Hspace\G, \text{ define }\Pr[Y=\y]:=\Pr[\mtrlap{\thresh,\eps,\G}(\x)\in f_{\phi}^{-1}(\y)],\]
where $f_{\phi}^{-1}(\y):=\{\z\in\Hspace\G:f_{\phi}(\z)=\y\}$ is the inverse mapping of $f_{\phi}$.

In the following two claims (which we prove in \Appendixref{bucketHist_proofs}), we show that the above defined $Y$ satisfies $\dnx_{\text{drop}}(\x',\prob{Y}) \le \alpha$ and $\Winf{}(\mtrlap{\thresh,\eps,\G}(\x),\prob{Y}) \leq \beta$.
\begin{claim}\label{clm:distortion_bw_xprime-Y}
$\widehat{\partial}_{\emph{drop}}(\x',\prob{Y}) \leq \alpha$.
\end{claim}
\begin{claim}\label{clm:Winf_Y_trLap-bound}
$\Winf{}(\mtrlap{\thresh,\eps,\G}(\x),\prob{Y}) \leq \beta$.
\end{claim}
It follows from \Claimref{distortion_bw_xprime-Y} and \Claimref{Winf_Y_trLap-bound} that $\inf_{Y: \dnx(\x',Y) \le \alpha}\Winf{}(\mtrlap{\thresh,\eps,\G}(\x),\prob{Y}) \leq \beta$. Since this holds for any two histograms $\x,\x'\in\Hspace\G$ such that $\dhist{\x}{\x'}\leq \beta$, we have proved \Lemmaref{hist-error-sens}.
\end{proof}

Now, applying \Theoremref{compose-accuracy} to $\mBhist{\alpha,\beta,[0,B)} = \mtrlap{\thresh,\eps,[0,B)} \circ \mbuc{t,[0,B)}$, we get that $\mBhist{\alpha,\beta,[0,B)}$ is $(\alpha, \beta, 0)$-accurate. 

This completes the proof of \Theoremref{bucketing-hist}.

\subsection{Omitted Proofs from \Sectionref{HBS}}\label{sec:proof_bucketing-general}
In this section, we will prove \Theoremref{bucketing-general}, \Corollaryref{bucketing-max}, and \Corollaryref{bucketing-supp}.
\subsubsection{Proof of \Theoremref{bucketing-general} -- Any Histogram-Based-Statistic}\label{sec:proof_bucketing-general}
First we show the flexible accuracy and then the differential privacy guarantee of our composed mechanism $\M_{\fhbs}^{\alpha,\beta,[0,B)} = \fhbs \circ \mBhist{\alpha,\beta,[0,B)}$.

\paragraph{Flexible accuracy.} 
Note that $\fhbs$ (as a mechanism) for computing $\fhbs$ is $(0,0,0)$-accurate, and we have from \Theoremref{bucketing-hist} that $\mBhist{\alpha,\beta,[0,B)}$ is $(\alpha,\beta,0)$-accurate for the identity function $f_{\id}$ w.r.t.\ the distortion measure \drop and the metric $\dhistx$. 
Applying \Theoremref{compose-accuracy}, we get that $\M_{\fhbs}^{\alpha,\beta,[0,B)}$ is $(\alpha+\distsens{f_{\id}}{}(0),\tau_{\fhbs,\fhbs}^{0,0}(0,\beta),0)$-accurate.
It follows from \eqref{eq:err-sens-deterministic} (by substituting $\M=\fhbs$ as a mechanism for $f=\fhbs$) and the definition of the metric sensitivity \eqref{eq:fhbs-sens}, that $\tau_{\fhbs,\fhbs}^{0,0}(0,\beta)=\Delta_{{\fhbs}}(\beta)$. We have also noted after \eqref{eq:dist-sens} that the distortion sensitivity of any randomized function at zero is equal to zero; in particular, $\distsens{f_{\id}}{}(0)=0$. Substituting these in the flexible accuracy parameters of $\M_{\fhbs}^{\alpha,\beta,[0,B)}$, we get that $\M_{\fhbs}^{\alpha,\beta,[0,B)}$ is $(\alpha,\Delta_{{\fhbs}}(\beta),0)$-accurate for \fhbs w.r.t.\ distortion \drop and metric \dAx.

\paragraph{Differential privacy.} Since $\mBhist{}$ is $\left(\eps,\eps e^{-\Omega(\eps\tau n)}\right)$-DP, and $\M_{\fhbs}^{}$ is a post-processing of $\mBhist{}$, it follows that $\M_{\fhbs}^{}$ is also differentially private with the same parameters.

This completes the proof of \Theoremref{bucketing-general}.

\subsubsection{Proof of \Corollaryref{bucketing-max} -- Computing the Maximum}\label{sec:bucketing-max_proof}
For any two histograms \y, $\y'$, by definition of $\dhist{\y}{\y'}=\Winf{}(\frac{\y}{|\y|},\frac{\y'}{|\y'|})$ and $\fmax$, it follows that 
$|\fmax(\y) -\fmax(\y')| \leq \dhist{\y}{\y'}$. 
Using this in \eqref{eq:fhbs-sens} implies that $\Delta_{\fmax}(\beta) \leq \beta$ for every $\beta\geq0$. 
Then, the corollary follows from \Theoremref{bucketing-general}, with $\fhbs = \fmax$. 

\subsubsection{Proof of \Corollaryref{bucketing-supp} -- Computing the Support}\label{sec:bucketing-supp_proof}
Since $\dsupp{\cS_1}{\cS_2}$ is the difference between the maximum or the minimum elements of $\cS_1$ and $\cS_2$, it follows that for any two histograms \y and $\y'$, we have $\dsupp{\fsupp(\y)}{\fsupp(\y')} \leq \max\{|\fmax(\y) -\fmax(\y')|,|\fmin(\y) -\fmin(\y')|\}$, where $|\fmax(\y) -\fmax(\y')|\leq \dhist{\y}{\y'}$ (from \Corollaryref{bucketing-max}), and similarly, $|\fmin(\y) -\fmin(\y')|\leq \dhist{\y}{\y'}$.
Using this in \eqref{eq:fhbs-sens} implies that $\Delta_{\fsupp}(\beta) \leq \beta$ for every $\beta\geq0$. 
Then, the corollary follows from \Theoremref{bucketing-general}, with $\fhbs = \fsupp$.

\subsection{Proof of \Theoremref{bucketing-general-drmv}}\label{sec:beyond-drop_proofs}
Since $\M_{\fhbs}^{\alpha, \beta, [0,B)}$ is the same mechanism for which the results in \Theoremref{bucketing-hist} hold, the same privacy results as in \Theoremref{bucketing-hist} will also hold here. In the rest of this proof, we prove the flexible accuracy part.

Since $\fhbs$ is a $(0,0,0)$-accurate mechanism for $\fhbs$ (which implies that $\Delta_{\fhbs}(0)=0$), in order to prove the accuracy guarantee of $\M_{\fhbs}^{\alpha, \beta, [0,B)}$, it suffices to show that $\mBhist{\alpha,\beta,[0,B)}$ is $(\alpha+\eta\beta,0,0)$-accurate w.r.t.\ $\drme$.
Note that $\M_{\fhbs}^{\alpha, \beta, [0,B)} = \mtrlap{\thresh,\eps,[0,B)} \circ \mbuc{w,[0,B)}$.
On any input $\x$, first we produce an intermediate bucketed output $\z:=\mbuc{w,[0,B)}(\x)$ and then produce $\y:=\mtrlap{\thresh,\eps,[0,B)}(\z)$ as the final output. 
We have shown in \Claimref{bucket-accuracy} in the proof of \Theoremref{bucketing-hist} that the output $\z$ produced by $\mbuc{w,[0,B)}$ on input $\x$ satisfies $\Winf{}(\x,\z)\leq\beta$. This, by definition of the distortion $\move$, implies $\move(\x,\z)\leq\beta$. 
We have also shown in the proof of \Theoremref{hist-priv-accu} that the output $\y$ produced by $\mtrlap{\thresh,\eps,[0,B)}$ on input $\z$ satisfies $\drop(\z,\y)\leq\alpha$. So, we have $\move(\x,\z)\leq\beta$ and $\drop(\z,\y)\leq\alpha$. This, together with \Lemmaref{drop-move-switch}, implies an existence of a histogram $\s$ such that $\drop(\x,\s)\leq\alpha$ and $\move(\s,\y)\leq\beta$. Using these in the definition of $\dropmove\eta$ in \eqref{eq:drop_move_defn} implies that $\dropmove\eta(\x,\y)\leq\alpha+\eta\beta$. Since we have attributed all the error to the input distortion, we have shown that $\mBhist{\alpha,\beta,[0,B)}$ is $(\alpha+\eta\beta,0,0)$-accurate w.r.t.\ the distortion $\dropmove\eta$.

This completes the proof of \Theoremref{bucketing-general-drmv}.

\section{Experimental Evaluations}\label{sec:eval}
We empirically compare our basic mechanism \mtrlap{\thresh,\eps,\G} (\Algorithmref{hist-mech})
on a ground set $\G=\{1,\cdots,B\}$, against various
competing mechanisms, for accuracy on a few histogram-based statistics
computed on it. We plot average
errors (actual and flexible), on different histogram distributions%
\footnote{For each data distribution, the plots were averaged over 100 data sets, with 100 runs each for each mechanism.}
for functions $\max_k(\x) := \max \{ i \mid \x(i) \ge k \}$, $\max:=\max_1$,
and $\mode(\x) := \arg\max_i \x(i)$; note that $\mode(\x)$ is equal to the most frequently occurring data item in $\x$. 
The parameters for \mtrlap{\thresh,\eps,\G} that we will use in the section are given in \Corollaryref{simpler-dp-hist}.

We emphasize that the plots are only indicative of the performance of our algorithm on specific histograms, and do not suggest \emph{worst-case accuracy guarantees}. On the other hand, our theorems do provide worst-case accuracy guarantees. 

We will empirically compare our results against the Exponential Mechanism \cite{ExponentialMech}, Propose-Test-Release Mechanism \cite{DworkL09}, Smooth-sensitivity Mechanism \cite{NissimRS07}, Stability-Based Sanitized Histogram \cite{BNS}, and Choosing-Based Histogram Mechanism \cite{BeimelNiSt16}.
First we present the comparison of our mechanisms against all these on different data distributions in \Sectionref{evals-carried-out} and then described these mechanisms briefly in \Sectionref{compared-mechs}.
We point out one notable omission from our plots: the Encode-Shuffle-Analyze
histogram mechanism \cite{erlingsson2020encode}, which appeared
independently and concurrently to our mechanism,%
\footnote{Preliminary versions of the current work were available online and had been presented publicly (as an invited talk \cite{prelimversion19}) before \cite{erlingsson2020encode} was available.}
also uses a shifted (but not truncated) Laplace mechanism, and in all the
examples plotted, yields a behavior that is virtually identical to our
mechanism's. However, we emphasize that \cite{erlingsson2020encode} claim
accuracy only for the histogram itself, and indeed, for the functions that
we consider, it does not enjoy the \emph{worst-case accuracy guarantees}
that we provide.

\begin{figure}
\centering
 \includegraphics[scale=0.55]{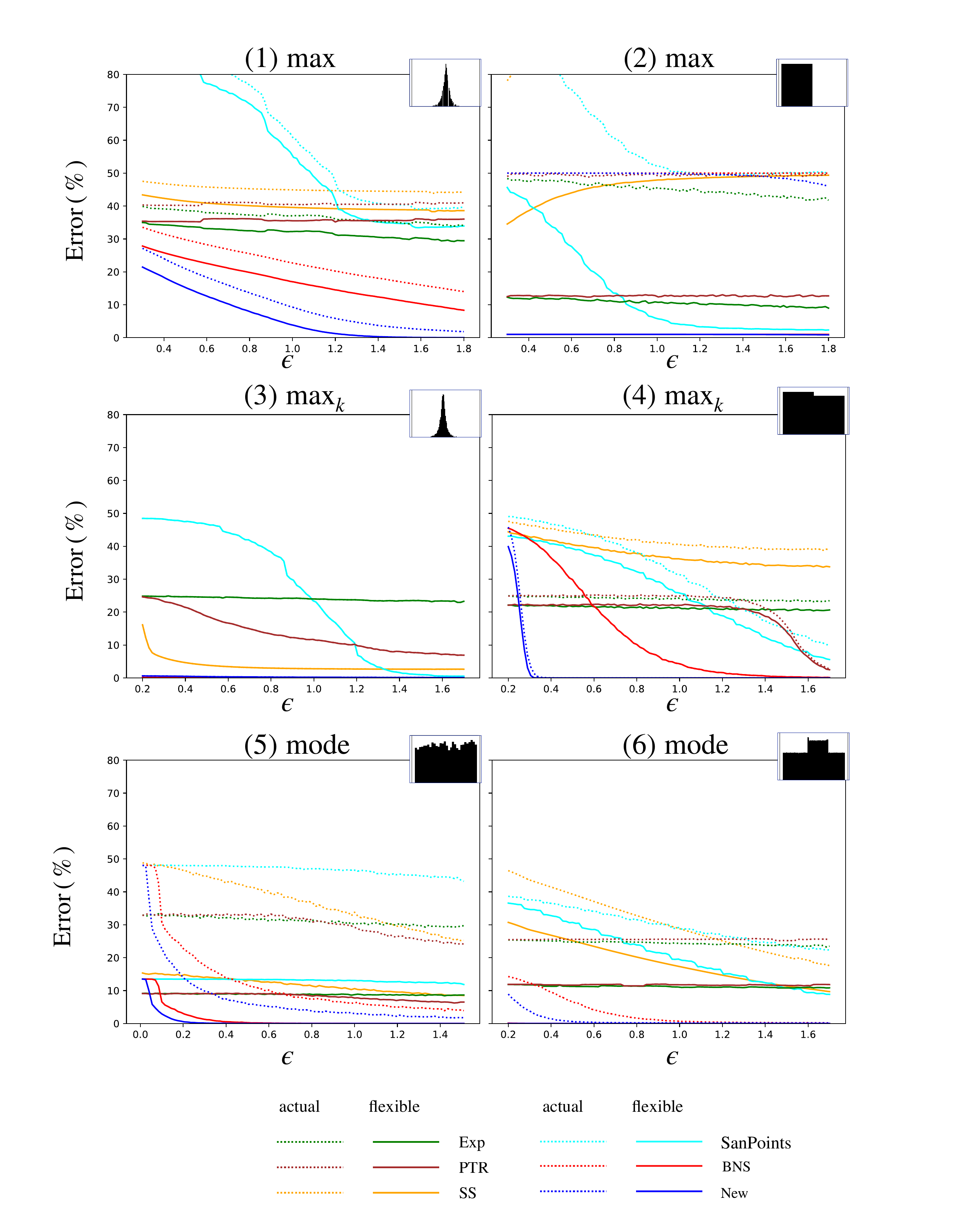}
\caption{For each evaluation, a typical histogram used is shown in inset. 
The different data distributions elicit a variety of behaviors of the
different mechanisms. Experiment (2) shows an instance which is hard for all
the mechanisms without considering flexible accuracy; on the other hand,
in Experiment (3), flexible accuracy makes no difference (the plots
overlap). In these two experiments, BNS and our new mechanism match each
other. In all the other experiments, our new mechanism dominates the others,
with or without considering flexible accuracy.
\label{fig:app-eval}}
\end{figure}

\subsection{Evaluations Carried Out}\label{sec:evals-carried-out}

In each of the following empirical evaluations, a histogram distribution and
one of the following functions were fixed: $\max_k(\x) := \max \{ i \mid
\x(i) \ge k \}$, $\max:=\max_1$, and $\mode(\x) := \arg\max_i \x(i)$.

\begin{enumerate}
\item[\textbf{(1)}] Function $\max$. Histogram of about 10,000 items drawn i.i.d.\ from a Cauchy distribution with median $45$ and scale $4$, restricted to 100 bars, with the last 10 set to empty bars. 

\item[\textbf{(2)}] Function $\max$. Step histogram with two steps (height $\times$ width) : [$1000 \times 50$, $1 \times 50$].

\item[\textbf{(3)}] Function $\max_{500}$. Same histogram distribution as in (1) above,
but without zeroing out the right-most bars.

\item[\textbf{(4)}] Function $\max_{500}$. Step histogram with 100 bars, with two steps (height $\times$ width) : [$540 \times 50$, $490 \times 50$]. 

\item[\textbf{(5)}] Function \mode. Histogram of 30 bars, each bar has height drawn from i.i.d.\ Poisson with mean 250.

\item[\textbf{(6)}] Function \mode. Noisy step histogram, with steps 
[$130 \times 120$, $200 \times 5$,  $185 \times 85$, $190 \times 10$, $130\times80$].

\end{enumerate}

The results are shown in \Figureref{app-eval}. In each experiment, a range of values for $\epsilon$ are chosen, while we fixed $\delta = 2^{-20}$. Errors are shown in the y-axis as a percentage of the full range $[0,B)$. In all experiments, for each mechanism we also compute flexible accuracy allowing distortion of $\drop=0.005$. 

\subsection{Description of the Compared Mechanisms}\label{sec:compared-mechs}

\paragraph{Exponential Mechanism.}
The Exponential Mechanism \cite{ExponentialMech} can be tailored for an
abstract utility function. We consider the negative of the error as the utility of a response $y$ on input histogram \x, i.e., $q(\x,y)=-\err(\x,y)=-|\max(\x)-y|$.
However, for both $\max_k$ and \mode, error has high sensitivity -- changing
a single element in the histogram can change the error by as much as the
number of bars in the histogram. Since the mechanism produces an output $r$
with probability proportional to $e ^ {\frac{\epsilon q(\x, r)}{2
\Delta_\err}}$, where $\Delta_\err$ is the sensitivity of $\err(\cdot,\cdot)$, having a
large sensitivity has the effect of moving the output distribution close to
a uniform distribution. This is reflected in the performance of this
mechanism in all our plots.

\paragraph{Propose-Test-Release Mechanism (PTR).}
We consider the commonly used form of the PTR mechanism of Dwork and Lei
\cite{DworkL09}, namely, ``releasing stable values'' (see \cite[Section 3.3]{VadhanSurvey}). On input \x, the
mechanism either releases the correct result $f(\x)$ or refuses to do so (replacing it with a random output value), depending on whether the radius of the neighborhood of \x where it remains constant is sufficiently large (after adding some noise). For computing a function $f$ and a setting of parameter $\beta = 0$ and privacy parameters $\epsilon, \delta$, the mechanism calculates this radius for an input $\x$ as,
$r = d(\x, \{\x' : \text{LS}_f(\x') > 0\}) + \text{Lap}(\nicefrac{1}{\epsilon})$, where $d(\x,\mathcal{S})$ is the minimum Hamming distance between $\x$ and any point in the set $\mathcal{S}$ and $\text{LS}_f(\y):=\max\{|f(\y)-f(\z)|:\y\sim\z\}$ is local sensitivity of function $f$ at $\y$. 
If this radius $r$ is greater than $\nicefrac{\ln{(\nicefrac{1}{\delta})}}{\epsilon}$, the
mechanism will output the exact answer $f(\x)$, otherwise it outputs a random value from the domain. 
For the functions we consider, this radius of stable region can be computed efficiently and is typically small or even zero for input distributions considered which is reflected in our plots. 

\paragraph{Smooth-sensitivity Mechanism (SS).}
This mechanism, due to Nissim et al.\ \cite{NissimRS07}, uses
the smooth sensitivity of a function $f$, defined as $SS_{f}^{\epsilon}(\x) = \max \{ LS_f(\x') e^{-\epsilon d(\x, \x')} | \x' \in \Hspace\G \}$, where $LS_f(\x')$ denotes the \emph{local sensitivity} of $f$ at $\x'$, and $d(\cdot,\cdot)$ is the Hamming distance. Given an input histogram \x, the mechanism adds noise $O(SS_f^{\beta}(\x)/\alpha)$ to $f(x)$ for appropriate values of $\alpha$ and $\beta$ to obtain $(\epsilon, \delta)$-DP. 
For functions like $\max_k$ and \mode, like sensitivity, local sensitvity (and hence smooth sensitivity) also tends to
be large on many histograms, which leads this mechanism to add a large
noise.

\paragraph{Stability-Based Sanitized Histogram Mechanism.}
This mechanism was proposed by Bun et al. \cite{BNS} (also see \cite[Theorem 3.5]{VadhanSurvey}) for
releasing histograms with provable worst-case guarantees. However, these
guarantees are in terms of the errors in the individual bar heights of the
histogram, and does not necessarily translate to the histogram based functions, as we
consider. Nevertheless, this mechanism provides a potential candidate for
a mechanism for any histogram based statistic.

For each bar of the histogram, the mechanism adds Laplace noise to the bar
height, and the resulting value is reported only if it is more than a
threshold, and
otherwise a $0$ is reported. By treating empty bars differently, this
mechanism achieves comparable flexible accuracy as our mechanism in the case
of $\max$. However, this does not generalize to $\max_k$. In particular, in
the example in (4) in \Figureref{app-eval}, by adding (possibly positive)
noise to histogram bars of height lower than $k$, the mechanism is very
likely to find a bar which is much further to the right than the point where
the bar heights cross $k$.

\paragraph{Choosing-Based Histogram Mechanism.}
Beimel et al.~\cite{BeimelNiSt16} presented a mechanism SanPoints for
producing a sanitized histogram, with formal PAC-guarantees for the height
of each bar of the histogram. The mechanism involves iteratively choosing
bars from the histogram, without replacement, and adding some noise
to the bar heights. The bars are chosen according to a DP mechanism for picking the tallest bar
(which in turn uses the exponential mechanism).

For the functions we consider, SanPoints yields mixed results, but is
dominated by BNS and our new mechanism.

\section*{Acknowledgements}
The work of Deepesh Data was supported in part by NSF grants \#1740047, \#2007714, and UC-NL grant LFR-18-548554.
The work of Manoj Prabhakaran was supported in part by the Joint Indo-Israel Project DST/INT/ISR/P-16/2017 and the Ramanujan Fellowship of Dept. of Science and Technology, India.

\bibliographystyle{alpha}
\bibliography{bib}

\newpage
\appendix
\section{Details Omitted from \Sectionref{lossy-wass}}
\subsection{Lossy $\infty$-Wasserstein Distance}\label{app:wasserstein}
\begin{lem}[\Lemmaref{wass-triangle} at $\gamma_1=\gamma_2=0$]\label{lem:Winf_triangle}
For distributions $P$, $Q$, and $R$ over a metric space $(\Omega,\met)$, we have 
\begin{align*}
\Winf{}(P, R) &\leq \Winf{}(P, Q) + \Winf{}(Q, R). 
\end{align*}
\end{lem}
\begin{proof}
Let $\phi_2\in\Phi(P, Q)$ and $\phi_3\in\Phi(Q, R)$ denote the optimal couplings for $\Winf{}(P, Q)$ and $\Winf{}(Q, R)$, respectively, i.e., $\Winf{}(P, Q)=\sup_{\substack{(x,y):\\\phi_2(x,y)\neq 0}}\met(x,y)$ and $\Winf{}(Q,R)=\sup_{\substack{(y,z):\\\phi_3(y,z)\neq 0}}\met(y,z)$.
It follows from the Gluing Lemma \cite{Villani_OptimalTransport08} that we can find a coupling $\phi'$ over $\Omega\times \Omega\times \Omega$ such that the projection of $\phi'$ onto its first two coordinates is equal to $\phi_2$ and its last two coordinates is equal to $\phi_3$.
Let $\phi_1$ denote the projection of $\phi'$ onto its first and the third coordinates. Note that $\phi_1\in\Phi(P, R)$, but it may not be an optimal coupling for $\Winf{}(P, R)$. 
Now the triangle inequality follows from the following set of inequalities:
\begin{align*}
\Winf{}(P, R) &= \inf_{\phi\in\Phi(P,R)} \sup_{\substack{(x,z):\\\phi(x,z)\neq 0}}\met(x,z) 
\quad\leq \sup_{\substack{(x,z):\\\phi_1(x,z)\neq 0}}\met(x,z) 
\quad= \sup_{\substack{(x,y,z):\\\phi'(x,y,z)\neq 0}}\met(x,z) \\
&\stackrel{\text{(a)}}{\leq} \sup_{\substack{(x,y,z):\\\phi'(x,y,z)\neq 0}}\met(x,y) + \met(y,z) \\
&= \sup_{\substack{(x,y,z):\\\phi'(x,y,z)\neq 0}}\met(x,y) \quad+ \sup_{\substack{(x,y,z):\\\phi'(x,y,z)\neq 0}} \met(y,z) \\
&= \sup_{\substack{(x,y):\\\phi_2(x,y)\neq 0}}\met(x,y) \quad+ \sup_{\substack{(y,z):\\\phi_3(y,z)\neq 0}} \met(y,z) \\
&= \Winf{}(P, Q) + \Winf{}(Q, R),
\end{align*}
where (a) follows from the fact that $\met$ is a metric, and so it satisfies the triangle inequality.
\end{proof}

\begin{claim}\label{clm:TV_dist_Q-Qprime}
$\Delta(Q, Q')\leq\gamma-\gamma_1$.
\end{claim}
\begin{proof}
The claim follows from the following set of inequalities.
\begin{align}
    \Delta(Q, Q') &\le \Delta(Q, Q_{opt}) + \Delta(Q_{opt}, Q') \notag \\
    &\le \gamma - \gamma_{opt} + \frac{1}{2}\int_{\Omega}\left|\Prob{Q_{opt}}{y} - \Prob{Q'}{y}\right|\dd y \tag{Since $\Delta(Q, Q_{opt})\leq\gamma - \gamma_{opt}$} \notag \\
    &= \frac{1}{2}\int_{\Omega} \left|\int_{\Omega}\phi_{opt}(x, y)\dd x -  \int_{\Omega}\phi'(x, y)\dd x\right| \dd y + (\gamma - \gamma_{opt}) \notag \\
    &\le \frac{1}{2}\int_{\Omega} \int_{\Omega} \left|\phi_{opt}(x, y) - \phi'(x, y)\right|\dd x\dd y + (\gamma - \gamma_{opt}) \notag
    \end{align}
    Define $\Omega_1:=\{x\in\Omega:\Prob{P_{opt}}{x}>0\}$ and $\overline{\Omega}_1:=\Omega\setminus\Omega_1$. 
    Since $\Prob{P_{opt}}{x} = 0$ for all $x\in\overline{\Omega}_1$ and $\Prob{{P}_{opt}}$ is the first marginal of $\phi_{opt}$, we have that $\phi_{opt}(x,y)=0$ for all $x\in\overline{\Omega}_1$ and $y\in\Omega$.
     Now, continuing from above, we get
    \begin{align}
    \Delta(Q, Q') &\leq \frac{1}{2}\int_{\Omega} \int_{x\in\Omega_1} \left|\phi_{opt}(x, y) - \phi'(x, y)\right|\dd x\dd y  + \frac{1}{2}\int_{\Omega} \int_{x\in\overline{\Omega}_1} \left|\phi_{opt}(x, y) - \phi'(x, y)\right|\dd x\dd y + (\gamma - \gamma_{opt}) \notag \\ 
    &= \frac{1}{2}\int_{\Omega}\int_{\Omega_1}\phi_{opt}(x, y)\left|1 - \frac{\Prob{P'}{x}}{\Prob{P_{opt}}{x}}\right|\dd x\dd y + \frac{1}{2}\int_{\Omega}\int_{\overline{\Omega}_1}\left|\phi'(x, y)\right|\dd x\dd y + (\gamma - \gamma_{opt}) \notag \\
    &= \frac{1}{2}\int_{\Omega_1}\left|1 - \frac{\Prob{P'}{x}}{\Prob{P_{opt}}{x}}\right|\dd x \int_{\Omega}\phi_{opt}(x, y)\dd y + \frac{1}{2}\int_{\Omega}\int_{\overline{\Omega}_1}\Prob{P'}{x}\delta(x -  y)\dd x\dd y + (\gamma - \gamma_{opt}) \tag{Since $\phi'(x, y) = \Prob{P'}{x}\delta(x -  y)$ for $x\in\overline{\Omega}_1$} \\
    &= \frac{1}{2}\int_{\Omega_1}\left|\Prob{P_{opt}}{x} - \Prob{P'}{x}\right|\dd x + \frac{1}{2}\int_{\overline{\Omega}_1}\Prob{P'}{x}\dd x + (\gamma - \gamma_{opt}) \tag{Since $\int_{\Omega}\phi_{opt}(x, y)\dd y = \Prob{P_{opt}}{x}$ and $\int_{\Omega}\delta(x-y)\dd y = 1$ for any $x$} \\
    &= \frac{1}{2}\int_{\Omega_1}\left|\Prob{P_{opt}}{x} - \Prob{P'}{x}\right|\dd x + \frac{1}{2}\int_{\overline{\Omega}_1}\left|\Prob{P_{opt}}{x} - \Prob{P'}{x}\right|\dd x + (\gamma - \gamma_{opt}) \tag{Since $\Prob{P_{opt}}{x}=0$ whenever $x\in\overline{\Omega}_1$} \notag \\
    &= \frac{1}{2}\int_{\Omega}\left|\Prob{P_{opt}}{x} - \Prob{P'}{x}\right|\dd x + (\gamma - \gamma_{opt}) \notag \\
        &\stackrel{\text{(a)}}{=} \frac{1}{2}\int_{\Omega}\left|\Big(1-\frac{\gamma_1}{\gamma_{opt}}\Big)R_{opt}(x)\right|\dd x + (\gamma - \gamma_{opt}) \notag \\
        &= \frac{(\gamma_{opt}-\gamma_1)}{2\gamma_{opt}} \int_{\Omega}\left|R_{opt}(x)\right|\dd x + (\gamma - \gamma_{opt}) \notag \\
        &= \gamma_{opt} - \gamma_1 \tag{Since $\int_{\Omega} |R_{opt}(\omega)|\dd\omega= 2\gamma_{opt}$} + (\gamma - \gamma_{opt}) \notag \\
        &= \gamma - \gamma_1
\end{align}
Here (a) follows because for every $x\in\Omega$, we have $\Prob{P_{opt}}{x} - \Prob{P'}{x} = R_{opt}(x) + \Prob{P}{x} - \Prob{P'}{x} = R_{opt}(x) - R'(x) = R_{opt}(x) - \frac{\gamma_1}{\gamma_{opt}}R_{opt}(x)$.
\end{proof}

\begin{claim*}[Restating \Claimref{wass_alternate}]
For distributions $P$ and $Q$ over a metric space $(\Omega,\met)$ and $\gamma\in [0,1]$, we have
\begin{align*}
\Winf{\gamma}(P,Q) \quad= \displaystyle \inf_{\substack{\hat{P},\hat{Q}:\\ \Delta(P,\hat{P}) + \Delta(Q,\hat{Q})\leq\gamma}} \Winf{} (\hat{P},\hat{Q}).
\end{align*}
\end{claim*}
\begin{proof}
This claim simply follows by by viewing the infimum set in the definition of $\gamma$-Lossy $\infty$-Wasserstein distance differently.
\begin{align}
\Winf{\gamma}(P,Q) \quad &\stackrel{\text{(a)}}{=} \displaystyle \inf_{\phi\in\Phi^{\gamma}(P,Q)} \max_{(x,y)\leftarrow\phi}\met(x, y) \notag \\
&\stackrel{\text{(b)}}{=} \displaystyle \inf_{\substack{\hat{P},\hat{Q}:\\ \Delta(P,\hat{P}) + \Delta(Q,\hat{Q})\leq\gamma}} \displaystyle \inf_{\phi\in\Phi^0(\hat{P},\hat{Q})} \max_{(x,y)\leftarrow\phi}\met(x, y) \notag \\
&\stackrel{\text{(c)}}{=} \displaystyle \inf_{\substack{\hat{P},\hat{Q}:\\ \Delta(P,\hat{P}) + \Delta(Q,\hat{Q})\leq\gamma}} \Winf{} (\hat{P},\hat{Q}). \notag
\end{align}
where (a) follows from the definition of $\gamma$-Lossy $\infty$-Wasserstein distance; 
(b) trivially holds by viewing the infimum set differently; in (c) we substituted the definition of $W_{\infty}$;
and (d) follows because $P',Q'$ satisfies $\Delta(P,P') + \Delta(Q,Q')\leq\gamma$.
\end{proof}

\subsection{Average Version of Lossy Wasserstein Distance}\label{app:average_lossy-wass}
Our definition of \Winf\theta uses a worst case notion of distance.
Many of the results using this notion have analogues using an average
case version. We formally present this definition below, as it may be of interest
elsewhere. 
\begin{defn}[$\theta$-Lossy Average Wasserstein Distance]\label{def:gamma-wass-dist}
Let $P$ and $Q$ be two probability distributions over a metric space
$(\Omega,\met)$, and let $\theta\in[0,1]$.
The \emph{$\theta$-lossy average Wasserstein distance} between $P$ and $Q$ is defined as:
\begin{equation}\label{eq:gamma-wass-dist}
\W{\theta}(P, Q) = \inf_{\phi\in\Phi^{\theta}(P, Q)}\E_{(x,y)\leftarrow\phi}[\met(x,y)]. 
\end{equation}
\end{defn}

The following lemma relates lossy average Wasserstein and lossy $\infty$-Wasserstein
distances.
\begin{lem}\label{lem:W-Winf}
For any two distributions $P, Q$, and $0 \le \beta' < \beta \le 1$,
\[ \W{\beta}(P, Q) \le \Winf{\beta}(P, Q) \le \frac{\W{\beta'}(P, Q)}{(\beta-\beta')}. \]
\end{lem}
\begin{proof}
Clearly from the definitions, $\W{\beta}(P, Q) \le \Winf{\beta}(P, Q)$. 

Suppose $\W{\beta'}(P, Q)=\gamma$ and $\phi\in\Phi^{\beta'}(P, Q)$ is an optimal
coupling that realizes this. Then, in $\phi$, the total mass that is
transported more than a distance $\gamma'$ is at most $\gamma/\gamma'$ and
the total mass that is lost is at most $\beta'$. By choosing to simply not
transport this mass at all, one loses $\beta'+\gamma/\gamma'$ mass, but no
mass is transported more than a distance $\gamma'$.  Choosing
$\gamma'=\gamma/(\beta-\beta')$ this upper bound on loss is $\beta$, and
hence this modified coupling shows that $\Winf{\beta}(P, Q) \le \gamma'$.
\end{proof}

\subsection{$\gamma$-Lossy $\infty$-Wasserstein Distance Generalizes Existing Notions}\label{app:wass-generalizes}
\begin{lem}\label{lem:generalize-pac}
Let $(\Omega,\met)$ be a metric space. Let $F_f$ be a point distribution on some $f\in\Omega$ and $G$ be a distribution over $\Omega$. Then for any $\gamma\in[0,1]$ and $\beta\geq0$, we have
\[\Winf{\gamma}(F_f,G)\leq\beta\quad \Longleftrightarrow \quad \Pr_{g\leftarrow G}[\met(f,g) > \beta] \le \gamma.\]
\end{lem}
\begin{proof}
We show both the directions below.
\begin{itemize}
\item {\bf Only if part ($\Rightarrow$):} Suppose $\Winf{\gamma}(F_f,G)\leq\beta$. It follows from \Lemmaref{wass-marginal-loss} that there exists a distribution $G'$ such that $\Delta(G',G)\leq\gamma$ and $\Winf{}(F_f,G')\leq\beta$. Since $F_f$ is a point distribution, all couplings $\phi\in\Phi^0(F_f,G')$ will be such that $\phi_1=F_f$ and $\phi_2=G'$, which implies that $\Winf{}(F_f,G')=\sup_{g'\leftarrow G'}\met(f,g')\leq\beta$. 
Now we show that, together with $\Delta(G',G)\leq\gamma$, this implies $\Pr_{g\leftarrow G}[\met(f,g)>\beta]\leq\gamma$:
\begin{align*}
\Pr_{g\leftarrow G}[\met(f,g)>\beta] &= \underbrace{\Pr_{g\leftarrow G}[\met(f,g)>\beta\ |\ g\in\support(G')]}_{=\ 0}\Pr_{g\leftarrow G}[g\in\support(G')] \notag \\
&\hspace{2cm} + \underbrace{\Pr_{g\leftarrow G}[\met(f,g)>\beta\ |\ g\notin\support(G')]}_{\leq\ 1}\Pr_{g\leftarrow G}[g\notin\support(G')] \notag \\
&\leq \Pr_{g\leftarrow G}[g\notin\support(G')] \notag \\
&= \int_{g\in\Omega:\ p_{G}(g)>0\ \&\ p_{G'}(g)=0}p_G(g)dg \notag \\
&= \int_{g\in\Omega:\ p_{G}(g)>0\ \&\ p_{G'}(g)=0}(p_G(g)-p_{G'}(g))dg \notag \\
&\stackrel{\text{(a)}}{\leq} \int_{g\in\Omega:\ p_G(g)>p_{G'}(g)}(p_G(g)-p_{G'}(g))dg \notag \\
&\stackrel{\text{(b)}}{=} \Delta(G,G') \leq \gamma,
\end{align*}
where (a) follows because $\{g\in\Omega:p_{G}(g)>0\ \&\ p_{G'}(g)=0\}\subseteq\{g\in\Omega:p_G(g)>p_{G'}(g)\}$, and (b) follows from the reasoning given below.

Define $\Omega_G^+ := \{g\in\Omega:p_G(g)>p_{G'}(g)\}$ and $\Omega_G^- := \{g\in\Omega:p_G(g)<p_{G'}(g)\}$. 
Since $\int_{g\in\Omega}p_G(g)dg = \int_{g\in\Omega}p_{G'}(g)dg$, it follows that $\int_{g\in\Omega_G^+}(p_G(g)-p_{G'}(g))dg = \int_{g\in\Omega_G^-}(p_{G'}(g)-p_G(g))dg$. 
Substituting this in the definition of $\Delta(G,G')$, we get $\Delta(G,G')=\int_{g\in\Omega_G^+}(p_G(g)-p_{G'}(g))dg$. 
\item {\bf If part ($\Leftarrow$):} Suppose $\Pr_{g\leftarrow G}[\met(f,g) > \beta] \le \gamma$. Let $\Omega'=\{g\in\Omega:\met(f,g)\leq\beta\}$ and $G'$ be a distribution supported on $\Omega'$ such that $p_{G'}(g)=\frac{1}{\eta}p_G(g)$ when $g\in\Omega'$, otherwise $p_{G'}(g)=0$. Here $\eta=\int_{g\in\Omega'}p_G(g)dg\geq(1-\gamma)$ is the normalizing constant.
First we show that $\Delta(G,G')\leq\gamma$.
\begin{align}
\Delta(G,G') &= \frac{1}{2}\int_{g\in\Omega}|p_{G'}(g)-p_G(g)|dg \notag \\
&= \frac{1}{2}\int_{g\in\Omega'}|p_{G'}(g)-p_G(g)|dg + \frac{1}{2}\int_{g\in\Omega\setminus\Omega'}p_G(g)dg \tag{Since $p_{G'}(g)=0$ when $g\in\Omega\setminus\Omega'$} \\
&= \frac{1}{2}\int_{g\in\Omega'}p_G(g)(\frac{1}{\eta}-1)dg + \frac{1}{2}\int_{g\in\Omega\setminus\Omega'}p_G(g)dg \notag \\
&= \frac{1}{2}(\frac{1}{\eta}-1)\eta + \frac{1}{2}(1-\eta) \tag{Since $\int_{g\in\Omega'}p_G(g)dg=\eta$} \\
&= 1-\eta \leq \gamma.
\end{align}
Now define a joint distribution $\phi$, whose first marginal is the point distribution $F_f$ and the second marginal is $G'$, which implies that $\sup_{(x,y)\leftarrow\phi}\met(x,y)=\sup_{g'\in\Omega'}\met(f,g')$. It follows from the argument above that $\phi\in\Phi^{\gamma}(F_f,G)$, which implies that $\Winf{\gamma}(F_f,G)\leq\sup_{(x,y)\leftarrow\phi}\met(x,y)=\sup_{g'\in\Omega'}\met(f,g')\leq\beta$, where the last inequality is by definition of $\Omega'$. Hence, we get $\Winf{\gamma}(F_f,G)\leq\beta$.
\end{itemize}
This completes the proof of \Lemmaref{generalize-pac}.
\end{proof}

\begin{lem}\label{lem:generalizing-tv}
For any two distributions $P,Q$ over a metric space $(\Omega,\met)$ and $\gamma\in[0,1]$, we have
\[\Winf{\gamma}(P,Q)=0 \quad \Longleftrightarrow \quad \Delta(P,Q)\leq\gamma.\]
\end{lem}
\begin{proof} We show both the directions below.
\begin{itemize}
\item {\bf Only if part ($\Rightarrow$):} Suppose $\Winf{\gamma}(P,Q)=0$. This implies that there exists a joint distribution $\phi\in\Phi^{\gamma}(P,Q)$ such that $\sup_{(x,y)\leftarrow\phi}\met(x,y)=0$. Since $\met$ is a metric, this implies that for all $(x,y)\leftarrow\phi$, we have $x=y$. Hence, the first marginal $\phi_1$ and the second marginal $\phi_2$ of $\phi$ are equal, which implies that $\Delta(\phi_1,P)+\Delta(\phi_2,Q)\leq\gamma$. Then, by triangle inequality and that $\phi_1=\phi_2$, we get $\Delta(P,Q)\leq\gamma$.

\item {\bf If part ($\Leftarrow$):} Suppose $\Delta(P,Q)\leq\gamma$. Define a joint distribution $\phi:=P\times P$. Since $\phi_1=\phi_2=P$, we have $\phi\in\Phi^{\gamma}(P,Q)$. This, by definition, implies $\Winf{\gamma}(P,Q)\leq\sup_{(x,y)\leftarrow\phi}\met(x,y)$. Since both the marginals of $\phi$ are the same, we have $\met(x,y)=0$ for every $(x,y)\leftarrow\phi$. This, by  the non-negativity of $\Winf{\gamma}(P,Q)$, gives $\Winf{\gamma}(P,Q)=0$. 
\end{itemize}
\end{proof}

\section{Details Omitted from \Sectionref{defn-fa} -- Usefulness~\cite{BLR} vs.\ Flexible Accuracy}\label{app:Comparison_BLR}
To express accuracy
guarantees of their mechanisms, Blum et al.~\cite{BLR} introduced a notion of
\emph{$(\beta,\gamma,\psi)$-usefulness} that parallels
$(\alpha,\beta,\gamma)$-accuracy, except that $\psi$ measures perturbation
of the function rather than input distortion.  Note that this is a reasonable
notion for the function classes they considered (half-space queries, range queries etc.),
but it is not applicable to queries like maximum.

Flexible accuracy generalizes the notion of usefulness.
Firstly, mechanisms which are
$(\beta,\gamma,0)$-useful are $(0,\beta,\gamma)$-accurate (in \cite{BLR},
such mechanisms were given for interval queries). But even general
usefulness can be translated to flexible accuracy generically, by redefining
the function to have an extra input parameter that specifies perturbation.
Further, the specific $(\beta,\gamma,\psi)$-useful DP mechanism of
\cite{BLR} for half-space counting queries -- with data points on a unit
sphere, and the perturbation of the function corresponded to rotating the
half-space by $\psi$ radians -- is $(\psi,\beta,\gamma)$-accurate for the
same functions, w.r.t.\ the distortion \move. This is because, the rotation
of the half-space can be modeled as moving all the points on the unit sphere
by a distance of at most $\psi$.

\section{Details Omitted from \Sectionref{compostion_FA} -- Proof of \Lemmaref{err-sens-deterministic}}\label{app:err-sens-deterministic_proof}
For convenience, we write the lemma statement below.
\begin{lem*}[Restating \Lemmaref{err-sens-deterministic}]
Let $\M:\B\to\C$ be a deterministic mechanism for a deterministic function $f:\B\to\C$. Then, for any $\beta_1\geq0$, we have
\begin{align*} 
\tau_{\M,f}^{0,0}(\beta_1,0) \quad = \sup_{\substack{X,X': \\ \Winf{}(\prob{X},\prob{X'})\leq \beta_1}} \Winf{}(\M(X),f(X')) \quad = \sup_{\substack{x,x'\in\A : \\ \dB{x}{x'} \leq \beta_1}} \dC{\M(x)}{f(x')}.
\end{align*}
\end{lem*}
\begin{proof}
The first equality follows from the definition of error sensitivity. We only need to prove the second equality.
\begin{itemize}
\item {\it LHS $\geq$ RHS:} This is the easy part.
\begin{align*}
    \sup_{\substack{X,X': \\ \Winf{}(\prob{X},\prob{X'})\leq \beta_1}} \Winf{}(\M(X),f(X')) \ \ \ge \sup_{\substack{x,x'\in\B: \\ \Winf{}(\prob{x},\prob{x'})\leq \beta_1}} \Winf{}(\M(x),f(x')) \ \
    = \sup_{\substack{x,x'\in\B: \\ \dB{\prob{x}}{\prob{x'}}\leq \beta_1}} \dC{\M(x)}{f(x')},
\end{align*}
where the inequality holds because considering only point distributions restricts the set over which we take supremum and the equality holds because the $\infty$-Wasserstein distance between any two point distributions in any metric is just the distance between the points on which the distributions are supported in that metric.

\item {\it LHS $\leq$ RHS:}
Consider any two distributions $\prob{X}, \prob{X'}$ over $\B$ s.t.\ $\Winf{}(\prob{X},\prob{X'})\leq \beta$. Let $\phi_1$ be the optimal coupling between $\prob{X}, \prob{X'}$ such that 
\[\Winf{}(\prob{X},\prob{X'}) \quad = \sup_{(\x, \x') \leftarrow \phi_1} \dB{\x}{\x'} \leq \beta_1. \]
Using $\phi_1,\M,f$, we define a joint distribution $\phi_2$ over $\C\times\C$ as follows: For any ${\bf a,b}\in\C$, define
\[\phi_2({\bf a}, {\bf b}) \quad := \sum_{\substack{\x, \x' : \\
\M(\x) = {\bf a}, f(\x') = {\bf b}}} \phi_1(\x, \x').\]
It can be verified that $\phi_2\in\Phi(\M(X),f(X'))$, i.e., $\phi_2$ is a valid coupling between $\M(X), f(X')$.
Now
\begin{align*}
\Winf{}(\M(X),f(X')) 
\ \ \leq \sup_{({\bf a},{\bf b})\leftarrow\phi_2}\dC{\bf a}{\bf b} 
\ \ = \sup_{(\x,\x')\leftarrow\phi_1}\dC{\M(\x)}{f(\x')} 
\ \ \leq \sup_{\substack{\x,\x'\in\B: \\ \dB{\x}{\x'}\leq \beta_1}}\dC{\M(\x)}{f(\x')},
\end{align*}
where the last inequality holds because $\{(\x,\x'):(\x,\x')\leftarrow\phi_1\}\subseteq\{(\x,\x'):\dB{\x}{\x'}\leq\beta_1\}$. 

Note that the RHS of the last inequality does not depend on $X,X'$. So, taking supremum over all distributions $X, X'$ such that $\Winf{}(\prob{X},\prob{X'})\leq\beta_1$ gives the required result.
\end{itemize}
This completes the proof of \Lemmaref{err-sens-deterministic}.
\end{proof}

\section{Omitted Details from \Sectionref{compose-accuracy_proof} -- Flexible Accuracy Under Composition}\label{app:comp-accuracy}
In this section, we prove \Lemmaref{distsens-distrib}.

\begin{lem*}[Restating \Lemmaref{distsens-distrib}]
Suppose $f: A \to B$ has distortion sensitivity \distsens{f}{} w.r.t.\ $(\dn_1,\dn_2)$. 
For all r.v.s $X_0$ over $A$ and $Y$ over $B$ such that $\dnx_2(f(X_0),\prob{Y}) \le \alpha$ for some $\alpha\geq0$, there must exist a r.v.\ $X$ over $A$ such that $Y=f(X)$ and $\dnx_1(\prob{X_0}, \prob{X}) \le \distsens{f}{}(\alpha)$, provided $\distsens{f}{}(\alpha)$ is finite.
\end{lem*}
\begin{proof}
Fix random variables $X_0$ over $A$ and $Y$ over $B$ such that $\dnx_2(f(X_0),\prob{Y}) \le \alpha$.
Let $\phi$  be an optimal coupling that achieves the infimum in the definition of $\dnx_2(f(X_0),\prob{Y})$, i.e., 
\begin{equation}\label{eq:distsens-distrib-interim1}
\dnx_2(f(X_0),\prob{Y}) = \sup_{(u,y)\leftarrow\phi}\dn_2(u,y) \leq \alpha. 
\end{equation}
For each $x_0 \in \support(X_0)$, consider the conditional distribution $\phi_{x_0} = \phi | \{X_0 = x_0\}$. 
Clearly, the first marginal of $\phi_{x_0}$ is a point distribution supported at $f(x_0)$. Let its second marginal be denoted by $\prob{Y_{x_0}}$. First we show that for each $x_0\in\support(X_0)$, we have $\dnx_2(f(x_0),\prob{Y_{x_0}})\leq \alpha$.
\begin{align*}
    \dnx_2(f(x_0),\prob{Y_{x_0}}) = \inf_{\phi\in\Phi^0(f(x_0),\prob{Y_{x_0}})} \sup_{(u,y)\leftarrow\phi}\dn_2(u,y) 
    \leq \sup_{(u,y)\leftarrow\phi_{x_0}}\dn_2(u,y) 
    \overset{\text{(a)}}{\leq} \sup_{(u,y)\leftarrow\phi}\dn_2(u,y) 
    \overset{\text{(b)}}{\leq} \alpha.
\end{align*}
Here $(a)$ follows from the fact that $\support(\phi_{x_0}) \subseteq \support(\phi)$ and (b) follows from \eqref{eq:distsens-distrib-interim1}.
Thus for each $x_0 \in \support(X_0)$, we have $\dnx_2(f(x_0),\prob{Y_{x_0}}) \leq \alpha$.
Since $\distsens{f}{}(\alpha)$ is finite, by the definition of $\distsens{f}{}$, there exist a r.v.\ $X_{x_0}$ such that
\begin{align}
    Y_{x_0} &= f(X_{x_0}), \label{distsens-distrib-interim3} \\
    \dnx_1(x_0,\prob{X_{x_0}}) &\le \distsens{f}{}(\alpha). \label{distsens-distrib-interim4}
\end{align} 
Define $X = \sum_{x_0\in\support(X_0)} \prob{X_0}(x_0) X_{x_0}$. 
Now we show that $Y = f(X)$ and $\dnx_1(\prob{X_0},\prob{X})\leq \distsens{f}{}(\alpha)$.
\begin{itemize}
\item {\bf Showing $Y = f(X)$:}
Note that $Y=\sum_{x_0\in\support(X_0)} \prob{X_0}(x_0) Y_{x_0}$ and $f(X)=\sum_{x_0\in\support(X_0)} \prob{X_0}(x_0) f(X_{x_0})$. Now the claim follows because because $Y_{x_0} = f(X_{x_0})$  for each $x_0\in\support(X_0)$ (from \eqref{distsens-distrib-interim3}).

\item {\bf Showing $\dnx_1(\prob{X_0},\prob{X})\leq \distsens{f}{}(\alpha)$:}
For each $x_0\in\support(X_0)$, let $\psi_{x_0}$ be the optimal coupling that achieves the infimum in the definition of 
$\dnx_1(x_0,\prob{X_{x_0}})$. That is, for each $x_0$,  $\psi_{x_0} \in \Phi^{0}(x_0,\prob{X_{x_0}})$ and 
$\dnx_1(x_0,\prob{X_{x_0}})=\sup_{(a,b)\leftarrow\psi_{x_0}}\dn_1(a,b)$.
Let $\psi$ be defined by $\psi(a,b)=\prob{X_0}(x_0)\psi_{x_0}(a,b)$. 
It is easy to verify that $\psi\in\Phi^{0}(\prob{X_0},\prob{X})$.
Further, 
\begin{align*}
\dnx_1(\prob{X_0},\prob{X}) &\leq \sup_{(a,b)\leftarrow\psi}\dn_1(a,b) 
= \sup_{x_0\leftarrow \prob{X_0}} \sup_{(a,b)\leftarrow\psi_{x_0}}\dn_1(a,b) 
= \sup_{x_0\leftarrow \prob{X_0}} \dnx_1(x_0,\prob{X_{x_0}}) 
\leq \distsens{f}{}(\alpha),
\end{align*}
where the last inequality follows from \eqref{distsens-distrib-interim4}.
\end{itemize}
This completes the proof of \Lemmaref{distsens-distrib}.
\end{proof}

\section{Proof of \Theoremref{compose-DP} -- Differential Privacy Under Composition}\label{app:comp-privacy}

\begin{thm*}[Restating \Theoremref{compose-DP}] 
Let $\M_1:A\to B$ and $\M_2:B\to C$ be any two mechanisms.
If $\M_1$ is neighborhood-preserving w.r.t.\
neighborhood relations $\sim_A$ and $\sim_B$ over $A$ and $B$, respectively,
and $\M_2$ is $(\epsilon, \delta)$-DP w.r.t.\ $\sim_B$, 
then $\M_2\circ \M_1:A\to C$ is $(\epsilon, \delta)$-DP w.r.t.\ $\sim_A$.
\end{thm*}

\begin{proof}
For simplicity, we consider the case when $B$ is discrete. The proof can be
generalized to the continuous setting.

Since the mechanism $\M_1$ is neighborhood preserving, for $x, x' \in A$
s.t.\ $x_1 \sim_A x_2$, there exists a pair of jointly distributed random
variables $(X_1,X_2)$ over $B\times B$ s.t, $\prob{X_1} = \M_1(x)$,
$\prob{X_2} = \M_1(x')$ and $\Pr[X_1 \sim_B X_2] = 1$.  So, for all $(x_1,
x_2)$ such that $\prob{X_1,X_2}(x_1,x_2)>0$, we have $x_1 \sim_B x_2$ and
hence, by the $(\epsilon , \delta)$-differential privacy of the mechanism
$\M_2$, for all subsets $S \subseteq C$, we have, 
\begin{align*}
    \Pr(\M_2(x_1) \in S) &\leq e^{\epsilon} \Pr(\M_2(x_2) \in S) + \delta.
\end{align*}
Thus, if $x\sim_A x'$, then for any subset $S\subseteq C$, we have,
\begin{align*}
        \Pr[\M_2(\M_1(x)) \in S]
        &= \sum_{x_1} \prob{X_1}(x_1)\Pr[\M_2(x_1) \in S]  \\
        &= \sum_{(x_1, x_2)} \prob{X_1,X_2}(x_1, x_2) \Pr[\M_2(x_1) \in S]  \\ 
        &\le \sum_{(x_1, x_2)} \prob{X_1,X_2}(x_1, x_2) \left(e^{\epsilon} \Pr[\M_2(x_2) \in S] + \delta\right) \\ 
        &= e^\epsilon \left( \sum_{(x_1, x_2)} \prob{X_1,X_2}(x_1, x_2) \Pr[\M_2(x_2) \in S] \right) + \delta \\ 
        &= e^\epsilon \left( \sum_{x_2} \prob{X_2}(x_2) \Pr[\M_2(x_2) \in S] \right) + \delta \\ 
        &= e^{\epsilon} \Pr[\M_2(\M_1(x')) \in S] + \delta 
\end{align*}
This completes the proof of \Theoremref{compose-DP}.
\end{proof}

\section{Details Omitted from \Sectionref{hist-priv-accu-proof} -- Shifted-Truncated Laplace Mechanism}\label{app:histogram}

\begin{claim*}[Restating \Claimref{hist-epdel-c2}]
$\Pr[\mtrlap{\thresh,\eps,\G}(\x')\in S_0\cup S_2]\leq e^{(1 + \nu)\eps}\Pr[\mtrlap{\thresh,\eps,\G}(\x)\in S_0\cup S_2]$, provided $n \ge \frac{2}{\eps\thresh} \ln\left(1 + \frac{1 -  e^{-\frac{\eps\thresh}{2}}}{e^{\eps(\nu + \frac{\thresh}{2})} - 1} \right)$.
\end{claim*}
\begin{proof}
First we show that for $\s\in S_0\cup S_2$, we have, $\dnoise{s_{i^*} - x_{i^*}'}\leq e^{(1 + \nu)\eps}\lnoise{s_{i^*} - x_{i^*}}$, provided $n \ge \frac{2}{\eps\thresh} \ln\left(1 + \frac{1 -  e^{-\frac{\eps\thresh}{2}}}{e^{\eps(\nu + \frac{\thresh}{2})} - 1} \right)$, and then we show how this implies the result.

For $\s\in S_0$, $\dnoise{s_{i^*} - x_{i^*}'} = 0$ so the inequality trivially holds.
For $\s\in S_2$, both $\dnoise{s_{i^*} - x_{i^*}'} > 0$ and $\lnoise{s_{i^*} - x_{i^*}} > 0$; hence, we will be done if we show that 
$\frac{\dnoise{s_{i^*} - x_{i^*}'}}{\lnoise{s_{i^*} - x_{i^*}}} \leq e^{(1 + \nu)\eps}$. Note that we are given the following inequality:
\begin{align*}
    n &\ge \frac{2}{\eps\thresh} \ln\left(1 + \frac{1 -  e^{-\eps\frac{\thresh}{2}}}{e^{\eps(\nu + \frac{\thresh}{2})} - 1} \right),
\end{align*}
which can be rewritten as (which we show in \Claimref{equiv-ineq_claim-hist} after this proof):
\begin{align}
    \ln\left(\frac{1 -  e^{-\eps\frac{\thresh (n+1)}{2}}}{1 - e^{-\eps\frac{\thresh n}{2}}}\right) &\leq \eps(\nu + \frac{\thresh}{2}).  \label{eq:n-bound}
\end{align}
By substituting $q=\thresh (n+1)$ and $q'=\thresh n$, \eqref{eq:n-bound} is equivalent to
\[\frac{1}{\eps}\ln\left(\frac{1 -  e^{-\eps\frac{q}{2}}}{1 - e^{-\eps\frac{q'}{2}}}\right) + (1 - \frac{\thresh}{2}) \leq 1 + \nu.\]
This, using the triangle inequality, implies that
\[\frac{1}{\eps}\ln\left(\frac{1 -  e^{-\eps\frac{q}{2}}}{1 - e^{-\eps\frac{q'}{2}}}\right) + \left|s_{i^*} - x_{i^*} + \frac{q}{2}\right| - \left|s_{i^*} - x_{i^*} + \frac{q}{2} + (1 - \frac{\thresh}{2})\right| \leq 1 + \nu.\]
Putting $q'=q-\thresh$ and $x_{i^*}'=x_{i^*}-1$, we get
\[\frac{1}{\eps}\ln\left(\frac{1 -  e^{-\eps\frac{q}{2}}}{1 - e^{-\eps\frac{q'}{2}}}\right) + \left|s_{i^*} - x_{i^*} + \frac{q}{2}\right| - \left|s_{i^*} - x_{i^*}' + \frac{q'}{2}\right| \leq 1 + \nu.\]
By taking exponents of both sides, this is equivalent to showing
\[\frac{(1 -  e^{-\eps\frac{q}{2}})}{(1 - e^{-\eps\frac{q'}{2}})} \frac{e^{-\eps|s_{i^*} - x_{i^*}' + \frac{q'}{2}|}}{e^{-\eps|s_{i^*} - x_{i^*} + \frac{q}{2}|}} \leq e^{(1 + \nu)\eps}\]
By substituting the values of $\lnoise{s_{i^*} - x_{i^*}}$ and $\dnoise{s_{i^*} - x_{i^*}'}$, this can be equivalently written as
\begin{equation}\label{hist-epdel-c2_interim1}
\frac{\dnoise{s_{i^*} - x_{i^*}'}}{\lnoise{s_{i^*} - x_{i^*}}} \leq e^{(1 + \nu)\eps}.
\end{equation}

\noindent Now we show $\Pr[\mtrlap{\thresh,\eps,\G}(\x')\in S_0\cup S_2]\leq e^{(1 + \nu)\eps}\Pr[\mtrlap{\thresh,\eps,\G}(\x)\in S_0\cup S_2]$. Recall that $\G_{\x}=\support(\x)$ for any histogram $\x\in\Hspace\G$.
\begin{align*}
\Pr[\mtrlap{\thresh,\eps,\G}(\x')\in S_0\cup S_2] &= \int_{S_0\cup S_2} \big[\prod_{i\in\G_{\x'}} \dnoise{s_i - x_i'}\big]\dd\s \nonumber\\
&= \int_{S_0\cup S_2} \big[\prod_{\substack{i\in\G_{\x'} : i\neq i^*}} \dnoise{s_i - x_i'}\big]\dnoise{s_{i^*} - x_{i^*}'}\dd\s \nonumber\\
&\leq \int_{S_0\cup S_2} \big[\prod_{\substack{i\in\G_{\x} : i\neq i^*}} \lnoise{s_i - x_i}\big]e^{(1 + \nu)\eps}\lnoise{s_{i^*} - x_{i^*}}\dd\s \nonumber \tag{Using \eqref{hist-epdel-c2_interim1} and that $x_i=x_i', \forall i\neq i^*$}\\
&= e^{(1 + \nu)\eps}\int_{S_0\cup S_2} \big[\prod_{i\in\G_{\x}} \lnoise{s_i - x_i}\big]\dd\s \nonumber \\
&= e^{(1 + \nu)\eps} \Pr[\mtrlap{\thresh,\eps,\G}(\x)\in S_0\cup S_2]
\end{align*}
This completes the proof of \Claimref{hist-epdel-c2}.
\end{proof}

\begin{claim}\label{clm:equiv-ineq_claim-hist}
\[n \ge \frac{2}{\eps\thresh} \ln\left(1 + \frac{1 -  e^{-\eps\frac{\thresh}{2}}}{e^{\eps(\nu + \frac{\thresh}{2})} - 1} \right)\quad \Longleftrightarrow \quad \ln\left(\frac{1 -  e^{-\eps\frac{\thresh (n+1)}{2}}}{1 - e^{-\eps\frac{\thresh n}{2}}}\right) \leq \eps(\nu + \frac{\thresh}{2}).\]
\end{claim}
\begin{proof}
We will start with the RHS and show that it is equivalent to the LHS.
\begin{align*}
\frac{1 -  e^{-\eps\frac{\thresh (n+1)}{2}}}{1 - e^{-\eps\frac{\thresh n}{2}}} &\leq e^{\eps(\nu + \frac{\thresh}{2})} \\
\Longleftrightarrow 1 -  e^{-\eps\frac{\thresh (n+1)}{2}} &\leq e^{\eps(\nu + \frac{\thresh}{2})} - e^{\eps(\nu + \frac{\thresh}{2})}e^{-\eps\frac{\thresh n}{2}} \\
\Longleftrightarrow 1 -  e^{-\eps\frac{\thresh n}{2}}e^{-\eps\frac{\thresh}{2}} &\leq e^{\eps(\nu + \frac{\thresh}{2})} - e^{\eps(\nu + \frac{\thresh}{2})}e^{-\eps\frac{\thresh n}{2}} \\
\Longleftrightarrow e^{-\eps\frac{\thresh n}{2}}\left(e^{\eps(\nu + \frac{\thresh}{2})} - e^{-\eps\frac{\thresh}{2}}\right) &\leq e^{\eps(\nu + \frac{\thresh}{2})} - 1 \\
\Longleftrightarrow e^{\eps\frac{\thresh n}{2}} &\geq \frac{e^{\eps(\nu + \frac{\thresh}{2})} - e^{-\eps\frac{\thresh}{2}}}{e^{\eps(\nu + \frac{\thresh}{2})} - 1} \\
\Longleftrightarrow e^{\eps\frac{\thresh n}{2}} &\geq 1 + \frac{1 - e^{-\eps\frac{\thresh}{2}}}{e^{\eps(\nu + \frac{\thresh}{2})} - 1} \\
\Longleftrightarrow n &\ge \frac{2}{\eps\thresh} \ln\left(1 + \frac{1 -  e^{-\eps\frac{\thresh}{2}}}{e^{\eps(\nu + \frac{\thresh}{2})} - 1} \right).
\end{align*}
\end{proof}

\begin{claim*}[Restating \Claimref{hist-epdel-c3}]
$\Pr[\mtrlap{\thresh,\eps,\G}(\x')\in S_{1}]\leq \frac{e^{\eps} - 1}{2(e^{\nicefrac{\eps q}{2}} - 1)}$.
\end{claim*}
\begin{proof}
Observe that, for every $\s\in S_1$, we have $-q' \le s_{i^*} - x_{i^*}' < -q' + (1 - \thresh)$.
Recall that $\G_{\x'}=\support(\x')$ and $|\x'|=n$. 
Let $|\G_{\x'}|=t$ for some $t\leq n$, and, for simplicity, assume that $\G_{\x'}=\{1,2,\hdots,t\}$.
For $i\in[t]$, define $S_1(i):=\{\hat{s}_i:\exists \s\in S_1 \text{ s.t. } \hat{s}_i=s_i\}$, which is equal to the collection of the multiplicity of $i$ in the histograms in $S_1$.
\begin{align*}
\Pr[\mtrlap{\thresh,\eps,\G}(\x')\in S_1] &= \int_{S_1} \big[\prod_{i=1}^t \dnoise{s_i - x_i'}\big]\dd\s \nonumber\\
& = \int_{S_1(1)}\ldots\int_{S_1(i^*)}\ldots\int_{S_1(t)} \big[\prod_{i=1}^t \dnoise{s_i - x_i'}\big]\dd s_t\ldots \dd s_{i^*}\ldots \dd s_1 \nonumber \\
& = \int_{S_1(i^*)}\dnoise{s_{i^*} - x_{i^*}'} \underbrace{\bigg(\int_{S_1(1)}\ldots \int_{S_1(t)} \big[\prod_{\substack{i=1 : i\neq i^*}}^t \dnoise{s_i - x_i'}\big]\dd s_t\ldots  \dd s_1\bigg)}_{\leq\ 1} \dd s_{i^*}  \nonumber\\
& \leq \int_{S_1(i^*)}\dnoise{s_{i^*} - x_{i^*}'} \dd s_{i^*}  \nonumber\\
&= \int_{q'}^{q' + (1 - \thresh)}\dnoise{z} \dd z \tag{Since $\forall \s\in S_1$, $(s_{i^*}-x_{i^*}')\in[-q',-q' + (1-\thresh))$} \nonumber \\
& =  \frac{e^{(1-\thresh)\eps} - 1}{2(1 - e^{-\eps q/2})}e^{-\eps \nicefrac{q}{2}} \nonumber\\
& \le  \frac{e^{\eps} - 1}{2(e^{\eps \nicefrac{q}{2}} - 1)}. \tag{Since $\thresh > 0$}
\end{align*}
This proves \Claimref{hist-epdel-c3}.
\end{proof}

\begin{lem}\label{lem:hist-priv-nu-bigger-0}
For any $\nu,\eps>0$ and $\x$ such that $\eps\nu>\ln\left(1+\frac{1}{|\x|}\right)$, \mtrlap{\thresh,\eps,\G} is 
$\left((1 + \nu)\eps,\frac{e^{\eps} - 1}{2(e^{\nicefrac{\eps q}{2}} - 1)}\right)$-DP w.r.t.\ \nhist, where $q=\tau|\x|$.
\end{lem}
\begin{proof}
We use \Lemmaref{hist-epdel} and put a restriction that $\nu$ should be $> 0$. We will analyze the effect of this restriction on the bound of $|\x|$. We restate the bound on $|\x|$ here again for convenience:
\[|\x| \ge \frac{2}{\eps\thresh} \ln\left(1 + \frac{1 -  e^{-\eps\frac{\thresh}{2}}}{e^{\eps(\nu + \frac{\thresh}{2})} - 1} \right)\]
It can be easily checked that for any fixed $\eps, \nu > 0$, the RHS is a decreasing function of $\thresh$.
Hence, if we set $\thresh$ to its minimum value, we get a lower bound on $|\x|$ which is independent of $\tau$.
Since this expression is not defined at $\thresh = 0$, we will take its one-sided limit as $\thresh \rightarrow 0^+$, i.e.,
\begin{align*}
    \lim_{\thresh \rightarrow 0^+}\frac{2}{\eps\thresh} \ln\left(1 + \frac{1 -  e^{-\eps\frac{\thresh}{2}}}{e^{\eps(\nu + \frac{\thresh}{2})} - 1} \right)
\end{align*}
We will replace $\frac{\eps\thresh}{2}$ with $l$. As $\thresh \rightarrow 0^+$, $l \rightarrow 0^+$, and we get
\begin{align*}
    \lim_{\thresh \rightarrow 0^+}\frac{2}{\eps\thresh} \ln\left(1 + \frac{1 -  e^{-\eps\frac{\thresh}{2}}}{e^{\eps(\nu + \frac{\thresh}{2})} - 1} \right) &= \lim_{l \rightarrow 0^+}\frac{1}{l} \ln\left(1 + \frac{1 -  e^{-l}}{e^{\eps\nu + l} - 1} \right)\\
    &= \lim_{l \rightarrow 0^+}\frac{1}{l} \ln\left(1 + \frac{1 -  e^{-l}}{e^{\eps\nu + l} - 1} \right)\left(\frac{1 -  e^{-l}}{e^{\eps\nu + l} - 1}\right)\left(\frac{e^{\eps\nu + l} - 1}{1 -  e^{-l}}\right)\\
    &= \lim_{l \rightarrow 0^+}\left(\frac{1}{e^{\eps\nu + l} - 1}\right)\left(\frac{1 -  e^{-l}}{l}\right) \left(\frac{\ln\left(1 + \frac{1 -  e^{-l}}{e^{\eps\nu + l} - 1} \right)}{\frac{1 -  e^{-l}}{e^{\eps\nu + l} - 1}}\right)\\
    &= \frac{1}{e^{\eps\nu} - 1} \tag{$\lim_{x \rightarrow 0^+} \frac{1-e^{-x}}{x} = 1$; $\lim_{x \rightarrow 0^+} \frac{\ln(1+x)}{x} = 1$}\\
\end{align*}
We have proved that on inputs $\x$ s.t.\ $|\x| > \frac{1}{e^{\eps\nu} - 1}$, which is equivalent to the condition that $\eps\nu>\ln\left(1+\frac{1}{|\x|}\right)$, \mtrlap{\thresh,\eps,\G} is
$\left((1 + \nu)\eps,\frac{e^{\eps} - 1}{2(e^{\nicefrac{\eps q}{2}} - 1)}\right)$-DP w.r.t.\ \nhist, where $q=\tau|\x|$. 
\end{proof}

\section{Details Omitted from \Sectionref{bucketHist-priv-accu-proof} -- Bucketed Histogram Mechanism}\label{app:bucketHist_proofs}

\begin{claim*}[Restating \Claimref{bucket-accuracy}]
$\mbuc{w,[0,B)}$ is $\left(0, \frac{w}{2}, 0\right)$-accurate for the identity function $f_{\emph{id}}$ over $\Hspace{[0,B)}$ w.r.t~ the metric $\dhistx$. 
\end{claim*}
\begin{proof}
Since both $f_{\text{id}}$ and $\mbuc{w,[0,B)}$ are deterministic maps, on any input $\x\in\Hspace{[0,B)}$, we denote $\x$ (as the output of $f_{\text{id}}(\x)$) and $\mbuc{w,[0,B)}$ as point distributions over $\Hspace{[0,B)}$.
Now, in order to prove the claim, we need to show that $\Winf{}(\mbuc{w,[0,B)}(\x),\x)\leq\frac{w}{2}$ holds for any $\x\in\Hspace{[0,B)}$. 

Fix any $\x\in\Hspace{[0,B)}$ and define $\y:=\mbuc{w,[0,B)}(\x)$.
Since $\x,\y$ are point distributions and the underlying metric is $\dhistx$, we have $\Winf{}(\y,\x)=\dhistx(\y,\x)$, where $\dhistx$ is defined as $\dhist{\y}{\x}=\Winf{}(\frac{\y}{|\y|},\frac{\x}{|\x|})$. Since $\y$ is a deterministic function of $\x$, $\Winf{}(\frac{\y}{|\y|},\frac{\x}{|\x|})$ is upper bounded by the maximum distance any point in $\x$ moves to form $\y$, which is equal to the the maximum distance of the center of a bucket from any point in that bucket, which is $\frac{w}{2}$.
\end{proof}

\begin{claim*}[Restating \Claimref{fphi_distort}]
For any $\x\in\Hspace\G$, if $\z$ is $\alpha$-distorted from $\x$, then $f_{\phi}(\z)$ is $\alpha$-distorted from $\x'$.
\end{claim*}
\begin{proof}
We need to show two things: 
{\sf (i)} $f_{\phi}(\z)(b) \leq \x'(b)$ holds for every $b\in\G$, and 
{\sf (ii)} $\sum_{b\in\G}f_{\phi}(\z)(b) \geq (1-\alpha)\sum_{b\in\G}\x'(b)$.
The first condition holds because $\z(a)\leq\x(a),\forall a\in\G$ (since $\z$ is $\alpha$-distorted from $\x$) and that $\sum_{a\in\G}\phi(a,b)=\frac{\x'(b)}{|\x'|}$.
For the second condition, 
\begin{align}
\sum_{b\in\G}f_{\phi}(\z)(b) &= \sum_{b\in\G}|\x'|\sum_{a\in\G} \frac{\z(a) \phi(a , b)}{\x(a)} 
= |\x'|\sum_{a\in\G}\frac{\z(a)}{\x(a)}\sum_{b\in\G}\phi(a,b) 
\stackrel{\text{(a)}}{=} \frac{|\x'|}{|\x|} \sum_{a\in\G}\z(a) 
\stackrel{\text{(b)}}{\geq} (1-\alpha)|\x'|, \label{fphi_dist_bound} 
\end{align}
where (a) follows from $\sum_{b\in\G}\phi(a,b)=\frac{\x(a)}{|\x|}$ and (b) follows because $\z$ is $\alpha$-distorted from $\x$, which implies that $\sum_{a\in\G}\z(a) = |\z| \geq (1-\alpha)|\x|$.
Therefore, $f_{\phi}(\z)$ is $\alpha$-distorted from $\x'$.
\end{proof}

\begin{claim*}[Restating \Claimref{distortion_bw_xprime-Y}]
$\widehat{\partial}_{\emph{drop}}(\x',\prob{Y}) \leq \alpha$.
\end{claim*}
\begin{proof}
Note that the support of $\mtrlap{\thresh,\eps,\G}(\x)$ is the set of all $\alpha$-distorted histograms from $\x$. 
We have shown in \Claimref{fphi_distort} that for any $\z\in\Hspace\G$ such that $\drop(\x,\z)\leq\alpha$, we have $\drop(\x',f_{\phi}(\z))\leq\alpha$.
This implies that $\sup_{\y\in\support(f_{\phi}(\mtrlap{\thresh,\eps,\G}(\x)))}\drop(\x',\y)\leq\alpha$, which in turn implies that $\widehat{\partial}_{\text{drop}}(\x',\prob{Y})\leq\alpha$.
\end{proof}

\begin{claim*}[Restating \Claimref{Winf_Y_trLap-bound}]
$\Winf{}(\mtrlap{\thresh,\eps,\G}(\x),\prob{Y}) \leq \beta$.
\end{claim*}
\begin{proof}
Define a coupling $\phi_{\x}$ of $\mtrlap{\thresh,\eps,\G}(\x)$ and $\prob{Y}$ over $\Hspace\G\times\Hspace\G$ as follows:
\[\phi_{\x}(\z,\y) :=
\begin{cases}
\Pr[\mtrlap{\thresh,\eps,\G}(\x) = \z] & \text{ if } \y= f_{\phi}(\z), \\
0 & \text{ otherwise}. 
\end{cases}
\]
It is easy to verify that the above defined $\phi_{\x}$ is a valid coupling of $\mtrlap{\thresh,\eps,\G}(\x)$ and $\prob{Y}$, i.e., its first marginal is equal to $\mtrlap{\thresh,\eps,\G}(\x)$ and the second marginal is equal to $\prob{Y}$. Note that $\phi_{\x}(\z,\y)$ is non-zero only when $\y=f_{\phi}(\z)$.
This implies that 
\[\Winf{}(\mtrlap{\thresh,\eps,\G}(\x),\prob{Y}) \leq \sup_{(\z,\y)\leftarrow\phi_{\x}}\dhist{\z}{\y} = \sup_{(\z,f_{\phi}(\z))\leftarrow\phi_{\x}}\dhist{\z}{f_{\phi}(\z)} \leq \beta,\]
where the last inequality follows from \Claimref{f_phi_preserves_dhist} (stated and proven below) and using the fact that $\z\sim\mtrlap{\thresh,\eps,\G}(\x)$ is $\alpha$-distorted from $\x$.
\end{proof}

\begin{claim}\label{clm:f_phi_preserves_dhist}
Let $\x,\x'\in\Hspace\G$ be such that $\dhist{\x}{\x'}\leq\beta$. Then, for any $\z$ that is $\alpha$-distorted from $\x$, we have $\dhist{\z}{f_{\phi}(\z)} \leq \dhist{\x}{\x'}\leq\beta$.
\end{claim}
\begin{proof}
Define $\phi'(a, b) = \frac{\z(a)|\x|\phi(a , b)}{\x(a)|\z|}$.
For any $a\in\G$, its first marginal is equal to $\sum_{b\in\G}\phi'(a,b)=\frac{\z(a)}{|\z|}$.
For any $b\in\G$, its second marginal is equal to 
$\sum_{a\in\G}\phi'(a,b) = \frac{|\x|}{|\z|}\sum_{a\in\G}\frac{\z(a)}{\x(a)}\phi(a , b) = \frac{|\x|}{|\x'||\z|}f_{\phi}(\z)(b)$. We would like to say that the quantity on the RHS is equal to $\frac{f_{\phi}(\z)(b)}{|f_{\phi}(\z)|}$. We show this as follows:
Since $|\z|\geq(1-\alpha)|\x|$, there exists $c\geq0$ such that $|\z|=(1-\alpha+c)|\x|$. If we put this instead of $|\z|\geq(1-\alpha)|\x|$ in \eqref{fphi_dist_bound}, we would get $\sum_{b\in\G}f_{\phi}(\z)(b)=(1-\alpha+c)|\x'|$. With these substitutions, we get $\frac{|\x|}{|\x'||\z|}f_{\phi}(\z)(b)=\frac{f_{\phi}(\z)(b)}{\sum_{b\in\G}f_{\phi}(\z)(b)}$, which implies that the second marginal of $\phi'$ is equal to $\sum_{a\in\G}\phi'(a,b)=\frac{f_{\phi}(\z)(b)}{|f_{\phi}(\z)|}$ for any $b\in\G$.

This means that $\phi'(a,b)$ is a valid coupling of $\z, f_{\phi}(\z)$. This implies that
\begin{align*}
\dhist{\z}{f_{\phi}(\z)} = \Winf{}(\z,f_{\phi}(\z)) \leq \sup_{\substack{(a', b') \leftarrow \phi'}}\dG{a'}{b'} \stackrel{\text{(c)}}{\leq} \sup_{\substack{(a', b') \leftarrow \phi}}\dG{a'}{b'} = \dhist{\x}{\x'} \leq \beta,
\end{align*}
where (c) holds because $\text{support}(\phi') \subseteq \text{support}(\phi)$ (by the definition of $\phi'$). 
\end{proof}

\section{$d$-Dimensional Analogues of our Mechanisms/Results}\label{app:d-dim-results}

In our $d$-dimensional bucketing mechanism for $\G=[0,B)^d$, we divide $[0,B)^d$ into $t=\lceil \frac{B}{w}\rceil^d$ $d$-dimensional cubes (buckets), each of side length $w$, and map each input point to the center of the nearest cube (bucket). Note that the distance between any point in $[0,B)^d$ to the center of the nearest bucket is $\frac{w}{2}\sqrt{d}$. In the following, we will ignore the ceil/floor for simplicity.

\begin{algorithm}
\caption{Bucketing Mechanism over $[0,B)^d$, \mbuc{w,[0,B)^d}}\label{algo:app-bucketing}
{\bf Parameter:} Bucket (which is $d$-dimensional cube) side length $w$, ground set $[0,B)^d$. \\
{\bf Input:} A histogram $\x$ over $[0,B)$. \\
{\bf Output:} A histogram $\y$ over $S = T^d$ where $T=\{ w(i-\frac12) : i \in [t], t = \lceil \frac{B}{w} \rceil \}$, and $|\y| = |\x|$. \\
\vspace{-0.3cm}
\begin{algorithmic}[1] 
\ForAll{$s \in S $}
\State $\y(s) := \sum_{g:g-s \in [\frac{-w}2,\frac{w}2)^d} \; \x(g)$
\EndFor
\State Return \y
\end{algorithmic}
\end{algorithm}

\begin{algorithm}
\caption{BucketHist Mechanism over $[0,B)^d$, \mBhist{\alpha,\beta,[0,B)^d}}\label{algo:app-histogram-mech}
{\bf Parameter:} Accuracy parameters $\alpha,\beta$; ground set $[0,B)^d$. \\
{\bf Input:} A histogram \x over $[0,B)^d$. \\
{\bf Output:} A histogram \y over $[0,B)^d$. \\
\vspace{-0.3cm}
\begin{algorithmic}[1] 
\State $w := 2\beta$, $t := \lceil \frac{B}{w}\rceil^d$, $\thresh := \alpha/t$
\State Return $\mtrlap{\thresh,\eps,[0,B)^d} \circ \mbuc{w,[0,B)^d} (\x)$ \Comment{where \mbuc{w,[0,B)} is in \Algorithmref{app-bucketing}}
\end{algorithmic}
\end{algorithm}

Our $d$-dimensional bucketing mechanism $\mbuc{w,[0,B)^d}$ and the final $d$-dimensional bucketed-histogram mechanism $\mBhist{\alpha,\beta,[0,B)^d}$ are presented in \Algorithmref{app-bucketing} and \Algorithmref{app-histogram-mech}, respectively.

As mentioned in \Remarkref{d-dim} in \Sectionref{beyond-drop}, with these modified mechanisms, all our results in \Theoremref{bucketing-hist}, \Theoremref{bucketing-general}, and \Theoremref{bucketing-general-drmv} will hold verbatim, except for the value of $\tau$, which will be replaced by $\tau=\alpha(\frac{2\beta}{B\sqrt{d}})^d$. Note that for the one dimensional case, we have $\tau=\frac{\alpha}{t}=\alpha(\frac{w}{B})$, where $w=2\beta$, which comes from the $(0,\frac{w}{2},0)$-accuracy of the bucketing mechanism $\mbuc{w,[0,B)}$ (see \Claimref{bucket-accuracy} in \Sectionref{bucketHist-priv-accu-proof}).
The $d$-dimensional analogue of that result is stated in the following claim which can be proven along the lines of the proof of \Claimref{bucket-accuracy}..
\begin{claim}\label{clm:bucket-accuracy-d-dim}
$\mbuc{w,[0,B)^d}$ is $\left(0, \frac{w}{2}\sqrt{d}, 0\right)$-accurate for the identity function $f_{\emph{id}}$ over $\Hspace{[0,B)^d}$ w.r.t~ the metric $\dhistx$. 
\end{claim}

It follows from \Claimref{bucket-accuracy-d-dim} that the output error of $\mbuc{w,[0,B)^d}$ is $\beta=\frac{w}{2}\sqrt{d}$. This implies $\tau=\frac{\alpha}{t}=\alpha(\frac{w}{B})^d=\alpha(\frac{2\beta}{B\sqrt{d}})^d$.

\section{Details Omitted from \Sectionref{distortion-measures}}\label{app:beyond-drop}
In this section, first we prove that $\dnx$ is a quasi-metric (assuming that $\dn$ is a quasi-metric), and then prove that our two distortions $\move$ and $\dropmove\eta$ (defined in \eqref{eq:move_defn} and \eqref{eq:drop_move_defn}, respectively) are metric and quasi-metric, respectively.

\begin{lem}\label{lem:dnx-quasi-metric}
If \dn is a quasi-metric, then \dnx is a quasi-metric. 
\end{lem}
\begin{proof}
We need to show that for any three distributions $P$, $Q$, and $R$ over the same space $A$, we have 
{\sf (i)} $\dnx(P, Q)\geq 0$, where the equality holds if and only if $P = Q$, and {\sf (ii)} $\dnx$ satisfies the triangle inequality: $\dnx(P, Q)\leq \dnx(P, R) + \dnx(R, Q)$. We show them one by one below:
\begin{enumerate}
\item The first property follows from the definition of $\dnx$ (see \Definitionref{distortion}): If $\dnx(P, Q)= 0$, then the optimal $\phi\in\Phi(P, Q)$ is a diagonal distribution, which means that $P = Q$. On the other hand, if $P = Q$, then there exists a coupling $\phi$ in $\Phi(P, Q)$, which is a diagonal distribution and hence $\dnx(P, Q)= 0$.
\item Since the definition of $\dnx$ is the same as that of $\Winf{}$, except for that the former is defined w.r.t.\ a quasi-metric, whereas, the latter is defined w.r.t.\ a metric, we can show the triangle inequality for $\dnx$ along the lines of the proof of \Lemmaref{Winf_triangle}. Note that we did not use the symmetric property of $\Winf{}$ while proving \Lemmaref{Winf_triangle}; we only used that the underlying metric $\met$ satisfies the triangle inequality, which also holds for $\dnx$ which is a quasi-metric.

\end{enumerate}
This completes the proof of \Lemmaref{dnx-quasi-metric}.
\end{proof}

In \Sectionref{distortion-measures}, we introduced two new distortions: $\move(\x,\y)$ in \eqref{eq:move_defn} and $\dropmove\eta(\x,\y)$ in \eqref{eq:drop_move_defn}. We prove that $\move$ is a metric in \Claimref{move-metric} and that $\dropmove\eta$ is a metric in \Claimref{drop-move-quasi-metric}.
We present the definitions of these distortions here again for convenience:

\begin{align*}
\move(\x,\y) &= \begin{cases}
\Winf{}(\frac{\x}{|\x|},\frac{\y}{|\y|}) & \text{ if } |\x|=|\y| \\
\infty & \text{ otherwise}
\end{cases}
&
\dropmove\eta(\x,\y) = \inf_{\z} \left(\drop(\x,\z) + \eta \cdot \move(\z,\y)\right).
\end{align*}

Note that in the definition of $\move$, when $|\x| = |\y| = 0$, we define $\move(\x, \y) = 0$.

\begin{claim}\label{clm:move-metric}
$\move(\cdot,\cdot)$ is a metric.
\end{claim}
\begin{proof}
Since $\move(\cdot,\cdot)$ is defined as the $\infty$-Wasserstein distance between normalized histograms, it suffices to show that the $\infty$-Wasserstein distance is a metric. We need to show three things for any triple of distributions $P,Q,R$ over a metric space $(\Omega,\met)$: 
{\sf (i)} $\Winf{}(P,Q)\geq0$ and equality holds if and only if $P=Q$, {\sf (ii)} $\Winf{}(P,Q)=\Winf{}(Q,P)$, and {\sf (iii)} $\Winf{}(P,R)\leq\Winf{}(P,Q)+\Winf{}(Q,R)$. 

By definition, $\Winf{}(P, R) = \inf_{\phi\in\Phi(P,R)} \sup_{(x,z):\\\phi(x,z)\neq 0}\met(x,z)$.
Now, the first two conditions follow because $\met$ is a metric, and the last condition (triangle inequality) we show in \Lemmaref{Winf_triangle} in \Appendixref{wasserstein}.

Note that when $|\x| = |\y| = 0$, the Wasserstein distance is undefined, but we have defined $\move(\x, \y)$ in this case separately as $0$ which is consistent with the properties of a metric.
\end{proof}

We first give an intermediate result (\Lemmaref{drop-move-switch} below) which will be used in proving that $\drme$ is a quasi-metric. The result of this lemma is also used in the proof of \Theoremref{bucketing-general-drmv}.
\begin{lem}\label{lem:drop-move-switch}
Let $\x$, $\y$ and $\z$ be any three histograms over a ground set $\G$ ,associated with a metric $\met$, such that $\move(\x, \z) = \alpha_1$ and $\drop(\z, \y) = \alpha_2$ with $\alpha_1 \ge 0$ and $\alpha_2 < 1$. Then there exists a histogram $\s$ such that $\drop(\x, \s) = \alpha_2$ and $\move(\s, \y) \leq \alpha_1$.
\end{lem}
\begin{proof}
Using the definitions of $\drop$ and $\move$, we have the following:
\begin{enumerate}[label=\textbf{Z.\arabic*}]
    \item \label{Z1} $|\x| = |\z|$
    \item \label{Z2} $\Winf{}(\frac{\x}{|\x|}, \frac{\z}{|\z|}) \leq \alpha_1$. We will use $\phi_z$ to denote the optimal joint distribution which achieves the infimum in the definition of $\Winf{}(\frac{\x}{|\x|}, \frac{\z}{|\z|})$.
    \item \label{Z3} $|\y| = (1 - \alpha_2)|\z|$
    \item \label{Z4} For all $g \in \G$, $0 \le \y(g) \le \z(g)$
\end{enumerate}

Now we want to prove the existence of a histogram $\s$ with the following property:
\begin{enumerate}[label=\textbf{S.\arabic*}]
    \item \label{S1} $|\s| = (1 - \alpha_2)|\x|$ 
    \item \label{S2} For all $g \in \G$, $0 \le \s(g) \le \x(g)$
    \item \label{S3} $|\s| = |\y|$
    \item \label{S4} $\Winf{}(\frac{\s}{|\s|}, \frac{\y}{|\y|}) \leq \alpha_1$. 
\end{enumerate}
Consider the following joint distribution $\phi_s$:
\begin{equation}\label{eq:drmv-phi}
    \phi_s(g_x, g_y) = \begin{cases}
    \frac{1}{1 - \alpha_2} \phi_z(g_x, g_y)\frac{\y(g_y)}{\z(g_y)} & \text{if } \z(g_y) > 0\\
    0 & \text{otherwise}
    \end{cases}
\end{equation}
We denote the first marginal of $\phi_s$ by $\frac{\s}{|\s|}$, where $\s$ corresponds to the histogram that we want to show exists. 

By definition, for all $g_x, g_y \in \G$, we have $\phi_s(g_x, g_y) \ge 0$. 
Also note that, if $\z(g_y) = 0$, then for all $g_x \in \G$, we have $\phi_z(g_x, g_y) = 0$; this is because $\frac{\z}{|\z|}$ is the second marginal of $\phi_z$. 
Now we show that the above-defined $\phi_s$ satisfies properties \ref{S1}-\ref{S4} -- we show these in the sequence of \ref{S4}, \ref{S3}, \ref{S1}, \ref{S2}.
\begin{itemize}
\item {\bf Proof of \ref{S4}.}
Note that the first marginal of $\phi_s$ is assumed to be $\frac{\s}{|\s|}$. 
Now we show that its second marginal is $\frac{\y}{|\y|}$ and that $\max_{(g_x,g_y)\leftarrow\phi_s}\met(g_x,g_y)\leq\alpha_1$. Note that these together imply that $\Winf{}(\frac{\s}{|\s|},\frac{\y}{|\y|})\leq\alpha_1$.

\begin{itemize}
\item {\it Second marginal of $\phi_s$ is $\frac{\y}{|\y|}$:}
We show it in two parts, first for $g_y\in\G$ for which $\z(g_y)=0$ and then for the rest of the $g_y\in\G$. Note that when $\z(g_y) = 0$, we have from \ref{Z4} that $\y(g_y)=0$.
Now we show that $\int_{\G}\phi_s(g_x,g_y)\dd g_x=0$.
It follows from \eqref{eq:drmv-phi} that for all $g_y$ such that $\z(g_y) = 0$, we have 
$\phi_s(g_x,g_y)=0, \forall g_x\in\G$, which implies that $\int_{\G}\phi_s(g_x,g_y)\dd g_x=0$.
Now we analyze the case when $\z(g_y) > 0$.
\begin{align*}
    \int_{\G} \phi_s(g_x, g_y)\dd g_x &= \int_{\G} \frac{1}{1 - \alpha_2} \phi_z(g_x, g_y)\frac{\y(g_y)}{\z(g_y)}\dd g_x \tag{using \Equationref{drmv-phi}}\\
    &= \frac{1}{1 - \alpha_2} \frac{\y(g_y)}{\z(g_y)} \int_{\G} \phi_z(g_x, g_y)\dd g_x\\
    &= \frac{1}{1 - \alpha_2} \frac{\y(g_y)}{\z(g_y)} \frac{\z(g_y)}{|\z|}\tag{using \ref{Z2}}\\
    &= \frac{\y(g_y)}{(1 - \alpha_2)|\z|}\\
    &= \frac{\y(g_y)}{|\y|}. \tag{Using \ref{Z3}}
\end{align*}
\item {\it $\Winf{}(\frac{\s}{|\s|},\frac{\y}{|\y|})\leq\alpha_1$:} 
We have shown that the first and the second marginals of $\phi_s$ are $\frac{\s}{|\s|}$ and $\frac{\y}{|\y|}$, respectively. So, it suffices to show that $\max_{(g_x,g_y)\leftarrow\phi_s}\met(g_x,g_y)\leq\alpha_1$.
Consider any pair $(g_x, g_y) \in \G^2$ s.t. $\phi_s(g_x, g_y) > 0$. This is possible only if $\phi_z(g_x, g_y) > 0$ (see \Equationref{drmv-phi}), which, when combined with \ref{Z2}, gives $\met(g_x, g_y) \le \alpha_1$. Hence, for any pair $(g_x, g_y) \in \G^2$ s.t. $\phi_s(g_x, g_y) > 0$, we have $\met(g_x, g_y) \le \alpha_1$. 
\end{itemize}

\item {\bf Proof of \ref{S3}.}
Note that \Equationref{drmv-phi} gives the normalized $\s$, but we still have the freedom to choose $|\s|$. To satisfy \ref{S3}, we set $|\s| = |\y|$. 

\item {\bf Proof of \ref{S1}.}
Note that \ref{S1} is already satisfied using \ref{Z1}, \ref{Z3}, and \ref{S3}. 

\item {\bf Proof of \ref{S2}.}
Let us denote $\{g \in \G\ |\ \z(g) > 0\}$ by $\G_z$. We will show that for any $g \in \G$, we have $\x(g) - \s(g) \ge 0$:
\begin{align*}
    \x(g) - \s(g) &= |\x|\int_{\G} \phi_z(g, g_y)\dd g_y - |\s|\int_{\G} \phi_s(g, g_y)\dd g_y \tag{using \ref{Z2} and \ref{S4}}\\
    &= |\x|\int_{\G_z} \phi_z(g, g_y)\dd g_y - |\s|\int_{\G} \phi_s(g, g_y)\dd g_y \tag{Since $\z(g_y) = 0 \implies \phi_z(g, g_y) = 0, \forall g\in\G$;~\ref{Z2}}\\
    &= |\x|\int_{\G_z} \phi_z(g, g_y)\dd g_y - |\s|\int_{\G_z} \frac{1}{1 - \alpha_2} \phi_z(g, g_y)\frac{\y(g_y)}{\z(g_y)}\dd g_y \tag{using \Equationref{drmv-phi}}\\
    &= |\x|\int_{\G_z} \phi_z(g, g_y)\dd g_y - \frac{(1 - \alpha_2)|\x|}{1 - \alpha_2}\int_{\G_z}  \phi_z(g, g_y)\frac{\y(g_y)}{\z(g_y)}\dd g_y \tag{using \ref{S1}}\\
    &= |\x|\int_{\G_z} \phi_z(g, g_y)\dd g_y - |\x|\int_{\G_z} \phi_z(g, g_y) \frac{\y(g_y)}{\z(g_y)}\dd g_y \\
    &= |\x|\int_{\G_z} \phi_z(g, g_y)\left(1 - \frac{\y(g_y)}{\z(g_y)}\right)\dd g_y\\
    &\ge 0 \tag{using \ref{Z4}, $\frac{\y(g_y)}{\z(g_y)} \le 1$}\\
\end{align*}
\end{itemize}
Thus, we have shown that the joint distribution $\phi_s$ defined in \eqref{eq:drmv-phi} satisfies all four properties \ref{S1}-\ref{S4}.
This completes the proof of \Lemmaref{drop-move-switch}.
\end{proof}
\begin{claim}\label{clm:drop-move-quasi-metric}
For all $\eta \in \Rplus$, $\dropmove\eta(\cdot,\cdot)$ is a quasi metric.
\end{claim}
\begin{proof}
Note that both $\drop$ and $\move$ are quasi-metrics. Hence, for any $\x, \y$, $\drop(\x, \y) \ge 0$ and $\move(\x, \y) \ge 0$. This implies that for every $\x, \y$, $\drme(\x, \y) \ge 0$. Now we one by one prove that $\drme$ satisfies the properties of quasi-metric:

\textit{Property \#1: For all $\x$ and $\y$, $\x = \y \Leftrightarrow \drme(\x, \y) = 0$.}
\begin{enumerate}
    \item For all $\x$, $\drme(\x, \x) = 0$:
    \begin{align*}
    \drme(\x, \x) &= \inf_{\z} \left( \drop(\x,\z) + \eta \cdot \move(\z,\x)\right) \\
    &\le \drop(\x,\x) + \eta \cdot \move(\x,\x) \tag{infimum over a set is $\le$ the value at any fixed point in set}\\
    &= 0
    \end{align*}
    Since $\drme(\x, \x) \ge 0$ as well as $\le 0$, $\drme(\x, \x) = 0$.
    \item For all $\x, \y$, $\drme(\x, \y) = 0 \implies \x = \y$:\\
    $\drme(\x, \y) = 0$ implies that $\inf_{\z} \left( \drop(\x,\z) + \eta \cdot \move(\z,\y) \right) = 0$. As both $\drop(\x,\z)$ and $\move(\z,\y)$ are $\ge 0$ for any value of $\x, \y, \z$, this is possible only if $\drop(\x,\z) = \move(\z,\y) = 0$ which means that $\x = \z = \y$. Hence $\x = \y$.
\end{enumerate}

\textit{Property \#2: For all $\x$, $\y$ and $\z$, $\drme(\x, \z) \le \drme(\x, \y) + \drme(\y, \z)$.}

We assume that the infimum in both $\drme(\x, \y)$ and $\drme(\y, \z)$ is achieved by $\s_1$ and $\s_2$, respectively
(the proof can be easily extended to the case when the infimum is not achieved).
This means that there exists $a,b,c,d\geq0$, such that
\[\drop(\x, \s_1) = a;\ \move(\s_1, \y) = b;\  \drop(\y, \s_2) = c;\  \move(\s_2, \z) = d,\]
which implies $\drme(\x, \y) = a + \eta b$ and $\drme(\y, \z) = c + \eta d$. 
We need to show that $\drme(\x, \z) \le (a + c) + \eta(b + d)$.

Using \Lemmaref{drop-move-switch} with $\move(\s_1, \y) = b$ and $\drop(\y, \s_2) = c$, we get that there is a $\y'$ such that $\drop(\s_1, \y') = c$ and $\move(\y', s_2) \leq b$. This gives the following:
\[\drop(\x, \s_1) = a;\ \drop(\s_1, \y') = c;\ \move(\y', \s_2) \leq b;\ \move(\s_2, \z) = d.\]
Now we prove that $\drme(\x, \z) \le (a + c) + \eta(b + d)$:
\begin{align*}
    \drme(\x, \z) &= \inf_{\z} \left( \drop(\x,\y) + \eta \cdot \move(\y,\z) \right) \\
    &\le \drop(\x,\y') + \eta \cdot \move(\y',\z)\\
    &\le \drop(\x,\s_1) + \drop(\s_1,\y') + \eta \cdot \move(\y',\z) \tag{$\drop$ is a quasi-metric}\\
    &\le \drop(\x,\s_1) + \drop(\s_1,\y') + \eta \cdot (\move(\y',\s_2) + \move(\s_2,\z)) \tag{$\move$ is a metric}\\
    &\le (a + c) + \eta (b + d).
\end{align*}
This concludes the proof of \Claimref{drop-move-quasi-metric}
\end{proof}

}

\end{document}